\documentclass[11pt,reqno,UTF8,twoside]{amsart}
\usepackage{mathrsfs}
\usepackage{amsfonts,amssymb,amsmath,amsthm}
\usepackage{cite}
\usepackage[colorlinks=true,citecolor=red]{hyperref}
\usepackage{titletoc}
\usepackage{geometry}
\usepackage{bm}
\usepackage{indentfirst}
\usepackage{graphicx}
\usepackage{float} 
\usepackage{booktabs}
\usepackage{longtable}
\usepackage{subfigure}
\usepackage{stmaryrd}
\usepackage{enumerate}
\usepackage{tikz}
\usepackage{ulem}
\usepackage{color,soul}
\usepackage{setspace}
\usepackage{stmaryrd}
\usepackage[pagewise]{lineno} 
\allowdisplaybreaks[2]
\usepackage{appendix}
\usepackage{cases}
\usepackage{todonotes}
\usepackage{enumitem}
\hypersetup{linkcolor=blue}

\makeatletter
\@namedef{subjclassname@2020}{\textup{2020} Mathematics Subject 
	Classification}
\makeatother
\usepackage{amssymb}
\usepackage[dvips]{epsfig}
\usepackage{graphicx}
\usepackage{color}
\usepackage{ulem}
\usepackage{tikz}
\usepackage{amsmath}
\usepackage{amssymb}
\usepackage{amsfonts}
\usepackage{amsthm}
\usepackage{tikz}
\usepackage{bm}
\usetikzlibrary{calc}
\usepackage{changes}

\newcommand{\R}{\mathbb R}

\newcommand{\Z}{\mathbb Z}

\newcommand{\C}{\mathbb C}

\usepackage{float}

\newcommand{\N}{\mathbb{N}}

\newcommand{\T}{\mathbb{T}}

\newtheorem{thm}{Theorem}[section]
\newtheorem{hp}{\bf Hypothesis}

\newtheorem{lem}[thm]{Lemma}
\newtheorem{prop}[thm]{Proposition}
\newtheorem{cor}[thm]{\bf Corollary}

\newtheorem{rem}[thm]{\bf Remark}

\theoremstyle{definition}
\newtheorem{defn}[thm]{Definition}

\newtheorem{claim}[thm]{\bf Claim}
\newtheorem{cd}[thm]{\bf Condition}
\theoremstyle{statement}

\numberwithin{equation}{section}

\usepackage{caption}

\usepackage{bm}
\begin{document}
	\title[Analytic quasiperiodic Schr\"odinger operators]{Anderson localization and H\"older continuity of  the  integrated density of states  for  analytic  quasiperiodic Schr\"odinger operators}
	\author[Cao]{Hongyi Cao}
	\address[H. Cao] {School of Mathematical Sciences,
		Peking University,
		Beijing 100871,
		China}
	\email{chyyy@math.pku.edu.cn}
	\author[Shi]{Yunfeng Shi}
	\address[Y. Shi] {School of Mathematics,
		Sichuan University,
		Chengdu 610064,
		China}
	\email{yunfengshi@scu.edu.cn}
	
	\author[Zhang]{Zhifei Zhang}
	\address[Z. Zhang] {School of Mathematical Sciences,
		Peking University,
		Beijing 100871,
		China}
	\email{zfzhang@math.pku.edu.cn}
	\date{\today}
	\maketitle
	\begin{abstract}
		We establish both Anderson localization and H\"older continuity of the integrated density of states for quasiperiodic Schr\"odinger operators on $\Z^d$ with any non-constant analytic potential and any Diophantine frequency in the perturbative regime. Our proof is based on a new method for controlling Green's functions and eliminating double resonances, in the spirit of multi-scale analysis. To the best of our knowledge, this is the first multi-scale analysis approach that works for fixed Diophantine  frequencies and potentials beyond the cosine type.
	\end{abstract} 
	\section{Introduction and main results} 
	We  study     quasiperiodic  Schr\"odinger operators     on $\Z^d$   \begin{equation}\label{model} 
 	H(\theta)=\varepsilon \Delta+v(\theta+ x\cdot{\omega})\delta_{x, y},\quad x,y\in\Z^d,
 \end{equation}
  where  $\varepsilon\geq0$ is the coupling constant, $\theta\in \T:=\R/\Z$ is  the phase,  $\omega\in (\R\setminus\mathbb{Q})^d$ is  the frequency,   $v:\T\to \R$ is a
  $1$-periodic  potential function,  $x\cdot\omega:=\sum_{i=1}^dx_i \omega_i $        and    $\Delta$ is the  discrete Laplacian on $\Z^d$  defined by 
$$
  	\Delta(x, y)=\delta_{{\|x- y\|_{1}, 1}},\quad \| x\|_{1}:=\sum_{i=1}^{d}\left|x_{i}\right|.
  $$ 
  
Ever since the fundamental works of Sinai \cite{Sin87} and Fr\"ohlich-Spencer-Wittwer \cite{FSW90}, which first independently proved Anderson localization (i.e., pure point spectrum with exponentially decaying eigenfunctions) for the operator \eqref{model} with $d=1$,  $C^2$ cosine-type $v$ and  Diophantine $\omega$ in the perturbative regime (requiring   $\omega$-dependent smallness of  $\varepsilon$), Anderson localization has been one of the central themes in the study of quasiperiodic Schr\"odinger operators. Subsequently, extending the results of \cite{Sin87, FSW90} to more general settings has become one of the main motivations for the community. In \cite{Eli97}, Eliasson made the first major contribution, establishing spectral localization (i.e., pure point spectrum) for the operator \eqref{model} with $d=1$,  transversal Gevrey $v$ and  Diophantine $\omega$ in the perturbative regime; however, his proof did not provide the exponential decay of eigenfunctions. In \cite{Jit94, Jit99}, Jitomirskaya achieved a breakthrough, introducing the nonperturbative approach and proving the localization transition for the almost Mathieu operator (i.e., the operator \eqref{model} with $d=1$ and $v(\theta)=\cos(2\pi\theta)$) with  Diophantine $\omega$. The nonperturbative approach of \cite{Jit94,Jit99} was later significantly extended in the seminal work of Bourgain-Goldstein \cite{BG00} to multi-frequency analytic $v$, which was further developed for $d\geq 2$ in the perturbative regime by the breakthroughs of Bourgain-Goldstein-Schlag \cite{BGS02} and Bourgain \cite{Bou07} (see also \cite{BK19,JLS20} for recent progress). 
Indeed, the Anderson localization results of \cite{BG00,BGS02,Bou07,BK19,JLS20} hold for any fixed $\theta$ and most $\omega$ (the permitted set of $\omega$ depends on this given $\theta$). To the best of our knowledge, however, all known Anderson localization results with fixed Diophantine $\omega$ are restricted to special classes of potentials such as cosine type; \textbf{it remains largely open whether Anderson localization can hold for the operator  \eqref{model} with general analytic $v$ and fixed Diophantine $\omega$}.

 In this paper, we prove Anderson localization for the operator \eqref{model} with \textbf{any $d\geq 1$, any non-constant analytic $v$ and any Diophantine $\omega$, provided that $\varepsilon$ is sufficiently small}, which addresses this long-standing open problem in the perturbative regime. Another important topic in the study of quasiperiodic Schr\"odinger operators is the H\"older continuity of the integrated density of states (IDS). Under the same assumptions, we also establish a new H\"older continuity result for the IDS of such analytic quasiperiodic Schr\"odinger operators, obtaining a quantitative H\"older exponent at energy $E^*$ related to the number of roots of $v(\theta)-E^*$.

   \subsection{Main results} Let $H(\theta)$ be the operator given by \eqref{model}, and let the frequency $\omega$ and the potential function $v$ satisfy the following assumptions:
   
   \underline{Assumption on $\omega$:}  We   assume that $\omega$ satisfies a   {\bf  Diophantine condition}   (denoted by $\omega\in 	{\rm DC}_{\tau, \gamma}$ for some $\tau>d,\gamma>0$), where 
   \begin{equation*}
   	{\rm DC}_{\tau, \gamma}=\left\{\omega \in (\R\setminus\mathbb{Q})^d:\ \|x\cdot\omega\|_\T:=\inf_{l\in\mathbb{Z}}|l- x\cdot\omega|\geq \frac{\gamma}{\| x\|_1^{\tau}},\ \forall\  x\in \Z^d\setminus\{o\}\right\}
   \end{equation*} 
  and  $o$ is  the origin of $\Z^d$. 
   
   \underline{Assumptions on $v$:} We  assume that  $v$ is a {\bf  real analytic} function on $\T$, and  for some $\eta>0, C_v>0$, it   has  an analytic extension in $\T_{2\eta}:=\{z\in \C/\Z :\ |\Im z|\leq 2\eta\} $ with 
   \begin{equation}\label{tiao1}
   \sup_{\theta\in  \T_{\eta}}|v(\theta)|\leq C_v. 
   \end{equation}
   In addition, we impose the following Condition \ref{con1} or Condition \ref{con2} on  $v$. It is easy to prove  that Condition \ref{con1} is fulfilled for any non-constant analytic function and Condition \ref{con2} is fulfilled for any analytic function with $1$ as its shortest period. We include the proofs in the Appendix  for the reader's convenience. 
   
   \begin{cd}[Root]\label{con1}
   	There exists some integer  $M\in [2,+\infty)$  such that for any $E^*\in \{v(\theta):\  \theta\in \T_{\eta/4} \} $, $v(\theta)-E^*$ has no more than $M
   	$ roots (counted with multiplicity) in $\T_{\eta/2}$. We denote  these roots by  $\theta_0^{(1)}(E^*),\cdots,\theta_0^{(m_0)}(E^*)$, where \begin{equation}\label{ges}
   		\text{	$m_0=m_0(E^*)\in [1,M]$  is the number of roots of $v(\theta)-E^*$ in $\theta\in \T_{\eta/2}$.}
   	\end{equation}  Moreover,  there exists some $\widetilde{C}_v\geq 1$ such that for any $E^*\in \{v(\theta):\  \theta\in \T_{\eta/4} \} $,  
   	\begin{equation*}
   		\widetilde{C}_v^{-1}\prod_{i=1}^{m_0}\|\theta- \theta_0^{(i)}(E^*)\|_\T \leq |v(\theta)-E^*|\leq \widetilde{C}_v\prod_{i=1}^{m_0}\|\theta- \theta_0^{(i)}(E^*)\|_\T, \quad  \forall    \theta\in \T_{\eta/4}. 
   	\end{equation*}
   \end{cd}
   \begin{cd}[Transversality]\label{con2}
   	There exist some $s\in \N$ and $c>0$  such that 
   	$$\max _{0 \leq l \leq s}|\partial_\theta^l(v(\theta+\phi)-v(\theta))| \geq c \|\phi\|_\T,  \quad \forall \theta , \phi\in \T .$$
   \end{cd}
   Under the above  assumptions,  we prove   Anderson localization (for Lebesgue a.e. $\theta$) and   H\"older continuity of the IDS for $H(\theta)$ given by \eqref{model}, provided that $\varepsilon$ is sufficiently small. 
 \begin{thm}\label{maina}  
 	Assume that $\omega\in 	{\rm DC}_{\tau, \gamma}$ and the  real analytic function $v$ satisfies  \eqref{tiao1},  Conditions \ref{con1} and  \ref{con2}.	Then there exists some  $\varepsilon_0=\varepsilon_0(\eta,C_v,M, \widetilde{C}_v, c, s  ,d,\tau,\gamma)>0$ such that for all $0\leq \varepsilon\leq \varepsilon_0$, $H(\theta)$ satisfies  Anderson localization  for Lebesgue a.e. $\theta\in \T$.	
 \end{thm}
 \begin{rem}
 	If $v$ does not satisfy Condition \ref{con2} but  has a shorter positive period,  we can scale the period 	to become $1$ and then apply  Theorem \ref{maina}. In particular, the theorem applies to any  non-constant	real analytic $v$.  We thus obtain  the following corollary.
 \end{rem}
 \begin{cor}
 	Assume that $\omega\in 	{\rm DC}_{\tau, \gamma}$ and $v$ is a non-constant real  analytic function on $\T$.	Then there exists some $\varepsilon_0=\varepsilon_0(v,d,\tau,\gamma)>0$ such that for all $0\leq \varepsilon\leq \varepsilon_0$, $H(\theta)$ satisfies Anderson localization for Lebesgue a.e. $\theta\in \T$.	
 \end{cor}  
\begin{rem}
	Our analysis yields dynamical localization as well (see \cite[\S 4.2]{CSZ24} for the argument).  
\end{rem}
 
 To state our second main result, we first briefly  recall   the definition of the IDS  $\mathcal{N}(E)$. Let $Q_N:=\{x\in \Z^d:\ \|x\|_1\leq N\}.$ Denote by ${\rm Spec(\cdot)}$ the spectrum of an operator. For $\Lambda\subset\Z^d$ and an  operator $H$ on $\Z^d$,  we let $H_\Lambda=R_\Lambda H R_\Lambda$ be the Dirichlet restriction of $H$ to $\Lambda$. We   define 
 \begin{align}
 	\label{count}	\mathcal{N}_{Q_N}(E , \theta)&:=\frac{1}{|Q_N|} |\left\{\lambda \in \operatorname{Spec}\left(H_{Q_N}(\theta)\right): \lambda \leq E\right\}|,\\
 	\mathcal{N}(E)&:=\lim _{N\rightarrow +\infty} \mathcal{N}_{Q_N}(E ,\theta), \quad \forall \theta\in \T,\label{ids}
 \end{align}
 where  the eigenvalues are  counted with multiplicity in \eqref{count}.   It is well known that the limit in \eqref{ids} exists and is independent of $\theta$ since $H(\theta)$ is unique ergodic and continuous.  
 \begin{thm}\label{mainb}
 	Assume that $\omega\in 	{\rm DC}_{\tau, \gamma}$ and the  real analytic function  $v$ satisfies  \eqref{tiao1}  and Condition \ref{con1}.	Then there exists some $\varepsilon_0=\varepsilon_0(\eta,C_v,M, \widetilde{C}_v,  d,\tau,\gamma)>0$ such that for all $0\leq \varepsilon\leq \varepsilon_0$, the IDS $\mathcal{N}(E)$ satisfies, for any  $\mu>0$, 
 	$$\mathcal{N}(E^*+\beta)-\mathcal{N}(E^*-\beta)\leq \beta^{\frac1{m_0(E^*)} -\mu} $$
 	for all $E^*\in \R$ and all $0<\beta<\beta_0=\beta_0(\widetilde{C}_v , m_0(E^*), \mu) $, where $m_0(E^*)$ is defined in \eqref{ges}.   In particular, for any such $\varepsilon$ and any $\mu>0$,   $\mathcal{N}(E)$ is $(\frac{1}{M_0}-\mu )$-H\"older continuous on any interval $I\subset\R$ such that   $m_0(E^*)\leq M_0$ holds for all $E^*\in I$. 
 \end{thm}
 From Theorem \ref{mainb},  we immediately obtain the following corollary.  
 \begin{cor}\label{coor}
 	Assume that $\omega\in 	{\rm DC}_{\tau, \gamma}$. Let  $v_0(\theta)=\sum_{k=-m}^{m} \widehat{v}_0(k) e^{2\pi i k\theta}$ be  a  non-constant trigonometric polynomial with  $\widehat{v}_0(k)=\overline{\widehat{v}_0(k)}$ for $-m \leq k \leq m$ and $u(\theta)$ be a real analytic function on $\T$.   
 	Consider the operator \begin{align}\label{model2}
 		\widetilde{H}(\theta)=\varepsilon_1 \Delta+\left(v_0(\theta+ x\cdot{\omega})+\varepsilon_2u(\theta+ x\cdot{\omega})\right)\delta_{x, y},\quad  x,y\in\Z^d.
 	\end{align}
 	Then there exists some $\varepsilon_0=\varepsilon_0(v_0,u,  d,\tau,\gamma)>0$ such that for all $0\leq \varepsilon_1,  \varepsilon_2\leq \varepsilon_0$, the IDS  $\widetilde{\mathcal N}(E)$ of the operator  \eqref{model2}  satisfies, for any  $\mu>0$, 
 	$$\widetilde{\mathcal N}(E^*+\beta)-\widetilde{\mathcal N}(E^*-\beta)\leq \beta^{\frac1{2m}  -\mu} $$
 	for all $E^*\in \R$ and all $0<\beta<\beta_0=\beta_0(v_0,\mu)$.  In particular, for any such $\varepsilon_1,\varepsilon_2$ and any $\mu>0$, $\widetilde{\mathcal N}(E)$ is $(\frac{1}{2m}-\mu)$-H\"older continuous. 
 \end{cor}
 
 \begin{rem}
 	Our approach  also   applies to  band operators on  $\Z^d\times\{1,\cdots,b\}$ (see  \cite{BJ00} for the definition)    and    exponentially decaying  long-range operators, namely,  the   discrete Laplacian $\Delta$ in \eqref{model} replaced    by a Toeplitz operator $S$ satisfying, for some $\rho>0$, 
 	$$|S(x,y)|\leq e^{-\rho\|x-y\|_1}.$$
 \end{rem}

 \subsection{Related work} 
Let us review  some of the results in the literature related to
this paper: 
 \subsubsection{Anderson localization}
We start  with the important  almost Mathieu operator (AMO), which has attracted considerable interest   in the study of quasiperiodic  Schr\"odinger  operators.   
The perturbative proofs of Anderson localization for AMO with  Diophantine frequencies $\omega$  are due to Sinai \cite{Sin87} and  Fr\"ohlich-Spencer-Wittwer \cite{FSW90}; both approaches use intricate multi-step iterations and rely  heavily on eigenvalue and eigenfunction parametrization and variation arguments, which  are robust in the sense that they apply to more general  $C^2$ cosine-type potentials $v$.   	 
 The common  feature of  the perturbative results   is that the smallness of the coupling constant $\varepsilon$ depends on  $\omega$ at least through the
 constants in the Diophantine characterization of $\omega$. 
In \cite{Jit94,Jit99}, Jitomirskaya established    the nonperturbative result  and localization transition for    AMO with Diophantine $\omega$  in the sense that Anderson localization holds in the positive Lyapunov exponent regime with $\omega$-independent  estimate on   $\varepsilon$, which  was  later   extended  to long-range operators by  Bourgain-Jitomirskaya \cite{BJ02};      more importantly, the nonperturbative approach she introduced  shed  new light on future studies.  
In  \cite{AYZ17},  Avila-You-Zhou established    the   sharp  localization transition  for  AMO         via   reducibility and  Aubry duality methods.
In \cite{JL18,JL24}, Jitomirskaya-Liu   further  developed the nonperturbative approach to obtain  arithmetic localization   (i.e., explicit arithmetic descriptions on  $\omega$ and  phases $\theta$ for which  Anderson localization holds). 
  A  natural analytic universality class  to consider for  similar results of  AMO    is  operators of   Type I (see \cite{GJ24} for the definition).  In \cite{HS26}, Han-Schlag established  Anderson localization for   even Type I operators with Diophantine $\omega$ (see also  \cite{Han24} for some recent progress).   It turns out that the evenness assumption on  $v$ in \cite{FSW90,HS26} plays a crucial role in eliminating  double resonances and it is significantly more difficult to obtain Anderson  localization for non-even $v$;  
for contributions in this direction, we refer to the works of  Forman-VandenBoom \cite{FV25} (for $C^2$ cosine-type quasiperiodic Schr\"odinger  operators)  and Ge-Jitomirskaya \cite{GJ24} (for Type I operators).
 We mention that \cite{GJ24}  also  established the  arithmetic localization transition for Type I operators. 
	   Note that  all the   results we mentioned  above are  restricted to operators defined on the   one-dimensional lattice $\Z$;  there has also been localization results for 	    multi-dimensional AMO and $C^2$  cosine-type quasiperiodic  operators, cf.  \cite{CD93, Sur96,JK16,GY20,GYZ23,CSZ22,CSZ23,CSZ24} just for a few.     We remark that all these  multi-dimensional results are perturbative   and  in fact it is  shown in \cite{Bou02} that the nonperturbative localization is typically not expected for dimension higher than two. 

The first  localization result for analytic quasiperiodic  operators beyond special classes (such as the aforementioned   cosine type and Type I) is due to  Eliasson \cite{Eli97}, where he proved spectral localization for  one-dimensional  non-constant  analytic quasiperiodic operators  (actually,  more general transversal Gevrey ones) with   Diophantine $\omega$  in the perturbative regime via KAM methods.  Eliasson has also suggested that the argument could be  modified to obtain exponential decay of the eigenfunctions, but he has not provided a proof of it (cf.  the discussions in \S 7 of \cite{Eli97}). 
 In \cite{BG00}, Bourgain-Goldstein established the nonperturbative Anderson  localization  for one-dimensional analytic quasiperiodic Schr\"odinger  operators 
by combining the nonperturbative approach with large deviation estimate  and  semi-algebraic  elimination methods, which was significantly extended to multiple dimensions  (which can mean  operators defined on $\Z^d$, as well as multi-frequency  $v$ defined on $\T^d$)  in the perturbative regime  by   Bourgain-Goldstein-Schlag \cite{BGS02}  and  Bourgain \cite{Bou07} (see also  \cite{BK19,JLS20} for recent progress)   by  further incorporating the perturbative  multi-scale analysis approach.  
   It turns out that the Anderson localization  results  of \cite{BG00,BGS02,Bou07,BK19,JLS20} hold  for any fixed  $\theta$  and   most $\omega$ (the permitted set of $\omega$ depends on  the given $\theta$) 
since  in their proofs,  $\omega$  serves  as the ``random parameter''  when   eliminating double resonances;  \textbf{it is even unknown if one can fix some Diophantine $\omega$ to prove Anderson localization for a.e. $\theta$ for the operator  \eqref{model} with non-constant analytic $v$}.   For instance,   in the nonperturbative regime,  Jitomirskaya \cite{Jit23} mentioned  that  ``so far even an arithmetic version of localization for the Diophantine case has not been established for the general analytic family, the current state-of-the-art result by Bourgain-Goldstein  (2000) being measure theoretic in $\alpha$''.     It  was also  mentioned by Bourgain \cite{Bou00}  that ``It indeed turns out to be easier to exploit a certain randomness in the frequency  than in the phase.'' 	   Here we provide an affirmative answer to this problem in the perturbative regime.    
Finally, we would like to mention that   shortly after the posting of this paper, Wang-You-Zhou \cite{WYZ26} also established  Anderson localization  for the operator \eqref{model} with  $v$ given by a non-constant  trigonometric polynomial, using  dynamical methods combined  with the spectral localization result from  \cite{Eli97}.

\subsubsection{H\"older continuity of the  IDS} 
 It is shown  in \cite{CS831, CS832, BK13} that the IDS  is at least log-H\"older continuous for any ergodic Schr\"odinger operator; 
 proving the  modulus of   continuity  of the IDS beyond log-H\"older ones (such as the  H\"older continuity) has also attracted  a lot of interest  in the study of quasiperiodic  Schr\"odinger  operators.     
   In \cite{GS01}, Goldstein-Schlag  proved the H\"older continuity   for  one-dimensional  analytic quasiperiodic  Schr\"odinger operators with   Diophantine $\omega$  in the nonperturbative regime   via an approach based on   large deviation estimates  for the Lyapunov exponents; moreover, they obtained the weak  H\"older continuity\footnote{The IDS $\mathcal N(E)$ satisfies, for  some $\zeta\in(0,1)$,    $$|\mathcal{N}(E)-\mathcal{N}(E')|\leq  e^{-\left(\ln\frac{1}{|E-E'|}\right)^\zeta} \text{ if   $|E-E'|$ is  sufficiently small.}$$ } of the IDS for those with  multi-frequency analytic $v$.  
Later in \cite{Bou00}, Bourgain introduced a  method based on Green's function estimates  to obtain the  $(\frac12-)$-H\"older continuity \footnote{The IDS $\mathcal N(E)$ satisfies, for any  $\mu >0$,   $$|\mathcal{N}(E)-\mathcal{N}(E')|\leq |E-E'|^{\frac{1}{2}-\mu} \text{ if   $|E-E'|$ is sufficiently small.} $$} of the IDS   for     AMO  with  Diophantine $\omega$  in the perturbative regime; indeed, some parts of our argument, like the idea of  fixed-energy Green's function estimates explained in \S \ref{Bo}, echo some of the arguments developed therein.  In  \cite{Amo09}, Amor   obtained  the  $\frac12$-H\"older continuity of the IDS    for  one-dimensional analytic  quasiperiodic  Schr\"odinger operators (actually,  more general multi-frequency ones) with  Diophantine $\omega$  in the perturbative regime (small $v$)  via  the reducibility method.  Amor’s result in the one-frequency case was extended  by Avila-Jitomirskaya \cite{AJ10} to the nonperturbative regime  via    reducibility and localization approaches.   In \cite{GS08},   Goldstein-Schlag  proved the  nonperturbative $(\frac{1}{2m}-)$-H\"older continuity of the IDS for one-dimensional  quasiperiodic  Schr\"odinger operators with $v$  given by an  analytic perturbation of a  certain trigonometric polynomial of degree $m\geq1$ and Diophantine $\omega$. 
For one-dimensional analytic quasiperiodic  Schr\"odinger operators,  You \cite{You18} proposed 	 a conjecture  that for Diophantine $\omega$ and in the positive Lyapunov exponent regime,    the H\"older exponent of the   IDS is at least  $\frac{1}{2\omega(E)}$, where $\omega(E)$ is the  acceleration of the Schr\"odinger  cocycle associated with $E$.   
This conjecture was partially solved by Ge-Jitomirskaya \cite{GJ24}  for Type I operators and by 	 Han-Schlag \cite{HS26}  with  an arbitrarily small loss on  the H\"older exponent. 

Finally, we mention some multi-dimensional results (again, only perturbative): 
The  weak H\"older continuity of the  IDS for multi-dimensional  analytic quasiperiodic Schr\"odinger operators with  $\omega$ outside  a small-measure set   has  been established in  \cite{Sch01,Bou07,Liu20}   via  Green's function estimates;       the  $(\frac{1}{2m}-)$-H\"older continuity of the IDS   for  multi-dimensional  quasiperiodic  Schr\"odinger operators  with $v$  given by a trigonometric polynomial of degree $m\geq1$ and  Diophantine $\omega $  has been  established in  \cite{GYZ22}  via the reducibility method;   the $(\frac12-)$-H\"older continuity of the IDS for  multi-dimensional $C^2$ cosine-type  quasiperiodic Schr\"odinger operators with  Diophantine $\omega $   has  also been proved by the authors \cite{CSZ24}  via  Green's function estimates. Here we establish a new  H\"older continuity  result for  multi-dimensional  analytic quasiperiodic Schr\"odinger operators  with   Diophantine $\omega $.

	\subsection{Main ideas and new ingredients of the proof}\label{new}
Our  proof is  based on  Green's function estimates  via the multi-scale analysis approach  originating from \cite{FS83} in the perturbative regime.  The new ingredients are: 
\begin{itemize}
	\item[\textbf{(1)}] We extend the Schur complement argument in \cite{Bou00} for the $\theta$-variable {\bf  beyond cosine potential} to control Green's functions for any fixed energy $E^*$. 
	\item[\textbf{(2)}] We obtain Rellich functions in the spirit of \cite{FV25}   via   a different method based on  {\bf Schur complement and Rouch\'e arguments for the  $E$-variable}. By Rellich function analysis, we  control Green's functions for  varying energies $E$ close to $E^*$. 	As such,  we can control   locally in the  phase-energy space the Green's functions via some polynomials in both phase and energy parameters.  Moreover,  we  establish quantitative transversality estimates for  the resultant of  Rellich functions  via a {\bf  new Weierstrass preparation argument}.
	\item[\textbf{(3)}] We identify double resonant phases as those where the resultant takes small values and {\bf  eliminate double resonances via the transversality of the resultant}.  
\end{itemize}
Although the general scheme of our proof seems expected,
almost all of the  usual  elements are not available and require new ideas or a more complicated setup. We explain the difficulties and our methods to overcoming  them below.

	\subsubsection{Extension of \cite{Bou00}'s argument beyond cosine potential}\label{Bo} In \cite{Bou00}, Bourgain proved the  $(\frac12-)$-H\"older continuity of the IDS   perturbatively  for   AMO by using the   Schur complement argument for the $\theta$-variable to control Green's functions for any fixed energy $E^*$.
The key argument of \cite{Bou00} is to locally analyze the $\theta$-roots of   $\operatorname{det}(H_{B_n}(\theta)-E^*)$ in some  $n$-step phase neighborhood  for any fixed energy $E^*$, where $B_n\subset\Z$ is a proper  $n$-step  interval nearly centered at the origin  with $\operatorname{Diameter}(B_n)\sim l_n$ satisfying   $\ln\ln l_n\sim n$.  Resonances at  step  $n$ depend on the pair of roots $\{{\theta_{n}^{(1)}(E^*),\theta_{n}^{(2)}(E^*)}\}$ (there are at most two such roots at each step  due  to the  cosine potential) and the Green's function
 $$G_{\Lambda}^{\theta,E^*}:= (H_\Lambda(\theta)-E^*)^{-1}$$ for a properly saturated set $\Lambda\subset\Z$ can be controlled as 
$$\|G_{\Lambda}^{\theta,E^*}\|\leq \delta_{n-1}^{-5} \max_{x\in \Lambda}\left(\|\theta+x\cdot\omega-\theta_n^{(1)}(E^*)\|_\T^{-1}\cdot  \|\theta+x\cdot\omega-\theta_n^{(2)}(E^*)\|_\T^{-1}\right),$$ provided 
$$\min_{x\in \Lambda}\left(\|\theta+x\cdot\omega-\theta_n^{(1)}(E^*)\|_\T,  \|\theta+x\cdot\omega-\theta_n^{(2)}(E^*)\|_\T\right) \geq \delta_n,$$ 
 where $ \ln \delta_n^{-1}=l_n^{2/3}$.     
 
When dealing with general analytic potentials, there arises a significant difficulty: the  proper choice  of the  scale parameter $l_n$.   To carry out  the Schur complement argument, we require two key structural properties of $B_n$ (cf. Figure \ref{F3}):  \textbf{(1)} The $(n-1)$-resonant points in $B_n$  (a point $x\in \Z^d$ is called $(n-1)$-resonant if $\theta +x\cdot\omega$ is close to one of the roots $\{\theta_{n-1}^{(i)}(E^*)\}_{i=1}^{m_{n-1}}$)  lie in the core of $B_n$, forming a nonresonant annulus; \textbf{(2)} Aside from  at most $M$ resonant points in the core (i.e., resonant within a distance $\sim\widehat{\delta}_{n-1}$, where $ (\ln\delta_{n-1}^{-1})^{1/2} \ll \ln\widehat{\delta}_{n-1}^{-1}\ll \ln\delta_{n-1}^{-1}  $ will be specified during the scale parameter selection process), all other points in some  superlinear extension of $B_n$ are separated from these roots by a distance  significantly larger than $\widehat{\delta}_{n-1}^{1/C}$ for some $C\gg1$.   We call this an {\bf   annulus structure}. 
The cosine potential provides  significant simplification in the analysis of \cite{Bou00}: Bourgain separated the construction of such annulus structure  into two different cases---the simple resonant case and the  double resonant case---based on  the  positions  of the two  roots $\theta_{n-1}^{(1)}(E^*),\theta_{n-1}^{(2)}(E^*)$ and carried out the resonance analysis for both cases.     However,  it  is highly nontrivial   to extend this case analysis argument to settings with more than two roots. 

In this paper, {\bf  we introduce a significantly simplified method based on the  pigeonhole principle  for the construction of  this  annulus structure      that avoids complicated case distinctions and treats all cases simultaneously}. 
Suppose that  the first $(n-1)$ steps are done and we  obtain a set of $(n-1)$-step roots $\{\theta_{n-1}^{(i)}(E^*)\}_{i=1}^{m_{n-1}}$;   we then  analyze the shift distances  	$\|\theta_{n-1}^{(i)}(E^*)+x\cdot\omega-\theta_{n-1}^{(i')}(E^*)\|_\T$ for $1\leq i,i'\leq m_{n-1}$ and $\|x\|_1 \leq l_n^8$. Relying on the separation of resonant points guaranteed by $\omega\in 	{\rm DC}_{\tau, \gamma}$,  we  use a  pigeonhole principle argument  to  select   proper  intermediate scale and resonance parameters  that satisfy  the  annulus structure (cf. Propositions \ref{scalen}).  This  selection  process also induces an equivalence relation on  $\{\theta_{n-1}^{(i)}(E^*)\}_{i=1}^{m_{n-1}}$,  enabling  us to partition this set into several equivalence classes  (cf. Propositions \ref{scalen} \textbf{(4)}). We then assign a neighborhood to each representative    and carry out  resonance analysis in these neighborhoods   to obtain the next-step roots  $\{\theta_{n}^{(i)}(E^*)\}_{i=1}^{m_{n}}$, using   the Schur complement argument and the  favorable   annulus structure.      
      Moreover, we bound the number of roots at each step  by $M$  (cf.  \eqref{geshu}), which  enables the entire induction to proceed. With these roots,  $\|G_\Lambda^{\theta,E^*}\|$ can be controlled by a polynomial of $\theta$  (cf. \eqref{123321} and Proposition \ref{1gon} \textbf{(1)})  via standard Schur complement and resolvent identity arguments.    In fact,   the fixed-energy Green's function estimates  together with  Lemma \ref{HOl} are  sufficient to   prove    H\"older continuity   of the IDS; however,  alone they are  insufficient to identify a full-measure set of phases for which Anderson localization holds, 
   except in the special cosine potential case, where the method of \cite{Bou00} was recently   extended by the authors \cite{CSZ22} to prove the arithmetic \footnote{Anderson localization holds for phases  in a set of full measure with explicitly 	arithmetic description: $$   \bigcup_{N\geq 1}  \left\{\theta\in \T:\  \|2\theta+x\cdot\omega \|_\T> \|x\|_1^{-(d+1)} \text{ for all $x\in \Z^d$ such that $\|x\|_1\geq N$}\right\} .$$}  Anderson localization   perturbatively  for  multi-dimensional AMO by using  the symmetry of the two roots to remove the bad sets of double resonant phases uniformly in the energy.
  However, such a symmetry structure no longer holds beyond even cosine-type  potentials. This leads to the second key ingredient to Anderson localization: Green's function estimates for varying energies.
\subsubsection{Weierstrass preparation argument for the  $E$-variable}
Indeed,  the Green's function estimates in  \cite{FSW90,Bou00,CSZ22} were obtained by fixing  the energy $E^*$ and varying the phase $\theta$. However, in  \cite{FV25}, Forman and VandenBoom used the  Rellich function method to control Green's functions for  varying  energies and handle the resonances  depending on energy,  thereby establishing  Anderson localization perturbatively for one-dimensional  non-even  $C^2$ cosine-type quasiperiodic Schr\"odinger  operators,  removing  the crucial  evenness requirement  on  the potential from  \cite{FSW90}.  The key argument of \cite{FV25} is to locally obtain, for any $E^*$,  (at most two) piecewise-defined    Rellich functions $E_n^{(1)}(\theta)$ and  $E_n^{(2)}(\theta)$---the $E^*$-relevant eigenvalue curves of $H_{B_n}(\theta)$---so that for any $E$ close to $E^*$, the  Green's function $G_{\Lambda}^{\theta,E}$ for a properly saturated set $\Lambda\subset\Z$ can be controlled as
$$\|G_{\Lambda}^{\theta,E}\|\leq 10  \max_{x\in \Lambda}(  | E_n^{(1)}(\theta+x\cdot\omega)-E|^{-1}, | E_n^{(2)}(\theta+x\cdot\omega)-E|^{-1}),$$  provided 
$$ |E-E^*|\sim\delta_{n-1} \ \  \text{and}    \ \  \min_{x\in \Lambda}(|E_n^{(1)}(\theta+x\cdot\omega)-E|, |E_n^{(2)}(\theta+x\cdot\omega)-E|)  \geq \delta_n .$$  
  Recently, this method was further extended  by the authors to the  multi-dimensional case \cite{CSZ24} and to monotone quasiperiodic Schr\"odinger operators \cite{CSZ25}. 
  
  In  \cite{FV25} (the $C^2$ cosine-type setting), the authors obtained  Rellich functions via an eigenfunction perturbation  argument and established the transversality (Morse property) for each single branch of  the  Rellich function via the  eigenvalue  variation  (Hellmann-Feynman) formula, combined with a new Cauchy interlacing argument. (See also \cite{CSZ24} for the multi-dimensional case.)   However,  it  is highly challenging  to extend the eigenvalue  variation and Cauchy  interlacing arguments to higher derivatives or to cases with more than two relevant curves.
  
    In this paper, {\bf  we introduce a new method based on    Schur complement and    Weierstrass preparation arguments   for the $E$-variable  to obtain and  analyze a proper cluster of no more than $M$ many $E^*$-relevant  Rellich functions}   (cf. Proposition \ref{sen}),  so that for any $E$ close to $E^*$,  $\|G_\Lambda^{\theta,E}\|$ can be controlled by a polynomial of  $E$ (cf. \eqref{456654} and Proposition \ref{1gon} \textbf{(2)}). Moreover,   we prove that  the derivatives of the elementary symmetric polynomials of   the cluster of Rellich functions    differ little from those of the previous step (cf. Proposition \ref{coen}). This,   together with a technique of establishing  transversality estimates for the product of functions from  \cite{Eli97} (cf. Lemma \ref{eli})  enables us to prove that    the resultant of   Rellich functions  inherits   some quantitative transversality estimates from  $v$ (cf. Condition \ref{con2})  at each inductive step (cf. Proposition \ref{h4n}).  Such transversality  is significantly crucial for the elimination of  double resonances (cf. Proposition \ref{0ce}),  the final ingredient for    Anderson localization.    
\subsubsection{Elimination of double resonances via the  transversality of the resultant}
To prove Anderson localization, we need to remove a zero-measure set of phases where double resonance happens infinitely often, which requires a full understanding of the bad sets of double resonant phases at each inductive step. 

To this end, we  discretize the energy dependence of  resonances  by inductively dividing  the energy interval  into small subintervals so that the bad sets of phases depend on these subintervals for which they are resonant. 
In each subinterval, {\bf we identify the set of bad phases as approximate roots of the resultant of Rellich functions at which the resultant takes small values (cf. Lemma  \ref{baohan}) and carefully  use the transversality of the resultant to bound  its measure}, which is significantly  smaller than the length of the subinterval (cf. Lemma \ref{transha}).  Summing up the bad phases from all subintervals, we  obtain  a uniform bad set of phases at each step, whose limit superior has zero measure by the Borel-Cantelli lemma. We eventually remove this zero-measure set to prove Anderson localization via  a standard Schnol argument.

\vspace*{0.15cm}
 We hope that our  approach could shed  new light on  the study of the  fixed-frequency  type  Anderson localization  for more general QP Schr\"odinger operators via exploiting randomness in  phase parameters rather than  that  in frequencies.
\subsection{Organization of the paper}
	The paper is organized as follows. In \S \ref{MSA}, the core of this paper, we present the multi-scale analysis induction, aimed at substantiating  the inductive Hypotheses stated in \S \ref{HP}.   We split the entire induction into the initial step (cf. \S \ref{n=0}), the first inductive step (cf. \S \ref{n=1})  and the general inductive step (cf. \S \ref{n=n}). The  key arguments in  \S \ref{n=1} and \S \ref{n=n} are  similar,  including  choosing  a proper scale parameter for the annulus structure (cf.  \S \ref{C1},  \S \ref{CN}), carrying out  resonance analysis for the Schur complement (cf.  \S \ref{R1},  \S \ref{RN}), establishing Green's function estimates for nonresonant sets  (cf.  \S \ref{G1},  \S \ref{GN}), and establishing quantitative transversality estimates for the resultant   of Rellich functions (cf.  \S \ref{trans},  \S \ref{TN}).  We expect the analysis in  \S \ref{n=1} to be more accessible to readers, as it involves fewer technicalities. In \S  \ref{AL}, we prove   Anderson localization, and in \S \ref{IDS}, we establish  H\"older continuity of the IDS by applying the analysis developed in  \S \ref{MSA}. The verification of Conditions  \ref{con1} and  \ref{con2}, as well as some results on the Schur complement, are provided in the appendix.
	\subsection{Some notation in the paper}\label{notation}
	\begin{itemize}
		\item    	For $N>0$, we denote 
		$$Q_N:=\{x\in \Z^d:\ \|x\|_1\leq N\}.$$
		\item For $a\in  \C,R>0$, we denote 
		$$D(a,R):=\{z\in \C :\  |z-a|<R\}$$
		 and 
		$$\|a\|_\T:= \inf_{l\in\mathbb{Z}}|l-a|, \quad D_\T(a,R):=\{z\in \C/\Z :\  \|z-a\|_\T<R\}.$$
			\item For any $\eta\geq 0$, we denote 
		$$\T_\eta:=\{z\in \C/\Z :\ |\Im z|\leq \eta\}.$$
		\item  For two sets $X,Y\subset\Z^d$, we  denote $$\operatorname{dist}_1(X,Y)=\min_{\substack{x\in X, y\in Y} }\|x-y\|_1,$$
		$$\partial_Y^-X=\{x\in X:\ \text{there exists }y\in Y\setminus X\text{ with }\|y-x\|_1=1\},$$
		$$\partial_Y^+X=\{y\in Y\setminus X:\ \text{there exists }x\in X\text{ with }\|y-x\|_1=1\},$$
		$$\partial_YX=\{(x,y):\ x\in X,\  y\in Y\setminus X \text{ with }\|y-x\|_1=1 \}.$$

		\item For $\Lambda\subset\Z^d, y\in \Z^d$, we  denote  $$\Lambda+y:=\{x+y:\ x\in \Lambda\}.$$

		 	\item Throughout the paper, $C$ denotes a large constant whose value may change from one occurrence to the next, but is ultimately bounded by some absolute constant $C_0$ depending only on $\eta,C_v,M, \widetilde{C}_v, c, s  ,d,\tau,\gamma$.
	\item 	For $a,b\in \C$, we use  $a\sim b$ (resp.	$a\overset{\delta}{\sim }b$ for some $0<\delta<1$) to mean that   $C^{-1}\leq |a/b| \leq C$, (resp.	$( C/\delta)^{-1}\leq a/b \leq C/\delta$). We use 
		 $a\lesssim b$ (resp. $a\gtrsim b$) to mean that   $|a/b| \leq C$ (resp. $|b/a| \leq C$). 
		 
		\item We denote by  $\|B\|$ the $\ell^2$-operator norm of the  operator (or matrix)  $B$.
	\item 	Let  $\alpha>1$ such that   $$\alpha^{8M^{8}}=2,$$ where  $M$ is from 	Condition \ref{con1}.
		To separate  the scale of resonance, we  introduce two functions: 
		$$ h(\delta):=e^{-|\ln\delta|^\alpha},  \quad g(\delta):=e^{-|\ln\delta|^\frac{1}{\alpha}}, \quad  \delta\in (0,1).$$
		  Thus for  all  sufficiently small $\delta$, we have  
		  $$Ch(\delta)^{1/C}<\delta< g(\delta)^{C}/C.$$
		For $k\in \N$, we denote  $h^{(k)}:=\underbrace{h\circ\cdots\circ h}_{k\text{ times}}$  (similar for $g^{(k)}$).
	\end{itemize}

	\section{Multi-scale analysis}\label{MSA}
 Let  $\varepsilon_0>0$ be sufficiently small depending on $\eta,C_v,M, \widetilde{C}_v, c, s  ,d,\tau,\gamma$ and let  $0\leq \varepsilon\leq \varepsilon_0$. By the open mapping theorem, we can further  assume that \begin{equation}\label{xij}
 	\operatorname{Spec}(H)\subset [\inf_{\theta\in \T}v(\theta)-2d\varepsilon_0,\sup_{\theta\in\T} v(\theta)+2d\varepsilon_0] \subset \{v(\theta):\  \theta\in \T_{\eta/4} \}.
 \end{equation} (In fact, by   a simple Rouch\'e argument,  one can prove that there exists some $\delta=\delta(\eta,M,\widetilde{C}_v)>0$  such that $[\inf_{\theta\in \T}v(\theta)-\delta,\sup_{\theta\in\T} v(\theta)+\delta] \subset \{v(\theta):\  \theta\in \T_{\eta/4}\}$.)  It suffices to consider   $$E^*\in \{v(\theta):\  \theta\in \T_{\eta/4} \}$$  in the following  spectral analysis. 

The  goal of  \S \ref{MSA}  is to prove that  the inductive Hypotheses in \S \ref{HP} hold for   all steps $n\geq 1$ (cf. Proposition \ref{1832} and Theorem \ref{key2}).  For a fixed $E^*$,  the key arguments at each step are:
 \begin{itemize}
 	\item [\textbf{(1)}]  Construct  a proper block $B_n$ (depending on $E^*$)  for resonance analysis. 
 	\item [\textbf{(2)}]  Obtain  the roots $\{\theta_{n}^{(i)}(E^*)\}_{i=1}^{m_n}$ and  the $E^*$-relevant Rellich functions of $H_{B_n}(\theta)$ by locally analyzing  the $\theta$-roots and  $E$-roots  of $\operatorname{det}(H_{B_n}(\theta)-E)$, respectively. 
 \item [\textbf{(3)}]	Use the roots  and the Rellich functions from \textbf{(2)}  to control Green's functions, both for the  fixed energy $E^*$ and for nearby energies $E$. 
 \item [\textbf{(4)}] Establish quantitative transversality estimates for the resultant of Rellich functions, which is crucial for the elimination of  double resonances in the proof of Anderson localization. 
 \end{itemize}
	\subsection{The initial step}\label{n=0}
The block $B_0$ is the origin $o$, the roots $\{\theta_0^{(i)}(E^*)\}_{i=1}^{m_0}$ are those from Condition \ref{con1}, and the Rellich function is simply the potential $v$.  Let $$\delta_{0}:=g^{(2)}(\varepsilon_0).$$
\begin{defn}
	Given $\theta^*\in \T_{\eta/4}$,   $\delta\in  [h(\delta_0), e^{-|\ln\delta_0|^{1/2}}] $ and  $E^*$, we define 
	\begin{align*}
		S_0(\theta^*,E^*,\delta):&=\{x\in \Z^d:\ \min_{1\leq i\leq m_0}\|\theta^*+x\cdot \omega-\theta_0^{(i)}(E^*)\|_\T<\delta\},\\
		\widetilde{S}_0(\theta^*,E^*,\delta):&=\{x\in \Z^d:\ |v(\theta^*+x\cdot \omega)-E^*|<\delta\}.
	\end{align*}
	We say  a finite set $\Lambda\subset \Z^d$ is $\theta$-type-$(\theta^*,E^*,\delta)$-$0$-nonresonant   if $\Lambda\cap S_0(\theta^*,E^*,\delta)=\emptyset$, and $E$-type-$(\theta^*,E^*,\delta)$-$0$-nonresonant   if $\Lambda\cap \widetilde{S}_0(\theta^*,E^*,\delta)=\emptyset$.
\end{defn}

	\subsubsection{Green's function estimates for $0$-nonresonant sets}
	\begin{prop}\label{0g}
	Let $\theta^*\in \T_{\eta/4}$,  $\delta\in  [h(\delta_0), e^{-|\ln\delta_0|^{1/2}} ] $ and   $\Lambda$ be    $\theta$-type-$(\theta^*,E^*,\delta)$-$0$-nonresonant. Then  for any $(\theta, E)\in D_\T(\theta^*,h(\delta))\times D(E^*,h(\delta))$, we have 
		\begin{align*}
			\|G_\Lambda^{\theta,E}\|&\leq 3\widetilde{C}_v\delta^{-M},\\
			|G_\Lambda^{\theta,E}(x,y)|&\leq e^{-\gamma_0\|x-y\|_1}, \quad \forall \|x-y\|_1\geq 1,
		\end{align*}	
	where 	 $\gamma_0:=\frac{1}{2}|\ln\varepsilon|$. The estimates   also hold  if $\Lambda$ is  $E$-type-$(\theta^*,E^*,\delta)$-$0$-nonresonant.
	\end{prop}
	\begin{proof}
	Assume $\Lambda\cap S_0(\theta^*,E^*,\delta)=\emptyset$. It follows from 	Condition \ref{con1} and $|\partial_\theta v(\theta)|\lesssim C_v\eta^{-1} $ (by Cauchy integral formula)   that 	for any $x\in \Lambda$ and  $(\theta, E)\in D_\T(\theta^*,h(\delta))\times D(E^*,h(\delta))$, we have  
		\begin{align*}
				|v (\theta+x\cdot \omega) -E| &\geq |v (\theta^*+x\cdot \omega) -E^*|-|v (\theta+x\cdot \omega)-v (\theta^*+x\cdot \omega)|-|E-E^*| \\
				&\geq \widetilde{C}_v^{-1}(\min_{1\leq i\leq m_0}\|\theta^*+x\cdot \omega-\theta_0^{(i)}(E^*)\|)^M-Ch(\delta)\\
				&\geq( 2\widetilde{C}_v)^{-1}  \delta^M  \\  & >10d\varepsilon.
		\end{align*}
 Then the  conclusion follows from the above inequality and a  standard  Neumann series argument.
	\end{proof}
	\subsection{The first inductive step}\label{n=1}
	To analyze the resonance in a smaller scale, we will assign  each $x\in S_0(\theta^*,E^*,\delta_0)\cap\Lambda$ to  a block $B_1(y)=B_1+y$ such that $\theta^*+y\cdot\omega$ is close to some $\theta_0^{(i)}(E^*)$ and  $x\in \operatorname{Core}(B_1(y)):=(Q_{(\operatorname{Diameter}(B_1))^{1/\alpha}}+y)\subset B_1(y)\subset \Lambda$.  Good estimate of $G_\Lambda^{\theta^*,E^*}$ will follow from good estimate of each $G_{B_1(y)}^{\theta^*,E^*}$ by a standard resolvent identity argument. Thus the issue is to estimate each  $G_{B_1(y)}^{\theta^*,E^*}$.    We will prove that $\|G_{B_1(y)}^{\theta^*,E^*}\|$ can be controlled  by some    polynomial of $\theta$, namely,  
    \begin{equation}\label{814}
        \|G_{B_1(y)}^{\theta^*,E^*}\|\leq \delta_0^{-3}\prod_{j=1}^{|\mathcal{C}_0^{(i)}|}\|\theta^*+y\cdot\omega-\theta^{(i,j)}_{1}(E^*)\|_\T^{-1},
    \end{equation}
    where  $\{\theta^{(i,j)}_{1}(E^*)\}_{j=1}^{|\mathcal{C}_0^{(i)}|}$ are  $\theta$-roots of $\operatorname{det}(H_{B_1}(\theta)-E^*)$ in  some   neighborhood of  $\theta_0^{(i)}(E^*)$ and $|\mathcal{C}_0^{(i)}|\in [1,m_0]$  will be specified later. 
	 Besides, for $E$ close to $E^*$,  $\|G_{B_1(y)}^{\theta^*,E}\|$ can   be controlled  by some   polynomial of $E$, namely, 
     \begin{equation}\label{815}
         \|G_{B_1(y)}^{\theta^*,E}\|\leq \delta_0^{-3}\prod_{j=1}^{ |\Omega_0^{(i)}|}|   E-E_1^{(i,j)}(\theta^*+y\cdot\omega)|^{-1},
     \end{equation}  where  $\{E_1^{(i,j)}(\theta)\}_{j=1}^{ |\Omega_0^{(i)}|}$, defined in  some   neighborhood of  $\theta_0^{(i)}(E^*)$,  are  eigenvalues  of $H_{B_1}(\theta)$ close to $E^*$ and $|\Omega_0^{(i)}|\in [1,m_0]  $ will be specified later. 

The proofs of \eqref{814} and \eqref{815} are based on a Schur complement argument similar to   that of \cite{Bou00}. However, we need to  extend this  argument---originally  for the cosine potential and the $\theta$-variable---to the setting of analytic potentials and both the $\theta$- and $E$-variables.  
	\subsubsection{Choice of scale}\label{C1} 
	 To carry out the  Schur complement argument, we require two key structural properties of the block $B_1$ (cf. Figure \ref{F3}):  
	\begin{itemize}
		\item[\textbf{(1)}] The $0$-resonant points in $B_1$ (a point $x\in \Z^d$ is called $0$-resonant if $\theta +x\cdot\omega$ is close to one of the roots $\{\theta_0^{(i)}(E^*)\}_{i=1}^{m_0}$)  lie in the core of $B_1$, forming a nonresonant annulus.
		\item[\textbf{(2)}]  Aside from   at most $M$ resonant points in the core (i.e., resonant within a distance  $\sim\widehat{\delta}_0$, where  $(\ln\delta_{0}^{-1})^{1/2} \ll \ln\widehat{\delta}_{0}^{-1}\ll \ln\delta_{0}^{-1}$  will be specified in  Propositions \ref{scale} and \ref{real} later), all other points in some superlinear extension of $B_1$ are separated from these roots by a distance  significantly larger than $\widehat{\delta}_0^{1/C}$.  
	\end{itemize}
	The existence of such  block  $B_1$ is guaranteed by  the separation of resonant points  due to   $\omega\in 	{\rm DC}_{\tau, \gamma}$. However,   the scale of $B_1$ depends on the $E^*$ (more precisely, depends on the   positions of   $\{\theta_0^{(i)}(E^*)\}_{i=1}^{m_0}$). The goal of this section is  to select an appropriate scale for  $B_1$  for a given $E^*$.
\begin{prop}\label{scale} (Recall that  $\alpha^{8M^{8}}=2$.)	 There exist a resonance parameter  $\bar{\delta}_0\in [g(\delta_0), g^{(4M^8+1)}(\delta_0)]$ and a scale   parameter  $l_1\in [|\ln\delta_0|^4,|\ln\delta_0|^8]$ such that the following statements hold: For each $\theta_0^{(i)}(E^*)$, there is a set $\Omega_0^{(i)}\subset Q_{l_1^{1/\alpha}}$ with $|\Omega_0^{(i)}|\leq M$ such that
	\begin{itemize}
		\item[\textbf{(1)}] For each  $x\in \Omega_0^{(i)}$,   there exists  $ i'\in [1,m_0]$ such that \begin{equation}\label{tiao}
			\|\theta_0^{(i)}(E^*)+x\cdot\omega-\theta_0^{(i')}(E^*)\|_\T\leq \bar{\delta}_0.
		\end{equation} 
	Moreover, if  $\theta_0^{(i'')}(E^*)$ does not  satisfy \eqref{tiao}, then a larger separation holds: \begin{equation}\label{fenk}
		\|\theta_0^{(i)}(E^*)+x\cdot\omega-\theta_0^{(i'')}(E^*)\|_\T>  g^{(M^4)}(\bar{\delta}_0).
	\end{equation}
		\item[\textbf{(2)}]   If $x\in Q_{l_1^{\alpha}}\setminus \Omega_0^{(i)}$, then we have   \begin{equation}\label{??1}
			\min_{1\leq i'\leq m_0}\|\theta_0^{(i)}(E^*)+x\cdot\omega-\theta_0^{(i')}(E^*)\|_\T>  g^{(M^4)}(\bar{\delta}_0).
		\end{equation}
		\item[\textbf{(3)}]  For a fixed  $ i'\in [1,m_0]$, there is at most one $x\in \Omega_0^{(i)}$ satisfying \eqref{tiao}.
	\end{itemize} 

\end{prop}

\begin{proof}
Since  $\omega\in 	{\rm DC}_{\tau, \gamma}$, for each pair $(i,i')$,  there is  at most one  $x\in Q_{|\ln\delta_0|^8}$  such that  $\|\theta_0^{(i)}(E^*)+x\cdot\omega-\theta_0^{(i')}(E^*)\|_\T\leq e^{-|\ln\delta_0|^{1/2}}$. 
  	Partition the scale  interval $(|\ln\delta_0|^4,|\ln\delta_0|^8]$ into $4M^{8}$  disjoint intervals $$(|\ln\delta_0|^{4\alpha^{2k}},|\ln\delta_0|^{4\alpha^{2(k+1)}}], \quad 0\leq k\leq 4M^{8}-1.$$ Since there are at most $M^2$ many  $(i,i')$-pairs,  by  the pigeonhole principle, there exists an interval  $I_{k_0}:=(|\ln\delta_0|^{4\alpha^{2k_0}},|\ln\delta_0|^{4\alpha^{2(k_0+1)}}]$ such that  $$\min_{i,i'}\min_{\|x\|_1\in I_{k_0}}\|\theta_0^{(i)}(E^*)+x\cdot\omega-\theta_0^{(i')}(E^*)\|_\T> e^{-|\ln\delta_0|^\frac{1}{2}}.$$  We define $$l_1:=|\ln\delta_0|^{4\alpha^{2k_0+1}}.$$
  	 Similarly, partition the resonance  interval $(g(\delta_0), g^{(4M^8+1)}(\delta_0)]$ into  $2M^{4}$  disjoint intervals $$(g^{(2M^4k+1)}(\delta_0), g^{(2M^4(k+1)+1)}(\delta_0)], \quad 0\leq k\leq 2M^{4}-1. $$ There exists an interval  $J_{k_0}:=(g^{(2M^4k_0+1)}(\delta_0), g^{(2M^4(k_0+1)+1)}(\delta_0)]$ such that 
  	$$J_{k_0}\bigcap\left(\bigcup_{i,i'}  \ \bigcup_{\|x\|_1\leq  |\ln\delta_0|^8} \left\{\|\theta_0^{(i)}(E^*)+x\cdot\omega-\theta_0^{(i')}(E^*)\|_\T\right\}\right)=\emptyset. $$
We define 
$$\bar{\delta}_0:=g^{(M^4(2k_0+1)+1)}(\delta_0). $$ 
Finally, for each $i$,  we define $$\Omega_0^{(i)}:=\{ x\in Q_{l_1^{1/\alpha}}:\ \min_{1\leq i'\leq m_0} \|\theta_0^{(i)}(E^*)+x\cdot\omega-\theta_0^{(i')}(E^*)\|_\T\leq  \bar{\delta}_0 \}.$$ Then we have $|\Omega_0^{(i)}|\leq M$ since $\omega\in 	{\rm DC}_{\tau, \gamma}$. Now one can check   \textbf{(1)} and \textbf{(2)} from the construction.  \textbf{(3)}  follows easily from $\omega\in 	{\rm DC}_{\tau, \gamma}$. 
\end{proof}
\begin{prop}\label{clas}
The following \eqref{class} defines an equivalence relation on  $\{\theta_0^{(i)}(E^*)\}_{i=1}^{m_0}$: 
\begin{equation}\label{class}
	\theta_0^{(i)}(E^*) \sim \theta_0^{(i')}(E^*) \Leftrightarrow \exists x\in Q_{l_1^{1/\alpha}},  \text{ {\rm s.t.}, } \| \theta_0^{(i)}(E^*)+x\cdot\omega-\theta_0^{(i')}(E^*)\|_\T\leq \bar{\delta}_0. 
\end{equation}
\end{prop}
\begin{proof}
	We only need to show that 
	$$	\theta_0^{(i)}(E^*) \sim \theta_0^{(i')}(E^*) \text{ and }  \theta_0^{(i')}(E^*) \sim \theta_0^{(i'')}(E^*) \Rightarrow 	\theta_0^{(i)}(E^*) \sim \theta_0^{(i'')}(E^*).$$
In fact, $\theta_0^{(i)}(E^*) \sim \theta_0^{(i')}(E^*)$ together with  $ \theta_0^{(i')}(E^*) \sim \theta_0^{(i'')}(E^*)$ implies that there exists some $  x\in Q_{2l_1^{1/\alpha}}$ such that 
$$\| \theta_0^{(i)}(E^*)+x\cdot\omega-\theta_0^{(i'')}(E^*)\|_\T\leq2 \bar{\delta}_0. $$
It follows from   \eqref{??1} and $g^{(M^4)}(\bar{\delta}_0)>2\bar{\delta}_0$ that $x\in \Omega_0^{(i)}\subset Q_{l_1^{1/\alpha}}$. Moreover, $\| \theta_0^{(i)}(E^*)+x\cdot\omega-\theta_0^{(i'')}(E^*)\|_\T\leq \bar{\delta}_0 $  results from Proposition  \ref{scale} \textbf{(1)} by  $g^{(M^4)}(\bar{\delta}_0)>2\bar{\delta}_0$. This means 	$\theta_0^{(i)}(E^*) \sim \theta_0^{(i'')}(E^*).$
\end{proof}

It follows from  Proposition \ref{clas} that  \eqref{class} partitions $\{\theta_0^{(i)}(E^*)\}_{i=1}^{m_0}$ into $  \widetilde{m}_0$ many  different equivalence classes $\mathcal{C}_0^{(1)},\cdots,\mathcal{C}_0^{( \widetilde{m}_0)}$. It follows from  \eqref{??1} and \eqref{class} that elements in the same equivalence class with  $\theta_0^{(i)}(E^*)$ can be shifted by $x\cdot\omega$  for some  $x\in \Omega_0^{(i)}\subset Q_{l_1^{1/\alpha}}$ into a $\bar{\delta}_0$-neighborhood of $\theta_0^{(i)}(E^*)$ (cf. Figure \ref{F2});  however,  elements from different equivalence classes with  $\theta_0^{(i)}(E^*)$ are separated from $\theta_0^{(i)}(E^*)$  by a  significantly larger distance $g^{(M^4)}(\bar{\delta}_0)$ under  any shift $x\cdot\omega$ with  $x\in Q_{l_1^{\alpha}}$.   Choosing a representative for each equivalence class $\mathcal{C}_0^{(i)}$,  we can correspond each  $\mathcal{C}_0^{(i)}$  to a neighborhood $\mathcal{N}_0^{(i)}$. 
We will construct   $\mathcal{N}_0^{(i)}$ as $D_\T(a_i,\widehat{\delta}_0^2)$,   where $a_i\in \T$ is the real part of the representative and  $\widehat{\delta}_0$ will be specified in Proposition \ref{real}.  We choose the center of $\mathcal{N}_0^{(i)}$ on $\T$ for two reasons:
 \begin{itemize}
    \item[\textbf{(1)}] Since  Anderson localization and the  IDS  are concerned with   $\theta\in \T$, it is sufficient to establish Green's function estimates for $\theta$ in any strip containing $\T$.
    \item[\textbf{(2)}] The operator  is  self-adjoint for  $\theta\in \T$ since $v(\theta )$ is real-valued on $\T$,   and it is easier to handle   (not necessary self-adjoint) variation  for self-adjoint operators (cf. Lemma \ref{line}).  
  \end{itemize}
  Note that if the  imaginary part of  the elements in an  equivalence class is too large (compared to others),   then this equivalence class is  noneffective to the resonance analysis for $\theta$ close to $\T$.  Thus  we don't need to consider these equivalence classes far from $\T$, and only need to construct  $\mathcal{N}_0^{(i)}$ for those  equivalence classes close to $\T$. To be clear,   we need the following Proposition \ref{real}:
\begin{prop}\label{real}
	There exists some $\widehat{\delta}_0\in [g^{(M)}(\bar{\delta}_0),g^{(M^3)}(\bar{\delta}_0)]$ such that each equivalence class $\mathcal{C}_0^{(i)}$  satisfies either \textbf{(1)} or \textbf{(2)}: \begin{itemize}
	\item [\textbf{(1)}] All the elements in $\mathcal{C}_0^{(i)}$ belong to $\T_{h(\widehat{\delta}_0)}$.
		\item [\textbf{(2)}] All the elements in $\mathcal{C}_0^{(i)}$ don't belong  to $\T_{g(\widehat{\delta}_0)}$. 
\end{itemize} 
\end{prop}

\begin{proof}
By  \eqref{class}, the imaginary part of the elements in one equivalence class differs at most $\bar{\delta}_0\ll h(g^{(M)}(\bar{\delta}_0))$. Note that there are at most $M$ classes. Thus  the conclusion follows from a pigeonhole principle argument  as in the proof of  Proposition \ref{scale}. 
\end{proof}
For the aforementioned reason, the equivalence classes satisfying \textbf{(2)} are excluded from consideration; therefore, the subsequent resonance analysis focuses only on those satisfying \textbf{(1)}.   Without loss of generality, we can assume  that the  equivalence classes satisfying \textbf{(1)} are the first $\kappa_0$ classes (The case $\kappa_0=0$ is vacuous and requires no further argument)   and the representative of $\mathcal{C}_0^{(i)}$ is $\theta_0^{(i)}(E^*)$ for $1\leq i\leq \kappa_0$.  For each  $1\leq i\leq \kappa_0$,  we  define  \begin{equation}\label{ty}
	\mathcal{N}_0^{(i)}:=D_\T(\Re(\theta_0^{(i)}(E^*)),{\widehat{\delta}_0}^2), \quad {\widetilde{\mathcal{N}}_0}^{(i)}:=D_\T(\Re(\theta_0^{(i)}(E^*)),\widehat{\delta}_0), 
\end{equation} where  $\Re(\theta_0^{(i)}(E^*))\in \T$ is the real part of $\theta_0^{(i)}(E^*)$ (cf. Figure \ref{F2}).  
 
 Figures  \ref{F1}--\ref{F3} below give a cartoon illustration of the construction for an example with   $m_0=6$, $\kappa_0=2$,  $$\mathcal{C}_0^{(1)}=\{\theta_0^{(1)}(E^*),   \theta_0^{(3)}(E^*), \theta_0^{(4)}(E^*)\},  \quad  \mathcal{C}_0^{(2)}=\{\theta_0^{(2)}(E^*), \theta_0^{(5)}(E^*), \theta_0^{(6)}(E^*)\}$$ and  $$\Omega_0^{(1)}=\{o,x,y\}, \quad   \Omega_0^{(2)}=\{o,z\}. $$
  		\begin{figure}[H] \centering
  \begin{tikzpicture}[>=stealth,scale=1]
  	
  	\draw[-> ] (0,0) -- (11.5,0) node[below] {$\theta\in \T$};
  	\draw[->] (0,0) -- (0,6.5) node[right] {$E$};

  	\node[below left] at (0,0) {$O$};
  	\node[above right] at (4.5,5) {$E=v(\theta)$};

  	\draw[ black] plot[smooth, domain=0:3, samples=100] (\x, {6*sin(180*\x/3)});

  	\draw[ black] plot[smooth, domain=3:4.5, samples=100] (\x, {-3*sin(180*(\x-3)/1.5)});

  	\draw[ black] plot[smooth, domain=4.5:7, samples=100] (\x, {5*sin(180*(\x-4.5)/2.5)});
  
  	\draw[ black] plot[smooth, domain=7:7.8, samples=100] (\x, {-1.6*sin(180*(\x-7)/0.8});

  	\draw[ black] plot[smooth, domain=7.8:9, samples=100] (\x, {2.4*sin(180*(\x-7.8)/1.2});
  	\draw[ black] plot[smooth, domain=9:10, samples=100] (\x, {-2*sin(180*(\x-9))});
  	\draw[ dashed, thick] (0,1.8) -- (10.5,1.8) node[right] {$E^*$}; 
  	\node[circle, fill=red, inner sep=1.2pt]  at (0.3,1.8) {}; 	\draw (0.54,2.0)node{{ \color{red}\tiny $\theta_0^{(3)}$}};
  	\node[circle, fill=red, inner sep=1.2pt]  at (2.7,1.8) {}; \draw (2.94,2.0)node{{ \color{red}\tiny $\theta_0^{(1)}$}};
  	\node[circle, fill=red, inner sep=1.2pt]  at (4.8,1.8) {};
  	\draw (5.04,2.0)node{{ \color{red}\tiny $\theta_0^{(4)}$}};
  	\node[circle, fill=blue, inner sep=1.2pt]  at (6.7,1.8) {};
  	\draw (6.94,2.0)node{{ \color{blue}\tiny $\theta_0^{(6)}$}};
  	\node[circle, fill=blue, inner sep=1.2pt]  at (8.13,1.8) {};
  	\draw (7.95,2.0)node{{ \color{blue}\tiny $\theta_0^{(2)}$}};
  	\node[circle, fill=blue, inner sep=1.2pt]  at (8.66,1.8) {};
  	\draw (8.9,2.0)node{{ \color{blue}\tiny $\theta_0^{(5)}$}};
  \end{tikzpicture}
 	\caption{The function  $v(\theta)-E^*$ has six roots in $\T_{\eta/2}$, denoted by $\{\theta_0^{(i)}(E^*)\}_{i=1}^{6}\subset\T$; for brevity, the dependence of $\theta_0^{(i)}$ on $E^*$ in omitted in the figure.} \label{F1}
  \end{figure}

	\begin{figure}[H] \centering
\begin{tikzpicture}[>=stealth,scale=1.1]
		\node[below left] at (0,0) {$O$};
	\draw[->,dashed] (0.3,0) to[out=15, in=165]  node[midway, fill=white, font=\tiny] {{\tiny \color{red}$-x\cdot \omega$}} (2.4,0);
	\draw[->,dashed] (4.5,0) to [out=165, in=15]  node[midway, fill=white, font=\tiny] {{\color{red}$-y\tiny \cdot \omega$}} (2.6, -0);
	\draw[->,dashed] (6.3,0) to [out=15, in=165]  node[midway, fill=white, font=\tiny] {{\color{blue}$-z\tiny \cdot \omega$}} (8.3, 0);
	\draw[->,thin] (0,0) -- (11.5,0) node[below] {$\Re \theta$};
	\draw[->,thin] (0,0) -- (0,2.5) node[right] {$\Im \theta$};
	\node[circle, fill=red, inner sep=1pt,] (t3) at (0.3,0) {};
	\node[circle, fill=red, inner sep=1pt] (t1) at (2.5,0) {}; 
	\draw (6.1,0.18)node{\color{blue}{\tiny$\theta_0^{(6)}$}}; 
	\draw (8.5,0.18)node{\color{blue}{\tiny$\theta_0^{(2)}$}}; 	\draw (9,0.18)node{\color{blue}{\tiny$\theta_0^{(5)}$}};
	\draw (4.8,0.15)node{\color{red}{\tiny $\theta_0^{(4)}$}}; \draw (0.5,0.15)node{{ \color{red}\tiny $\theta_0^{(3)}$}};
	\draw (2.5,1.7)node{\color{red}{\tiny$\widetilde{\mathcal{N}}_0^{(1)}$}};
	\draw (2.5,0.8)node{\color{red}{\tiny$\mathcal{N}_0^{(1)}$}};
	\draw (8.4,1.7)node{\color{blue}{\tiny$\widetilde{\mathcal{N}}_0^{(2)}$}};
	\draw (8.4,0.8)node{\color{blue}{\tiny$\mathcal{N}_0^{(2)}$}};
	\draw (2.7,0.18)node{\color{red}{\tiny$\theta_0^{(1)}$}}; 
	\node[circle, fill=red, inner sep=1pt, ]  at (4.5,0) {};
	\node[circle, fill=blue, inner sep=1pt, ] at (8.4,0) {}; 
	\node[circle, fill=blue, inner sep=1pt, ]  at (6.3,0) {}; 	\node[circle, fill=blue, inner sep=1pt, ]  at (8.7,0) {};
	\draw[dashed] (2.5,0)  ellipse (0.6 and 0.6); 
	\draw (2.5,0) ellipse (1.5 and 1.5);
	\draw [dashed] (8.4,0)  ellipse (0.6 and 0.6); 
	\draw (8.4,0) ellipse (1.5 and 1.5);  
\end{tikzpicture}
	\caption{These roots $\{\theta_0^{(i)}(E^*)\}_{i=1}^{6}$ are plotted in the complex plane and partitioned by \eqref{class} into two equivalence classes:  $\mathcal{C}_0^{(1)}$ (marked in red) and  $\mathcal{C}_0^{(2)}$ (marked in blue), with  representatives  $\theta_0^{(1)}(E^*)$ and $\theta_0^{(2)}(E^*)$, respectively.    Here      $\mathcal{N}_0^{(1)}, {\widetilde{\mathcal{N}}_0}^{(1)}$ and  $\mathcal{N}_0^{(2)}, {\widetilde{\mathcal{N}}_0}^{(2)}$  are neighborhoods associated  with $\theta_0^{(1)}(E^*)$ and $\theta_0^{(2)}(E^*)$, respectively  (cf. \eqref{ty}).    Elements in $\mathcal{C}_0^{(i)}$  can be shifted by $p\cdot\omega$   for some  $p\in \Omega_0^{(i)}$ into $\mathcal{N}_0^{(i)}$, e.g., $\theta_0^{(3)}-x\cdot\omega\in \mathcal{N}_0^{(1)} $.  By contrast, elements not belonging to $\mathcal{C}_0^{(i)}$  are separated from ${\widetilde{\mathcal{N}}_0}^{(i)}$  by a distance which is significantly larger than $\operatorname{Diameter}({\widetilde{\mathcal{N}}_0}^{(i)})$  under any shift  $p\cdot\omega$ with  $p\in Q_{l_1^{\alpha}}$.} \label{F2}
\end{figure}

	\begin{figure}[H]  \centering
\begin{tikzpicture}[>=stealth]

	\draw[dashed] (0,0) -- (3,3) -- (6,0) -- (3,-3) -- cycle;

	\draw[dashed] (2.3,0) -- (3,0.7) -- (3.7,0) -- (3,-0.7) -- cycle;
	\draw(1.16,0) -- (3,1.84) -- (4.84,0) -- (3,-1.84) -- cycle;

	\fill[red] (3,0) circle (1pt);
	\fill[red] (3.1,0.2) circle (1pt);
	\fill[red] (3,-0.2) circle (1pt);
	
\node[red] at (2.8,0) {\tiny $o$};
\node[red] at (3.2,0.1) {\tiny $x$};
\node[red] at (2.8,-0.2) {\tiny $y$};
\node[ above right] at (1.2,0.5) {\tiny$Q_{l_1}$};
	\node[ above right] at (1.3,2) {\tiny$Q_{l_1^\alpha}$};
	\node[ right] at (3.1,0.4) {\tiny$Q_{l_1^{1/\alpha}}$};
	\draw[dashed] (7,0) -- (10,3) -- (13,0) -- (10,-3) -- cycle;

	\draw[dashed] (9.3,0) -- (10,0.7) -- (10.7,0) -- (10,-0.7) -- cycle;
\draw(8.16,0) -- (10,1.84) -- (11.84,0) -- (10,-1.84) -- cycle;

\node[ above right] at (8.2,0.5) {\tiny$Q_{l_1}$};

	\fill[blue] (10,0) circle (1pt);
\fill[blue] (10.1,-0.1) circle (1pt);
\node[blue] at (10.2,-0.2) {\tiny $z$};
\node[blue] at (9.8,0) {\tiny $o$};
	\node[ above right] at (8.3,2) {\tiny$Q_{l_1^\alpha}$};
	\node[ right] at (10.1,0.4) {\tiny$Q_{l_1^{1/\alpha}}$};
\end{tikzpicture} 
	\caption{For $\theta\in  {\widetilde{\mathcal{N}}_0}^{(1)}$ (resp. $\theta\in  {\widetilde{\mathcal{N}}_0}^{(2)}$), the resonant points $p\in Q_{l_1^\alpha}$  satisfying   $\min_{i=1}^{6}\|\theta +p\cdot\omega- \theta_0^{(i)}(E^*)\|_\T\lesssim \widehat{\delta}_0$  are  precisely the red points $\{o,x,y\}$ (resp. the blue points $\{o,z\}$), which lie in the core $Q_{l_1^{1/\alpha}}$;  moreover,  all other  nonresonant points $q\in Q_{l_1^\alpha}$   satisfy  $\min_{i=1}^{6}\|\theta +q\cdot\omega- \theta_0^{(i)}(E^*)\|_\T> g^{(2)}(\widehat{\delta}_0)$. } \label{F3} 
\end{figure}

\subsubsection{Resonance analysis via Schur complement}\label{R1}
In this section, we carry out  the resonance analysis of $B_1:=Q_{l_1}$  in each  ${\widetilde{\mathcal{N}}_0}^{(i)}$ with  $1\leq i\leq \kappa_0$ to obtain the roots $\{\theta^{(i,j)}_{1}(E^*)\}_{j=1}^{|\mathcal{C}_0^{(i)}|}$ and Rellich functions $\{E_1^{(i,j)}(\theta)\}_{j=1}^{ |\Omega_0^{(i)}|}$.  

Recall $l_1,\bar{\delta}_0$ and the set $\Omega_0^{(i)}$ (cf. Proposition \ref{scale}). We define  $$ \widehat{B}^{(i)}_1:= B_1\setminus \Omega_0^{(i)}.$$  For $x\in \Omega_0^{(i)}$, we define  the set $$\mathcal{C}_0^{(i),x}:=\{\theta_0^{(i')}(E^*)\in \mathcal{C}_0^{(i)}:\ \|\theta_0^{(i)}(E^*)+x\cdot \omega-\theta_0^{(i')}(E^*)\|_\T\leq \bar{\delta}_0\}. $$
By Proposition  \ref{scale} \textbf{(3)}, we have \begin{equation}\label{jish}
	\sum_{x\in\Omega_0^{(i)}}|\mathcal{C}_0^{(i),x}|=|\mathcal{C}_0^{(i)}|.
\end{equation}
\begin{prop}\label{SL}The following statements hold: 
	\begin{itemize}
		\item[\textbf{(1)}]  For $x\in \Omega_0^{(i)}$ and $\theta_0^{(i')}(E^*)\in \mathcal{C}_0^{(i),x}$, we have  $\theta_0^{(i')}(E^*)-x\cdot \omega  \in \mathcal{N}_0^{(i)}$ (considered in $\C/\Z$) (cf. Figure \ref{F2}).  
		\item  [\textbf{(2)}] For $x\in \Omega_0^{(i)}$ and $\theta\in  {\widetilde{\mathcal{N}}_0}^{(i)}$, we have 
		$$v(\theta+x\cdot\omega)-E^*\overset{g(\bar{\delta}_0)}{\sim} \prod_{\theta_0^{(i')}(E^*)\in \mathcal{C}_0^{(i),x}}\|\theta-(\theta_0^{(i')}(E^*)-x\cdot \omega)\|_\T.$$
		\item[\textbf{(3)}] For any  $\theta\in \widetilde{\mathcal{N}}_0^{(i)}$, the set $\widehat{B}^{(i)}_1$ is $\theta$-type-$(\theta,E^*,g^{(2)}(\widehat{\delta}_0))$-$0$-nonresonant (cf. Figure \ref{F3}).  
	\end{itemize}
\end{prop}
\begin{proof}  Recall that  $\widehat{\delta}_0\in [g^{(M)}(\bar{\delta}_0),g^{(M^3)}(\bar{\delta}_0)]$ (cf. Proposition \ref{real}).
Now 	\textbf{(1)} follows from $\bar{\delta}_0\ll h(\widehat{\delta}_0)\ll\widehat{\delta}_0^2$;	\textbf{(2)} follows from 	Condition \ref{con1} and \eqref{fenk} (\eqref{fenk} implies  $	\|\theta_0^{(i)}(E^*)+x\cdot\omega-\theta_0^{(i')}(E^*)\|_\T> g^{(M^4)}(\bar{\delta}_0)\gg \widehat{\delta}_0\gg g(\bar{\delta}_0)^{1/M} $ for $\theta_0^{(i')}(E^*)\notin \mathcal{C}_0^{(i),x} $); 	\textbf{(3)} follows from Proposition \ref{scale} \textbf{(2)} and $g^{(M^4)}(\bar{\delta}_0)\gg g^{(2)}(\widehat{\delta}_0)$.
\end{proof}
Write 
\begin{equation}\label{juzhen}
		E-H_{B_1}(\theta)=\begin{pmatrix}
	E-H_{\Omega_0^{(i)}}(\theta) & \Gamma^{{\rm T}} \\
	\Gamma & E-H_{\widehat{B}^{(i)}_1}(\theta)
\end{pmatrix},
\end{equation}where $\Gamma$ is the coupling term due to $\varepsilon\Delta$.
Now restrict   $(\theta,E)\in U_0^{(i)}$, where  $$U_0^{(i)}:= \widetilde{\mathcal{N}}_0^{(i)}\times D(E^*,g(\widehat{\delta}_0)).$$ We consider the Schur complement 
$$S_{\Omega_0^{(i)}}(\theta,E):=	E-H_{\Omega_0^{(i)}}(\theta)-\Gamma^{{\rm T}} G_{\widehat{B}^{(i)}_1}^{\theta,E} \Gamma,$$
 By Propositions \ref{0g} and \ref{SL} \textbf{(3)},  we have $\|G_{\widehat{B}^{(i)}_1}^{\theta,E}\|\leq 3\widetilde{C}_vg^{(2)}(\widehat{\delta}_0)^{-M} \leq \widehat{\delta}_0^{-1}$. Since $\|\Gamma\|\leq 2d\varepsilon$, we have 
\begin{equation}\label{s1}
	S_{\Omega_0^{(i)}}(\theta,E)=	\operatorname{Diag}(E-v(\theta+x\cdot\omega))_{x\in \Omega_0^{(i)}}-R(\theta,E),
\end{equation} where 
 $R(\theta,E)$ is an $|\Omega_0^{(i)}|\times |\Omega_0^{(i)} |$ analytic matrix of  $(\theta,E)\in U_0^{(i)}$  satisfying  $$ \|R(\theta,E)\|\leq 2d\varepsilon+4d^2\varepsilon^2\widehat{\delta}_0^{-1}\lesssim\varepsilon.$$ 
Define 
$$s_{\Omega_0^{(i)}}(\theta,E):=\operatorname{det}S_{\Omega_0^{(i)}}(\theta,E).$$ Hence $s_{\Omega_0^{(i)}}(\theta,E)$ is an  analytic function in $U_0^{(i)}$. 
 By \eqref{s1}, we have \begin{equation}\label{1712}
	s_{\Omega_0^{(i)}}(\theta,E)=Q_1^{(i)}(\theta,E)-r(\theta,E), 
\end{equation}where $$\text{ $ Q_1^{(i)}(\theta,E):=	\prod_{x\in \Omega_0^{(i)}}\left(E-v(\theta+x\cdot\omega)\right),$ \quad $|r(\theta,E)|\lesssim\varepsilon$.} $$ 
\begin{rem}\label{dif}
In estimating  $r(\theta,E)$, we use the following fact: for any two $N\times N$ matrices  $S$ and $T$, 
	$$	|\operatorname{det} S-\operatorname{det} T| \leq N\|S-T\|\cdot \left(\max (\|S\|,\|T\|)\right)^{N-1}. $$
This inequality will be used repeatedly below to estimate differences of determinants.
\end{rem}

  We are going to  apply  Rouch\'e  argument 
    to \eqref{1712} to show that $s_{\Omega_0^{(i)}}(\theta,E)$  is  equivalent to  some polynomial. We will do this for both $\theta$ and $E$ variables. We start with the  $\theta$-variable. 
\begin{prop}\label{654}
	 Fix $E=E^*$. Then the $\theta$-variable analytic function  $s_{\Omega_0^{(i)}}(\theta,E^*)$ has $|\mathcal{C}_0^{(i)}|$  many  roots in  ${\rm Db} (\mathcal{N}_0^{(i)}):= D_\T(\Re(\theta_0^{(i)}(E^*)),2\widehat{\delta}_0^2)$,  denoted by $\{\theta^{(i,j)}_1(E^*)\}_{j=1}^{|\mathcal{C}_0^{(i)}|}$,    and has  no other root in the extension $\widetilde{\mathcal{N}}_0^{(i)}$.  Moreover, we have  
	  $$s_{\Omega_0^{(i)}}(\theta,E^*)\overset{\delta_0}{\sim} \prod_{j=1}^{|\mathcal{C}_0^{(i)}|}\left(\theta-\theta^{(i,j)}_1(E^*) \right),\quad \forall \theta\in \widetilde{\mathcal{N}}_0^{(i)}.$$
\end{prop}

\begin{proof}
For $\theta\in \widetilde{\mathcal{N}}_0^{(i)}$,	by Proposition \ref{SL} \textbf{(1)}, \textbf{(2)} and \eqref{jish},  we have 
	\begin{align}
Q_1^{(i)}(\theta,E^*)	&\overset{g(\bar{\delta}_0)^M}{\sim} \prod_{x\in \Omega_0^{(i)}} \  \prod_{\theta_0^{(i')}(E^*)\in \mathcal{C}_0^{(i),x}}\|\theta-(\theta_0^{(i')}(E^*)-x\cdot \omega)\|_\T\\
	&\ \ \	\sim \prod_{\nu=1}^{|\mathcal{C}_0^{(i)}|} \left(\theta-\widetilde{\theta}_0^{(i,\nu)}(E^*)\right)=:K_1^{(i)}(\theta), \label{fafaf}
	\end{align}
		where $\{\widetilde{\theta}_0^{(i,\nu)}(E^*)\}_{\nu=1}^{|\mathcal{C}_0^{(i)}|}=\{\theta_0^{(i')}(E^*)-x\cdot \omega:\ x\in \Omega_0^{(i)},  \theta_0^{(i')}(E^*)\in \mathcal{C}_0^{(i),x}\}\subset \mathcal{N}_0^{(i)}$ (considered in $\C/\Z$). Combining  \eqref{fafaf}   with  \eqref{1712} yields 
\begin{equation}\label{1926}
			 s_{\Omega_0^{(i)}}(\theta,E^*)\overset{g(\bar{\delta}_0)^M}{\sim} K_1^{(i)}(\theta)-O(\varepsilon g(\bar{\delta}_0)^{-M}):=L_1^{(i)}(\theta).
\end{equation}
Since 	$ \widetilde{\theta}_0^{(i,\nu)}(E^*)\in \mathcal{N}_0^{(i)}=D_\T(\Re(\theta_0^{(i)}(E^*)),\widehat{\delta}_0^2)$, we have  $|K_1^{(i)}(\theta)|\geq  \widehat{\delta}_0^{2M}\gg \varepsilon g(\bar{\delta}_0)^{-M} $ on the circle $\{\theta\in \C/\Z:\ \|\theta-\Re(\theta_0^{(i)}(E^*))\|_\T=2\widehat{\delta}_0^2\}.$  It follows from  Rouch\'e theorem that 
$L_1^{(i)}(\theta)$ 	 has 	$|\mathcal{C}_0^{(i)}|$ roots in $D_\T(\Re(\theta_0^{(i)}(E^*)),2\widehat{\delta}_0^2)$, denoted by $\{\theta^{(i,j)}_1(E^*)\}_{j=1}^{|\mathcal{C}_0^{(i)}|}$.   Thus $L_1^{(i)}(\theta)/ \prod_{j=1}^{|\mathcal{C}_0^{(i)}|}(\theta-\theta^{(i,j)}_1(E^*) )$ is analytic in $\theta\in \widetilde{\mathcal{N}}_0^{(i)}$. 
Noting  that for $\theta\in \partial  \widetilde{\mathcal{N}}_0^{(i)}= \{\theta\in \C/\Z:\ \|\theta-\Re(\theta_0^{(i)}(E^*))\|_\T=\widehat{\delta}_0\}$, 
$$\left|\frac{\widetilde{\theta}_0^{(i,\nu)}(E^*)-\Re(\theta_0^{(i)}(E^*)) }{\theta-\Re(\theta_0^{(i)}(E^*))}\right|,\  \left|\frac{\theta^{(i,j)}_1(E^*)-\Re(\theta_0^{(i)}(E^*)) }{\theta-\Re(\theta_0^{(i)}(E^*))}\right|\leq2\widehat{\delta}_0,$$  we have 
\begin{align*}
	\frac{L_1^{(i)}(\theta)}{\prod_{j=1}^{|\mathcal{C}_0^{(i)}|}(\theta-\theta^{(i,j)}_1(E^*) )}&=\frac{\prod_{\nu=1}^{|\mathcal{C}_0^{(i)}|}
		 \left(1-\frac{\widetilde{\theta}_0^{(i,\nu)}(E^*)-\Re(\theta_0^{(i)}(E^*)) }{\theta-\Re(\theta_0^{(i)}(E^*))}\right)-O(\varepsilon g(\bar{\delta}_0)^{-M}\widehat{\delta}_0^{-M})}{ \prod_{j=1}^{|\mathcal{C}_0^{(i)}|} \left(1-\frac{\theta^{(i,j)}_1(E^*)-\Re(\theta_0^{(i)}(E^*)) }{\theta-\Re(\theta_0^{(i)}(E^*))}\right) } \\
	&=1+O(\widehat{\delta}_0), \quad \forall \theta\in \partial  \widetilde{\mathcal{N}}_0^{(i)}.
\end{align*}
By the maximum modulus principle, we have 
$$\sup_{\theta\in \widetilde{\mathcal{N}}_0^{(i)}} \left|\frac{L_1^{(i)}(\theta)}{\prod_{j=1}^{|\mathcal{C}_0^{(i)}|}(\theta-\theta^{(i,j)}_1(E^*) )}-1\right|=\sup_{\theta\in \partial  \widetilde{\mathcal{N}}_0^{(i)}}\left|\frac{L_1^{(i)}(\theta)}{\prod_{j=1}^{|\mathcal{C}_0^{(i)}|}(\theta-\theta^{(i,j)}_1(E^*) )}-1\right|\lesssim \widehat{\delta}_0, $$
which implies  $$L_1^{(i)}(\theta)\sim \prod_{j=1}^{|\mathcal{C}_0^{(i)}|}(\theta-\theta^{(i,j)}_1(E^*) ), \quad   \forall \theta\in \widetilde{\mathcal{N}}_0^{(i)}.$$
Recalling \eqref{1926} and $\delta_0\ll g(\bar{\delta}_0)^M$, we finish the proof.
\end{proof}
	Next we do this for the  $E$-variable.

\begin{prop}\label{se}
	For any $\theta^*\in \widetilde{\mathcal{N}}_0^{(i)}$,  the $E$-variable analytic function $s_{\Omega_0^{(i)}}(\theta^*,E)$  has $|\Omega_0^{(i)}|$ many  roots in $ D(E^*,\widehat{\delta}_0^\frac{1}{2})$, denoted by $\{E_1^{(i,j)}(\theta^*)\}_{j=1}^{|\Omega_0^{(i)}|}$,   and has no other root in the extension $ D(E^*,g(\widehat{\delta}_0))$.  By Lemma \ref{Su}, these roots are the eigenvalues of $H_{B_1} (\theta^*)$ in $D(E^*,g(\widehat{\delta}_0)).$  Moreover,  we have 
$$s_{\Omega_0^{(i)}}(\theta^*,E){\sim} \prod_{j=1}^{|\Omega_0^{(i)}|}\left(E-E_1^{(i,j)}(\theta^*) \right),\quad \forall E \in D(E^*,g(\widehat{\delta}_0)) .$$
\end{prop}
\begin{proof}
 By Proposition \ref{SL} \textbf{(1)}, 	for $\theta^*\in \widetilde{\mathcal{N}}_0^{(i)}$ and $x \in \Omega_0^{(i)}$,  we have $\|\theta^*+x\cdot\omega- \theta_0^{(i')}(E^*)\|_\T\leq2\widehat{\delta}_0 $ for some  $\theta_0^{(i')}(E^*)\in \mathcal{C}_0^{(i),x}$. Since $|\partial_\theta v(\theta)|\lesssim C_v\eta^{-1} $, we have  $$|v(\theta^*+x\cdot\omega)-E^*|= |v(\theta^*+x\cdot\omega)-v(\theta_0^{(i')}(E^*))|   \lesssim \|\theta^*+x\cdot\omega- \theta_0^{(i')}(E^*)\|_\T \lesssim  \widehat{\delta}_0.$$
Using the fact that  $\widehat{\delta}_0\ll\widehat{\delta}_0^\frac{1}{2}\ll   g(\widehat{\delta}_0)$,	the  conclusion follows from the same argument as in the proof of Proposition \ref{654} by applying Rouch\'e  argument to \eqref{1712} for the $E$-variable. 
\end{proof} 
	 \begin{rem}\label{kuoda}  
	Using  the fact that  $\widehat{\delta}_0\ll g(\widehat{\delta}_0)$, we can replace $ \widetilde{\mathcal{N}}_0^{(i)}$ by the  	quadrupling   ${\rm Qdp}(\widetilde{\mathcal{N}}_0^{(i)}):=D_\T(\Re(\theta_0^{(i)}(E^*)),4\widehat{\delta}_0)$ in the proof of Propositions \ref{SL}, \ref{654}, \ref{se} and establish   the same  conclusions in  $\theta\in {\rm Qdp}(\widetilde{\mathcal{N}}_0^{(i)})$.
\end{rem} 

	\subsubsection{Green's function estimates for $1$-nonresonant  sets}\label{G1}
Our goal in this section is to use the roots $\{\theta_{1}^{(i,j)}(E^*)\}_{j=1}^{|\mathcal{C}_0^{(i)}|}$ and the Rellich functions $\{E_1^{(i,j)}(\theta)\}_{j=1}^{|\Omega_0^{(i)}|}$ from the previous section to control Green's functions, both for the  fixed energy $E^*$ and for nearby energies $E$. 
	
	\begin{lem}\label{1833}
		For $(\theta,E)\in U_0^{(i)}$, we have $$\|G_{B_1}^{\theta,E}\|\leq \delta_0^{-1}|s_{\Omega_0^{(i)}}(\theta,E)|^{-1}.$$
	\end{lem}
	\begin{proof}
		By Lemma \ref{Su},  we have $\|G_{B_1}^{\theta,E}\|\lesssim \|S_{\Omega_0^{(i)}}(\theta,E)^{-1}\|\cdot   \|G_{\widehat{B}^{(i)}_1}^{\theta,E}\|^2$. 
		Now the conclusion follows from  $\|G_{\widehat{B}^{(i)}_1}^{\theta,E}\|\lesssim \widehat{\delta}_0^{-M}\ll\delta_0^{-1/2}$  and $\|S_{\Omega_0^{(i)}}(\theta,E)^{-1}\|=|s_{\Omega_0^{(i)}}(\theta,E)|^{-1}\cdot \|S_{\Omega_0^{(i)}}(\theta,E)^{\#}\|\leq |s_{\Omega_0^{(i)}}(\theta,E)|^{-1} \cdot MC^M$, where we denote by $S^{\#}$ the adjugate  matrix of $S$. 
	\end{proof}
Combining Propositions \ref{654}, \ref{se} and Lemma \ref{1833} yields  
	\begin{prop}\label{aha}  \begin{itemize}
			\item [\textbf{(1)}]
		For $\theta\in \widetilde{\mathcal{N}}_0^{(i)}$, we have 		$$  \|G_{B_1}^{\theta,E^*}\|\leq \delta_0^{-3}\prod_{j=1}^{|\mathcal{C}_0^{(i)}|}\|\theta-\theta^{(i,j)}_1(E^*)\|_\T^{-1}. $$
			\item [\textbf{(2)}]For $(\theta,E)\in U_0^{(i)}$, we have  
$$  \|G_{B_1}^{\theta,E}\|\leq \delta_0^{-3}\prod_{j=1}^{|\Omega_0^{(i)}|}|E-E_1^{(i,j)}(\theta)|^{-1}. $$\end{itemize}
	\end{prop}

	\begin{defn} We say  a finite set $\Lambda\subset\Z^d$ is 
		$(\theta^*,E^*)$-$1$-regular if $$x\in S_0(\theta^*,E^*,g(\delta_0)) \cap \Lambda \Rightarrow (Q_{5l_1}+x)\subset \Lambda.$$
	 Let $\delta_1:=e^{-l_1^{2/3}}$. 	Given $\theta^*\in \T_{\widehat{\delta}_0}$, $\delta\in [h(\delta_1),\delta_0]$ and $E^*$,  we define     	
	 \begin{align*}
	 		S_1(\theta^*,E^*,\delta):=&\{x\in \Z^d:\ \exists i\in [1,\kappa_0],   \text{  {\rm s.t.}, }\theta^*+x\cdot \omega \in \widetilde{\mathcal{N}}_0^{(i)}
	 		 \\\text{ and } &\min_{1\leq j\leq |\mathcal{C}_0^{(i)}|}\|\theta^*+x\cdot \omega-\theta_1^{(i,j)}(E^*)\|_\T<\delta\},
	 \end{align*}
	and for  $E\in D(E^*,h(\delta_0))$, we define 
	 \begin{align*}
	\widetilde{S}_1(\theta^*,E,\delta):=&\{x\in \Z^d:\ \exists i\in [1,\kappa_0],  \text{  {\rm s.t.}, }\theta^*+x\cdot \omega \in \widetilde{\mathcal{N}}_0^{(i)}
		\\ \text{ and } &\min_{1\leq j\leq |\Omega_0^{(i)}|}|E_1^{(i,j)}(\theta^*+x\cdot\omega )-E|<\delta\}.
	\end{align*}
	  We say  $\Lambda$ is  $\theta$-type-$(\theta^*,E^*,\delta)$-$1$-nonresonant   if $\Lambda\cap S_1(\theta^*,E^*,\delta)=\emptyset$. 
	For $E\in D(E^*,h(\delta_0))$, we say  $\Lambda$  is $E$-type-$(\theta^*,E,\delta)$-$1$-nonresonant   if $\Lambda\cap\widetilde{S}_1(\theta^*,E,\delta)=\emptyset$.
	\end{defn}
	\begin{prop}\label{1go}
	Let $\theta^*\in \T_{\widehat{\delta}_0}$, $\delta\in [h(\delta_1),\delta_0]$  and   $\Lambda$ be   $(\theta^*,E^*)$-$1$-regular. Then we have 
	\begin{itemize}
		\item[\textbf{(1)}] If  $\Lambda$ is  $\theta$-type-$(\theta^*,E^*,\delta)$-$1$-nonresonant,  then  for any $(\theta, E)\in D_\T(\theta^*,h(\delta))\times D(E^*,h(\delta))$, we have 
	\begin{align}
		\|G_\Lambda^{\theta,E}\|&\leq \delta_0^{-5}\delta^{-M},\label{12}\\
		|G_\Lambda^{\theta,E}(x,y)|&\leq e^{-\gamma_1\|x-y\|_1}, \quad \forall \|x-y\|_1\geq l_1^\frac{2}{\alpha+1}.\label{13}
	\end{align}	
	where  $\gamma_1:=(1-l_1^{-(\alpha-1)^8})\gamma_0$. 
	\item[\textbf{(2)}]  Let $E\in D(E^*,h(\delta_0))$. If  $\Lambda$ is  $E$-type-$(\theta^*,E,\delta)$-$1$-nonresonant, then  \eqref{12} and  \eqref{13} hold for $\theta=\theta^*$. 
		\end{itemize}
\end{prop}

\begin{proof}If $ S_0(\theta^*,E^*,\delta_0)\cap \Lambda=\emptyset$, the conclusions follow from Proposition \ref{0g}. Now assume $ S_0(\theta^*,E^*,\delta_0)\cap \Lambda\neq \emptyset$. 
First,  we  prove the following  Claim \ref{1goc}:  
 \begin{claim}\label{1goc}
 	 For  every $x\in S_0(\theta^*,E^*,\delta_0)\cap \Lambda$, there exists a block $B_1(y)(=B_1+y)$  satisfying $x\in  (Q_{l_1^{1/\alpha}}+y)\subset B_1(y) \subset \Lambda$ and    \begin{align}
 		\|G_{B_1(y)}^{\theta,E}\|&\leq 2\delta_0^{-3}\delta^{-M},\label{fafa}\\
 		|G_{B_1(y)}^{\theta,E}(x',y')|&\leq e^{-\gamma_1'\|x'-y'\|_1}, \quad \forall \|x'-y'\|_1\geq l_1^\frac{3}{2\alpha+1}, \label{fafa1}
 	\end{align}	
 	where  $\gamma_1':=(1-l_1^{-2(\alpha-1)^8})\gamma_0$.
 \end{claim}
	\begin{proof}[Proof of Claim]
	
	Since $x\in S_0(\theta^*,E^*,\delta_0)$, there exists some $1\leq i\leq m_0$ such that \begin{equation}\label{1921}
		\|\theta^*+x\cdot\omega-\theta_0^{(i)}(E^*)\|_\T<\delta_0.
	\end{equation} By  $\theta^* \in \T_{\widehat{\delta}_0}$ and Proposition \ref{real}, 
	$\theta_0^{(i)}(E^*)$ must belong to $\mathcal{C}_0^{(i')}$ for some $1\leq i'\leq \kappa_0$.  Thus by \eqref{class}, there exists $z\in Q_{l_1^{1/\alpha}}$ such that \begin{equation}\label{21}
		\|\theta_0^{(i)}(E^*)+z\cdot\omega-\theta_0^{(i')}(E^*)\|_\T\leq \bar{\delta}_0. 
	\end{equation}
Let $y=x+z$. By \eqref{1921} and \eqref{21}, we have $\|\theta^*+y\cdot\omega-\theta_0^{(i')}(E^*)\|_\T<2\bar{\delta}_0.$ Thus $$\theta^*+y\cdot\omega\in \mathcal{N}_0^{(i')}.$$  Since $\Lambda$ is  $(\theta^*,E^*)$-$1$-regular and $x\in S_0(\theta^*,E^*,\delta_0)\cap \Lambda$, we have $B_1(y)\subset   (Q_{5l_1}+x)\subset \Lambda.$
\begin{itemize}
	\item If $\Lambda\cap S_1(\theta^*,E^*,\delta)=\emptyset$, we have $\min_{1\leq j\leq |\mathcal{C}_0^{(i')}|}\|\theta^*+y\cdot \omega-\theta_1^{(i',j)}(E^*)\|_\T\geq \delta$. Thus by Proposition \ref{aha} \textbf{(1)}, we have  
	$$\|G_{B_1(y)}^{\theta^*,E^*}\|=\|G_{B_1}^{\theta^*+y\cdot\omega,E^*}\|\leq \delta_0^{-3}\left(\min_{1\leq j\leq |\mathcal{C}_0^{(i')}|}\|\theta^*+y\cdot \omega-\theta_1^{(i',j)}(E^*)\|_\T\right)^{-M}\leq \delta_0^{-3}\delta^{-M}. $$ To prove \eqref{fafa}, we just note that   $$\|(H_{B_1(y)}(\theta)-E)-(H_{B_1(y)}(\theta^*)-E^*)\|\lesssim h(\delta)\ll \delta^{M+3}$$ and for two matrices $A,B$ such that $\|A-B\|<\|B^{-1}\|^{-1}$,  
	$$\|A^{-1}\|=\|(I+B^{-1}(A-B))^{-1}B^{-1}\|\leq (1-\|B^{-1}\|\cdot\|A-B\|)^{-1}\|B^{-1}\|.$$
	\item  If $\Lambda\cap \widetilde{S}_1(\theta^*,E,\delta)=\emptyset$, we have 
$\min_{1\leq j\leq |\Omega_0^{(i')}|}|E_1^{(i',j)}(\theta^*+y\cdot\omega )-E|\geq \delta$. 
Thus by Proposition \ref{aha} \textbf{(2)}, we have  
$$\|G_{B_1(y)}^{\theta^*,E}\|=\|G_{B_1}^{\theta^*+y\cdot\omega,E}\|\leq \delta_0^{-3}\left(\min_{1\leq j\leq |\Omega_0^{(i')}|}|E-E_1^{(i',j)}(\theta^*+y\cdot\omega)|\right)^{-M}\leq \delta_0^{-3}\delta^{-M}. $$
\end{itemize}
 Thus \eqref{fafa} is proved.  Now \eqref{fafa1}  follows from a standard resolvent identity argument,  using \eqref{fafa}, $\ln\delta^{-1}\leq  \ln h(\delta_1)^{-1}=(l_1^{2/3})^\alpha\ll l_1$ and the fact that the set $B_1(y)\setminus (\Omega_0^{(i')}+y)$  is $\theta$-type-$(\theta^*,E^*,g^{(2)}(\widehat{\delta}_0))$-$0$-nonresonant. We refer   to \cite[Lemma  2.12]{CSZ25}   for   details. 
\end{proof} 

With Claim \ref{1goc}, Proposition \ref{1go} follows from a standard resolvent identity argument, using the fact  that $\ln\delta^{-1}\ll l_1$. We refer   to \cite[Proposition 2.13]{CSZ25}   for  details. 
\end{proof}

\subsubsection{Transversality }\label{trans}
 To prove Anderson localization, we need to eliminate infinitely often double resonance phases. The double resonance phases  at the initial step can be defined by 
 $${\rm DR}_0:=\bigcup_E\left \{\theta\in \T :\ \exists   x\neq y,\  {\rm s.t.}, \ x,y\in \widetilde{S}_0(\theta,E,\delta_0)\bigcap Q_{|\ln\delta_0| }\right \}. $$
 Note that $${\rm DR}_0\subset \left \{\theta\in \T :\ \exists   x\neq y \in Q_{|\ln\delta_0|},\  {\rm s.t.},\ |v(\theta+x\cdot\omega)-v(\theta+y\cdot\omega)|\lesssim\delta_0\right \}.$$ Then one can prove  $\operatorname{Leb}({\rm DR}_0)\leq \delta_0^{1/C}$ by the transversality of $v$ and $\omega\in 	{\rm DC}_{\tau, \gamma}$, where $\operatorname{Leb}(\cdot)$ is the Lebesgue measure of a set.  
At the first inductive  step, we will partition the energy interval into several smaller subintervals and control the double resonance  phases relevant to each subinterval. 
The $E^*$-relevant double resonance phases  at the first step can be defined by 
$${\rm DR}_1(E^*):=\bigcup_{E\in D(E^*,h(\delta_0))}\left \{\theta\in \T :\ \exists   x\neq y, \  {\rm s.t.}, \ x,y\in \widetilde{S}_1(\theta,E,\delta_1)\bigcap Q_{l_1^{30} }\right \}. $$
The goal of this section is to prove Proposition \ref{zhizhu}, which establishes quantitative transversality estimates for the resultant of  Rellich functions. This property is a crucial ingredient  that allows us to control $\operatorname{Leb}({\rm DR}_1(E^*))$ (see Lemma \ref{transha} below).

We view the $E$-polynomial  \begin{equation}\label{a}
P_1^{(i)}(\theta,E):=	\prod_{j=1}^{|\Omega_0^{(i)}|}\left(E-E_1^{(i,j)}(\theta) \right)=:\sum_{k=0}^{|\Omega_0^{(i)}|} a_{k}^{(i)}(\theta)E^{|\Omega_0^{(i)}|-k}
\end{equation}
as a perturbation of 
\begin{equation}\label{b}
Q_1^{(i)}(\theta,E)=	\prod_{x\in \Omega_0^{(i)}}\left(E-v(\theta+x\cdot\omega)\right) =:\sum_{k=0}^{|\Omega_0^{(i)}|} b_{k}^{(i)}(\theta)E^{|\Omega_0^{(i)}|-k}.
\end{equation}
  First,  we prove  the  following approximation Proposition \ref{coe}, which   will be used to handle transversality estimates  later. 
\begin{prop}\label{coe}
	The polynomial coefficients $a_k^{(i)}(\theta),b_k^{(i)}(\theta)$ are analytic functions in  
	$\theta\in {\rm Qdp} (\widetilde{\mathcal{N}}_0^{(i)})$. Moreover, for 
	$\theta\in {\rm Db}(\widetilde{\mathcal{N}}_0^{(i)}):=D_\T(\Re(\theta_0^{(i)}(E^*)),2\widehat{\delta}_0)$, we have 
		$$|\partial_\theta^l(a_k^{(i)}(\theta)-b_k^{(i)}(\theta))|\leq l! \widehat{\delta}_0^{-l}\varepsilon^\frac{1}{2}, \quad \forall  \theta\in {\rm Db}(\widetilde{\mathcal{N}}_0^{(i)}),  0\leq k\leq |\Omega_0^{(i)}|, l\geq 0. $$ 
\end{prop}
\begin{proof}
The proof is based on the Weierstrass preparation argument. 	Choose a circle $\Gamma_0:=\{E\in \C :\ |E-E^*|= g(\widehat{\delta}_0)/2\}$. Recall Proposition \ref{se} and Remark \ref{kuoda} that  for $\theta\in {\rm Qdp}(\widetilde{\mathcal{N}}_0^{(i)})$,   $\{E_1^{(i,j)}(\theta) \}_{j=1}^{ |\Omega_0^{(i)}|}$ are the $E$-roots of $s_{\Omega_0^{(i)}}(\theta,E)$ in  $ D(E^*,\widehat{\delta}_0^\frac{1}{2})$ and $s_{\Omega_0^{(i)}}(\theta,E)$ has no other root in $E\in D(E^*,g(\widehat{\delta}_0))$. Thus for every  $k\geq 0$, we have 
	$$ \widetilde{a}^{(i)}_k(\theta):=\sum_{j=1}^{|\Omega_0^{(i)}|}\left(E_1^{(i,j)}(\theta)\right)^k=\frac{1}{2 \pi i} \oint_{\Gamma_0}E^k\frac{ \partial_E s_{\Omega_0^{(i)}}(\theta,E)}{s_{\Omega_0^{(i)}}(\theta,E)} d E.$$
	Note that  $a_k^{(i)}(\theta )$, which are  the elementary symmetric polynomials of    $\{E_1^{(i,j)}(\theta) \}_{j=1}^{ |\Omega_0^{(i)}|}$, 
	are polynomials in the $ \widetilde{a}^{(i)}_k(\theta)$  and the coefficients of the polynomials are rational
	numbers: 
	\begin{equation}\label{936}
		a_k^{(i)}=\frac{(-1)^k}{k!} \operatorname{det}\left(\begin{array}{ccccc}
			\widetilde{a}^{(i)}_1 & 1 & 0 & 0 & \ldots \\
			\widetilde{a}^{(i)}_2 & \widetilde{a}^{(i)}_1 & 2 & 0 & \ldots \\
			\widetilde{a}^{(i)}_3 & \widetilde{a}^{(i)}_2 & \widetilde{a}^{(i)}_1 & 3 & \ldots \\
			\vdots & \vdots & \vdots & \vdots & \ddots \\
			\widetilde{a}^{(i)}_k & \widetilde{a}^{(i)}_{k-1} & \widetilde{a}^{(i)}_{k-2} & \widetilde{a}^{(i)}_{k-3} & \ldots
		\end{array}\right).
	\end{equation} Analyticity of $ \widetilde{a}^{(i)}_k(\theta)$ and   $a^{(i)}_k(\theta)$ follows from  standard arguments. Similarly,
	$$ \widetilde{b}_k^{(i)}(\theta):=	\sum_{x\in \Omega_0^{(i)}}\left(v(\theta+x\cdot\omega)\right)^k=\frac{1}{2 \pi i} \oint_{\Gamma_0}E^k\frac{ \partial_E  Q_1^{(i)}(\theta,E)}{Q_1^{(i)}(\theta,E)} d E$$ and  $b_k^{(i)}(\theta )$ are polynomials in the $ \widetilde{b}^{(i)}_k(\theta)$ similar to  \eqref{936}. 
Note  that  the coefficients $\widetilde{a}^{(i)},\widetilde{b}^{(i)}$  are bounded above by $M(C_v+1)^M$ since the eigenvalues $E_1^{(i,j)}(\theta)$ and $v(\theta)$ are bounded above by $C_v+1$.  Thus by \eqref{936} and Remark \ref{dif},     to prove \begin{equation}\label{1jie}
		 |a_k^{(i)}(\theta)-b_k^{(i)}(\theta)|\leq \varepsilon^{\frac{1}{2}}, \quad \forall\theta\in{\rm Qdp}(\widetilde{\mathcal{N}}_0^{(i)}),
	\end{equation}    it suffices to prove $ |\widetilde{a}^{(i)}_k(\theta)- \widetilde{b}_k^{(i)}(\theta)|\lesssim \varepsilon^{\frac{2}{3}}$ for all $0\leq k\leq |\Omega_0^{(i)}|$.  This will follow  from 	
	\begin{equation}\label{4242}
		\left|\frac{ \partial_E s_{\Omega_0^{(i)}}(\theta,E)}{s_{\Omega_0^{(i)}}(\theta,E)}-\frac{ \partial_E Q_1^{(i)}(\theta,E)}{Q_1^{(i)}(\theta,E)}\right|\leq \varepsilon^{\frac{2}{3}},\quad  \forall (\theta,E)\in{\rm Qdp}(\widetilde{\mathcal{N}}_0^{(i)})\times \Gamma_0.
	\end{equation}
	Recalling \eqref{1712},   the denominators $s_{\Omega_0^{(i)}}(\theta,E)$ and $Q_1^{(i)}(\theta,E)$ are close with difference $\lesssim\varepsilon$ for  $(\theta,E)\in  {\rm Qdp}(\widetilde{\mathcal{N}}_0^{(i)})\times D(E^*,g(\widehat{\delta}_0))$.  They are also bounded away from zero with modulus $\gtrsim g(\widehat{\delta}_0)^M\gg\varepsilon^\frac{1}{9}$  for $(\theta,E)\in {\rm Qdp}(\widetilde{\mathcal{N}}_0^{(i)})\times \Gamma_0$. Since the denominators are close and bounded away from zero, it suffices to show that the numerators are close. But this follows easily from the Cauchy integral formula:
	$$  \partial_E(s_{\Omega_0^{(i)}}-Q_1^{(i)}) (\theta,E)
	=\frac{1}{2 \pi i} \oint_{|w|=g(\widehat{\delta}_0)/2} \frac{(s_{\Omega_0^{(i)}}-Q_1^{(i)})(\theta,E+w)}{w^2} d w,
	$$
	which shows that the left-hand side is bounded above by $C\varepsilon/g(\widehat{\delta}_0) \ll\varepsilon^{\frac{8}{9}} $. Thus \eqref{4242} is proved.   
Now the  derivative estimates in  $\theta \in {\rm Db}(\widetilde{\mathcal{N}}_0^{(i)})$ follow  from  \eqref{1jie} and  again the Cauchy integral formula
	$$  \partial_\theta^l (a_k^{(i)}-b_k^{(i)})(\theta)
	=\frac{l!}{2 \pi i} \oint_{|w|=\widehat{\delta}_0} \frac{(a_k^{(i)}-b_k^{(i)})(\theta+w)}{w^{l+1}} d w.
	$$
\end{proof}
In order to handle the transversality of the product of functions, we will apply the following Lemma \ref{eli} from \cite{Eli97}. The version presented here is an inconsequential variation of  Eliasson's  statement, adapted for our applications.
\begin{lem}[{\cite[Lemma 4]{Eli97}}]\label{eli} 
	Let  $u_j$ be a sequence of smooth functions defined on an open interval $\theta\in I$ and satisfying
	\begin{align*}
		 \sup_{\theta\in I}\left|\partial_\theta^lu_{j}(\theta)\right| &\leq C_1,  \quad \forall l \leq  N:=m_1+\cdots+m_J,\\
	 \max _{0 \leq l \leq m_j}\left| \partial^l_\theta  u_j(\theta)\right| 	&\geq \beta,  \quad\ \  \forall \theta \in I,
	\end{align*}

	with $C_1 \geq 1$.	If $u=u_1 \cdots  u_J$,  then
	\begin{align*}
		\max_{0 \leq l \leq N}\left| \partial^l_\theta u(\theta)\right| &\geq\left(\frac{1}{e}\right)^{J^{8 N}}\left(\frac{\beta}{C_1^2}\right)^{J^{N+1}},  \quad \forall \theta \in I .	
	\end{align*}

\end{lem}
	For  $\theta\in \widetilde{\mathcal{N}}_0^{(i)}\cap\T$ and $\phi\in \T$ such that $\theta+\phi\in \widetilde{\mathcal{N}}_0^{(i')}\cap\T$, we define the   $E$-resultant of $ P_1^{(i)}(\theta,E)$ and  $ P_1^{(i')}(\theta+\phi,E)$ by  
\begin{align*}
	R_1^{(i,i')}(\theta,\phi)&:= \operatorname{Res}_E(P_1^{(i)}(\theta,E),P_1^{(i')}(\theta+\phi,E))\\
	&:=\prod_{j,j'}(E_1^{(i,j)}(\theta)-E_1^{(i',j')}(\theta+\phi)).
\end{align*}
The resultant $	R_1^{(i,i')}(\theta,\phi)$ can also be expressed explicitly in terms of the coefficients 
\begin{equation}\label{resx}
R_1^{(i,i')}(\theta,\phi)	\operatorname{det}\left(\begin{array}{cccccc}
	1 & a_{1}^{(i)}(\theta) & \cdots & a_{|\Omega_0^{(i)}|}^{(i)}(\theta) & & \\
	& \ddots & \ddots & \ddots & \ddots & \\
	& & 1 & a_{1}^{(i)}(\theta) & \cdots & a_{|\Omega_0^{(i)}|}^{(i)}(\theta) \\
	1 & a_{1}^{(i')}(\theta+\phi) & \cdots & a_{|\Omega_0^{(i')}|}^{(i')}(\theta+\phi) & & \\
	& \ddots & \ddots & \ddots & \ddots & \\
	& & 1 & a_{1}^{(i')}(\theta+\phi) & \cdots & a_{|\Omega_0^{(i')}|}^{(i')}(\theta+\phi)
\end{array}\right). 
\end{equation}
\begin{prop}\label{zhizhu}
Let $\theta, \phi$ and  	$R_1^{(i,i')}(\theta,\phi)$ be described as above.  We have   derivative	estimates:

	\item[\textbf{(1)}] Upper bound:  \begin{equation*}
		|\partial_\theta^lR_1^{(i,i')}(\theta,\phi)|\leq  C l! \widehat{\delta}_0^{-l}, \quad \forall l\geq 0,
	\end{equation*}  with $C=(M(2+2C_v^M))^{2M}. $ 
	\item[\textbf{(2)}] Lower  bound: 
	\begin{itemize}
		\item 
			If $i\neq i'$, then 
\begin{equation}\label{2101}
	 \max _{0 \leq l \leq s M^2}| \partial^l_\theta 	R_1^{(i,i')}(\theta,\phi)| 	\geq \delta_0.
\end{equation}
		\item    If $i=i'$, then 
\begin{equation}\label{2102}
		 \max _{0 \leq l \leq s M^2}| \partial^l_\theta 	R_1^{(i,i)}(\theta,\phi)| 	\geq \delta_0\|\phi\|_\T^{|\Omega_0^{(i)}|}.
\end{equation}
	\end{itemize}
\end{prop}

	If moreover $|\mathcal{C}_0^{(i)}|=|\mathcal{C}_0^{(i')}|=1$, we  have: 
		\begin{itemize}
		\item 
		If $i\neq i'$, then 
	\begin{equation}\label{2103}
			 \max _{0 \leq l \leq s }| \partial^l_\theta 	R_1^{(i,i')}(\theta,\phi)| 	\geq \delta_0.
	\end{equation}
		\item    If $i=i'$, then 
\begin{equation}\label{2104}
		\max _{0 \leq l \leq s }| \partial^l_\theta 	R_1^{(i,i)}(\theta,\phi)| 	\geq \delta_0\|\phi\|_\T.
\end{equation}
	\end{itemize}

\begin{proof} 
	First, we prove \textbf{(1)}. By \eqref{resx}, for such fixed  $\theta$ and $\phi$,   $R_1^{(i,i')}(w,\phi)$ is an analytic function in $w\in  D_\T(\theta,\widehat{\delta}_0)$ since the coefficients are analytic by  Proposition \ref{coe}.  Since $a_k^{(i)}(\theta)$ are  the elementary symmetric polynomials of eigenvalues  $\{E_1^{(i,j)}(\theta)\}_{j=1}^{ |\Omega_0^{(i)}|}$,  we have $$\sup_{w\in D_\T(\theta,\widehat{\delta}_0)} |a_k^{(i)}(w )  |\leq2^M (1+C_v)^M. $$
	Applying   Hadamard inequality to \eqref{resx}, we obtain 
	$$\sup_{w\in D_\T(\theta,\widehat{\delta}_0)}| R_1^{(i,i')}(w,\phi)|\leq (M(2+2C_v)^M)^{2M}.$$
Bounding  the derivatives by  the Cauchy integral formula	yields   \textbf{(1)}. 
	
	For the proof of  \textbf{(2)}, define the   $E$-resultant of $ Q_1^{(i)}(\theta,E)$ and $ Q_1^{(i')}(\theta+\phi,E)$ by 
\begin{align}
\nonumber	\widetilde{R}_1^{(i,i')}(\theta,\phi):&= \operatorname{Res}_E(Q_1^{(i)}(\theta,E),Q_1^{(i')}(\theta+\phi,E))\\
	&=\prod_{x\in \Omega_0^{(i)},y\in \Omega_0^{(i')}}(v(\theta+x\cdot\omega)-v(\theta+\phi+y\cdot\omega)). \label{zhaodaole}
\end{align}
We will first establish  lower bound estimates  for $	\widetilde{R}_1^{(i,i')}(\theta,\phi)$, and then show  $|\partial^l_\theta( R_1^{(i,i')}-	\widetilde{R}_1^{(i,i')})(\theta,\phi)|$ is small. Thus the  estimates also hold for $R_1^{(i,i')}(\theta,\phi) $. 
\begin{itemize}
	\item  $i\neq i'$: In this case, $\theta_0^{(i)}(E^*)$ and $\theta_0^{(i')}(E^*)$ belong to different equivalence classes. Thus by \eqref{fenk} and \eqref{??1},  we have    \begin{equation}\label{133}
		\|\theta_0^{(i)}(E^*)+z\cdot\omega-\theta_0^{(i')}(E^*)\|_\T>g^{(M^4)}(\bar{\delta}_0),  \quad \forall z\in Q_{l_1^{\alpha}}.  
	\end{equation}   Since $\theta\in \widetilde{\mathcal{N}}_0^{(i)}\cap\T$ and $\theta+\phi\in \widetilde{\mathcal{N}}_0^{(i')}\cap\T$, we have \begin{equation}\label{14}
		\|\theta_0^{(i)}(E^*)+\phi-\theta_0^{(i')}(E^*)\|_\T\leq 4  \widehat{\delta}_0. 
	\end{equation}
Since  $\Omega_0^{(i)},\Omega_0^{(i')}\subset Q_{l_1^{1/\alpha}}$ and  $g^{(M^4)}(\bar{\delta}_0)\gg  \widehat{\delta}_0$,  \eqref{133} together with \eqref{14} implies 
\begin{equation}\label{faf}
	\|\phi+(y-x)\cdot \omega\|_\T>\widehat{\delta}_0, \quad \forall  x\in \Omega_0^{(i)},y\in \Omega_0^{(i')}. 
\end{equation} Thus by Condition \ref{con2}, for any $x\in \Omega_0^{(i)},y\in \Omega_0^{(i')}$,  we have 
\begin{equation}\label{2120}
	\max _{0 \leq l \leq s}\left|\partial_\theta^l(v(\theta+x\cdot\omega)-v(\theta+\phi+y\cdot\omega))\right| \geq c \|\phi+(y-x)\cdot\omega\|_\T\geq c  \widehat{\delta}_0.
\end{equation}
 Bounding  the derivatives by  the Cauchy integral formula, we have 
\begin{equation}\label{2121}
	|\partial_\theta^l(v(\theta+x\cdot\omega)-v(\theta+\phi+y\cdot\omega))|\leq 2 C_vl!\widehat{\delta}_0^{-l}, \quad  \forall l \geq 0. 
\end{equation}
Applying  Lemma \ref{eli} to \eqref{zhaodaole} with     $ J=|\Omega_0^{(i)}|\cdot|\Omega_0^{(i')}|\leq M^2$, $u_j(\theta)= v(\theta+x\cdot\omega)-v(\theta+\phi+y\cdot\omega)$, $m_j=s$, $N=s\cdot |\Omega^{i}|\cdot|\Omega_0^{(i')}| \leq sM^2$, $C_1=C_v(sM^2)!\widehat{\delta}_0^{-sM^2}$ and  $\beta= c  \widehat{\delta}_0$, we obtain 
	\begin{align}
		\nonumber \max _{0 \leq l \leq s M^2}| \partial^l_\theta 		\widetilde{R}_1^{(i,i')}(\theta,\phi)| &	\geq  \left(\frac{1}{e}\right)^{M^{16sM^2}} \left(\frac{c  \widehat{\delta}_0}{(2C_v(sM^2)!\widehat{\delta}_0^{-sM^2})^2}\right)^{M^{2(sM^2+1)}} \\
		  & \geq C \widehat{\delta}_0^{C}    \geq    2\delta_0. \label{11}
	\end{align}
In \eqref{11}, we used the fact that   $\widehat{\delta}_0\geq g(\delta_0) =e^{-|\ln\delta_0|^\frac{1}{\alpha}}\gg \delta_0^{1/C}$  for  $\delta_0\ll1$. 
 Note that $\widetilde{R}_1^{(i,i')}$ can be  expressed explicitly by \eqref{resx} with $a_k^{(i)}$ replaced by  $b_k^{(i)}$ and $|a_k^{(i)}(w)-b_k^{(i)}(w)|\leq \varepsilon^{\frac{1}{2}}$ in $w\in {\rm Db}(\widetilde{\mathcal{N}}_0^{(i)})$ by Proposition \ref{coe}.   Using Remark \ref{dif}, we have  \begin{equation}\label{1018}
 	|R_1^{(i,i')}(w,\phi)-\widetilde{R}_1^{(i,i')}(w,\phi)|\leq C \varepsilon^{\frac{1}{2}}, \quad \forall w\in   D_\T(\theta,\widehat{\delta}_0) .
 \end{equation}
 Bounding  the derivatives by  the Cauchy integral formula yields 
\begin{equation*}
	 |\partial^l_\theta (R_1^{(i,i')}(\theta,\phi)-\widetilde{R}_1^{(i,i')}(\theta,\phi))|\leq C  l! \widehat{\delta}_0^{-l} \varepsilon^{\frac{1}{2}}, \quad \forall l\geq 0.
\end{equation*}
 In particular, 
\begin{equation}\label{33}
		 \max _{0 \leq l \leq s M^2}| \partial^l_\theta  (R_1^{(i,i')}(\theta,\phi)-\widetilde{R}_1^{(i,i')}(\theta,\phi))| \leq C  (s M^2)! \widehat{\delta}_0^{-s M^2} \varepsilon^{\frac{1}{2}}\leq \delta_0.
\end{equation}
 \eqref{2101} follows from \eqref{11} and \eqref{33}.

If moreover  $|\mathcal{C}_0^{(i)}|=|\mathcal{C}_0^{(i')}|=1$, then $\Omega_0^{(i)}=\Omega_0^{(i')}=\{o\}$ and  $\widetilde{R}_1^{(i,i')}(\theta,\phi)=v(\theta)-v(\theta+\phi)$. \eqref{2103} follows from  \eqref{2120} and \eqref{33}.
	\item  $i= i'$:  In this case, $\|\phi\|_\T\leq 2\widehat{\delta}_0$.  Thus for $x\neq y\in \Omega_0^{(i)}$,  we have 	\begin{equation}\label{2119}
		\|\phi+(y-x)\cdot \omega\|_\T\geq \|(y-x)\cdot\omega \|_\T- \|\phi\|_\T \geq \widehat{\delta}_0.
	\end{equation} Assume $ \|\phi\|_\T\neq 0$.  Write   \begin{align}
\|\phi\|_\T^{-|\Omega_0^{(i)}|}\widetilde{R}_1^{(i,i)}(\theta,\phi) =  \prod_{\substack{ x,y \in \Omega_0^{(i)}\\  x\neq   y }}(v(\theta+x\cdot\omega)-v(\theta+\phi+y\cdot\omega))\label{944} \\
\times 		  \prod_{x\in  \Omega_0^{(i)}} \left( \|\phi\|_\T^{-1}(v(\theta+x\cdot\omega)-v(\theta+\phi+x\cdot\omega))\right) .\label{945}
	\end{align}
By \eqref{2119}, the derivatives of  the  factors  in the first product (i.e., the right hand side of \eqref{944})  can be estimated by  similar arguments to those  of   \eqref{2120} and \eqref{2121}.  For the factors  in the second  product (i.e., \eqref{945}),  by Condition \ref{con2}, we have  
\begin{equation}\label{xa}
		\max _{0 \leq l \leq s}|\|\phi\|_\T^{-1}\partial^l_\theta (v(\theta+x\cdot\omega)-v(\theta+\phi+x\cdot\omega))|\geq c,
\end{equation}
and by mean value inequality with $\theta$ and $\theta+\phi$ in the same interval $\widetilde{\mathcal{N}}_0^{(i)}\cap\T$, we have  
\begin{align*}
		&|\|\phi\|_\T^{-1}\partial_\theta^l   (v(\theta+x\cdot\omega)-v(\theta+\phi+x\cdot\omega))
		 | \\ \leq & |\partial_\theta^{l+1}  v(\theta+\phi'+x\cdot\omega)
		 |\\  \leq & C_v(l+1)!\widehat{\delta}_0^{-(l+1)}, \quad  \forall l \geq 0. 
\end{align*}
Using the  factors estimates above and applying  Lemma \ref{eli} to \eqref{944} and  \eqref{945} yields
	\begin{equation}\label{111}
 \|\phi\|_\T^{-|\Omega_0^{(i)}|} \max _{0 \leq l \leq s M^2}|  \partial^l_\theta 		\widetilde{R}_1^{(i,i)}(\theta,\phi)|  \geq    2\delta_0. 
\end{equation} In order to bound   $|\partial^l_\theta (R_1^{(i,i)}(\theta,\phi)-\widetilde{R}_1^{(i,i)}(\theta,\phi))|$, by the  Cauchy integral  formula, it suffices to 
to bound $|R_1^{(i,i)}(w,\phi)-\widetilde{R}_1^{(i,i)}(w,\phi)| $ for $w\in   D_\T(\theta,\widehat{\delta}_0)$. 
Let $c_{k}^{(i)}(w,\phi):=  a_{k}^{(i)}(w+\phi)-a_{k}^{(i)}(w)$ and  $d_{k}^{(i)}(w,\phi):=  b_{k}^{(i)}(w+\phi)-b_{k}^{(i)}(w)$.  By mean value inequality, we have 
\begin{align}
\label{qus}	|c_{k}^{(i)}(w,\phi)| , \  |d_{k}^{(i)}(w,\phi)|\leq C \widehat{\delta}_0^{-1}\|\phi\|_\T , \\
\label{qusx}|c_{k}^{(i)}(w,\phi)-d_{k}^{(i)}(w,\phi)|\leq  \widehat{\delta}_0^{-1}\varepsilon^\frac{1}{2}\|\phi\|_\T 
\end{align}since 
$|\partial_wa_k^{(i)}(w)|, |\partial_wb_k^{(i)}(w)|\leq  C  \widehat{\delta}_0^{-1}$ and  
 $|\partial_w(a_k^{(i)}(w)-b_k^{(i)}(w))|\leq \widehat{\delta}_0^{-1}\varepsilon^\frac{1}{2}$ for $w\in {\rm Db}(\widetilde{\mathcal{N}}_0^{(i)})$ by Proposition \ref{coe}.  Now we  can rewrite the resultant  in \eqref{resx} as 
\begin{equation}\label{A}
	R_1^{(i,i)}(w,\phi)=\operatorname{det}\left(\begin{array}{cccccc}
		1 & a_{1}^{(i)}(w) & \cdots & a_{|\Omega_0^{(i)}|}^{(i)}(w) & & \\
		& \ddots & \ddots & \ddots & \ddots & \\
		& & 1 & a_{1}^{(i)}(w) & \cdots & a_{|\Omega_0^{(i)}|}^{(i)}(w) \\
		0 & c_{1}^{(i)}(w,\phi) & \cdots & c_{|\Omega_0^{(i)}|}^{(i)}(w,\phi) & & \\
		& \ddots & \ddots & \ddots & \ddots & \\
		& & 0 & c_{1}^{(i)}(w,\phi) & \cdots & c_{|\Omega_0^{(i)}|}^{(i)}(w,\phi)
	\end{array}\right),
\end{equation}
 \begin{equation}\label{B}
 \widetilde{R}_1^{(i,i)}(w,\phi)=\operatorname{det}\left(\begin{array}{cccccc}
 		1 & b_{1}^{(i)}(w) & \cdots & b_{|\Omega_0^{(i)}|}^{(i)}(w) & & \\
 		& \ddots & \ddots & \ddots & \ddots & \\
 		& & 1 & b_{1}^{(i)}(w) & \cdots & b_{|\Omega_0^{(i)}|}^{(i)}(w) \\
 		0 & d_{1}^{(i)}(w,\phi) & \cdots & d_{|\Omega_0^{(i)}|}^{(i)}(w,\phi) & & \\
 		& \ddots & \ddots & \ddots & \ddots & \\
 		& & 0 & d_{1}^{(i)}(w,\phi) & \cdots & d_{|\Omega_0^{(i)}|}^{(i)}(w,\phi)
 	\end{array}\right).
 \end{equation}
By \eqref{qus}, \eqref{A} and \eqref{B}, we have \begin{equation}\label{ys}
|	R_1^{(i,i)}(w,\phi)|, \ |\widetilde{R}_1^{(i,i)}(w,\phi)|\leq  C(\widehat{\delta}_0^{-1}\|\phi\|_\T)^{|\Omega_0^{(i)}|}.
\end{equation}
Denote the matrix in \eqref{A} by $A(w,\phi)$ and \eqref{B} by  $B(w,\phi)$. We define the function $d(t,w,\phi):=\operatorname{det}(A(w,\phi)+t(B(w,\phi)-A(w,\phi)), t \in [0,1]$. By \eqref{qus}, \eqref{qusx} and  Proposition \ref{coe}, we can see from the  derivative rule for determinants  that $|\partial_td(t,w,\phi)|\lesssim   \widehat{\delta}_0^{-M}\varepsilon^\frac{1}{2}\|\phi\|_\T^{|\Omega_0^{(i)}|}$. 
Thus $$|R_1^{(i,i)}(w,\phi)-\widetilde{R}_1^{(i,i)}(w,\phi)|=|d(1,w,\phi)-d(0,w,\phi)|\lesssim  \widehat{\delta}_0^{-M}\varepsilon^\frac{1}{2}\|\phi\|_\T^{|\Omega_0^{(i)}|},  $$
which implies  \begin{equation}\label{333}
 \|\phi\|_\T^{-|\Omega_0^{(i)}|}\max _{0 \leq l \leq s M^2}|	 \partial^l_\theta  (R_1^{(i,i)}(\theta,\phi)-\widetilde{R}_1^{(i,i)}(\theta,\phi))| \leq C  (s M^2)! \widehat{\delta}_0^{-2s M^2} \varepsilon^{\frac{1}{2}} \leq \delta_0 .
\end{equation} 
\eqref{2102} follows from  \eqref{111} and \eqref{333}.

If moreover $|\mathcal{C}_0^{(i)}|=1$, then $\Omega_0^{(i)}=\{o\}$ and  $\widetilde{R}_1^{(i,i)}(\theta,\phi)=v(\theta)-v(\theta+\phi)$. \eqref{2104} follows from  \eqref{xa} and \eqref{333}.

\end{itemize}
\end{proof}

\subsection{Statement of Inductive Hypotheses}\label{HP}
At the first inductive step, from  Proposition \ref{654},  we obtained a set of numbers   $\bigcup_{i=1}^{\kappa_0}\{\theta_{1}^{(i,j)}(E^*)\}_{j=1}^{|\mathcal{C}_0^{(i)}|}$, 
where for each $i$, $\{\theta_{1}^{(i,j)}(E^*)\}_{j=1}^{|\mathcal{C}_0^{(i)}|}$   are the roots of $\operatorname{det}(H_{B_1}(\theta)-E^*)$ in $\widetilde{\mathcal{N}}^{(i)}_0$, a neighborhood of the representative $\theta_0^{(i)}(E^*)$ of $\mathcal{C}_0^{(i)}$. And $ S_1(\theta^*,E^*,\delta_1)$, the resonant points   set  at the first step,  can be defined via these roots. We order the numbers  $\bigcup_{i=1}^{\kappa_0}\{\theta_{1}^{(i,j)}(E^*)\}_{j=1}^{|\mathcal{C}_0^{(i)}|}$ as  $\{\theta_{1}^{(i)}(E^*)\}_{i=1}^{m_1}$, where $m_1= \sum _{i=1}^{\kappa_0}|\mathcal{C}_0^{(i)}| $. 
We treat  $\{\theta_{1}^{(i)}(E^*)\}_{i=1}^{m_1}$ similar to $\{\theta_{0}^{(i)}(E^*)\}_{i=1}^{m_0}$ and  $ S_1(\theta^*,E^*,\delta_1)$ similar to  $ S_0(\theta^*,E^*,\delta_0)$.  Instead of elaborating on the details for this next step, we explain the general inductive step. First, we introduce some inductive definitions.
\begin{defn}
	Recall  $\delta_0=g^{(2)}(\varepsilon_0)$, $l_1\in [|\ln\delta_0|^4,|\ln\delta_0|^8]$,  $B_1=Q_{l_1}$,  $\delta_1=e^{-l_1^{2/3}}  $,  $\gamma_0=\frac{1}{2}|\ln\varepsilon|$ defined in \S \ref{n=1}.  We introduce  the inductive parameters for $n\geq 1$:
	\begin{itemize}
		\item $l_n$ satisfies $l_n\in [l_{n-1}^4, l_{n-1}^8]$.  \hfill	(Length)
		\item $ B_n$ satisfies  $	Q_{l_n}\subset B_n\subset Q_{l_n+50M^2l_{n-1}}$.  \hfill	(Block)
		\item $\delta_n:=e^{-l_n^{2/3}}$. \hfill (Resonance)
		\item $\gamma_n:=\prod\limits_{k=1}^{n}(1-l_k^{-(\alpha-1)^8})\gamma_0$. \hfill (Green's function off-diagonal decay rate)
	\end{itemize}      
	We  introduce  three auxiliary resonance parameters $\bar{\delta}_n$,  $\widehat{\delta}_n$, $\widetilde{\delta}_n$, which will be  specified  in the proof $$\bar{\delta}_n\in [g(\delta_n), g^{(4M^8+1)}(\delta_n)],\ \  \ \widehat{\delta}_n\in [g^{(M)}(\bar{\delta}_n),g^{(M^3)}(\bar{\delta}_n)],\ \ \ \widetilde{\delta}_n\in [\widehat{\delta}_n^{\frac{1}{4}}, g(\widehat{\delta}_n )] .$$
	We assume  $\varepsilon$ is  sufficiently small so that  	$	\gamma_n\geq  \frac{1}{2}\gamma_0\geq 10$ for all $n$.
	\begin{rem}\label{123123}
	The exact choice of the  parameters depends  on $E^*$ and  will be  specified in the proof. To be rigorous, we may better use the notation $l_n(E^*),B_n(E^*)$ instead of $l_n,B_n$. But  since our analysis in the inductive process  is carried out for each fixed 
	$E^*$, we adopt shorter notations for brevity when no ambiguity arises. However, in the later Section \ref{lc}, we will show that the parameters can be chosen in a locally constant way so that the  $l_n(E^*),B_n(E^*)$ associated with $E^*$   can be used for the resonance analysis  of any $\widetilde{E}^*\in D(E^*,h(\delta_{n-1}(E^*) ))$. 
	\end{rem}
	
\end{defn}

\begin{flushleft}
	\textbf{Inductive Hypotheses at step $n$.} The following inductive  Hypotheses \ref{h1}--\ref{h4} hold  for any  $1\leq r\leq n $: 
\end{flushleft}
\begin{hp}[Root and eigenvalue]\label{h1}
	Fix $E^*$.	There is a set of numbers   $\{\theta_{r}^{(i)}(E^*)\}_{i=1}^{m_r}$, which are the   roots of  $\operatorname{det}(H_{B_r}(\theta)-E^*)$ obtained from the previous step in the following sense:  
	
	In	the previous step,  $\{\theta_{r-1}^{(i)}(E^*)\}_{i=1}^{m_{r-1}}$ can be 
	partitioned by the  equivalence relation
	$$	\theta^{(i)}_{r-1}(E^*) \sim \theta^{(i')}_{r-1}(E^*) \Leftrightarrow \exists x\in Q_{l_r^{1/\alpha}}, \text{ {\rm s.t.}, } \| \theta^{(i)}_{r-1}(E^*)+x\cdot\omega-\theta^{(i')}_{r-1}(E^*)\|_\T\lesssim \bar{\delta}_{r-1} $$ 
	into $\widetilde{m}_{r-1}$ many equivalence classes $\mathcal{C}_{r-1}^{(1)},\cdots,\mathcal{C}_{r-1}^{(\widetilde{m}_{r-1})}$. All  the elements in the first $\kappa_{r-1}$ classes  belong to $\T_{h(\widehat{\delta}_{r-1})}$ and  all  the elements in the other $(\widetilde{m}_{r-1}-\kappa_{r-1})$ classes  don't belong  to $\T_{g^{(M)}(\widehat{\delta}_{r-1})}$. 
	For each  $1\leq i\leq \kappa_{r-1}$, 	the representative of  $\mathcal{C}_{r-1}^{(i)}$ is $\theta_{r-1}^{(i)}(E^*)$ and    we assign  $\theta^{(i)}_{r-1}(E^*)$  to  two  neighborhoods  $$ \mathcal{N}_{r-1}^{(i)}:=D_\T(\Re(\theta_{r-1}^{(i)}(E^*)),\widehat{\delta}_{r-1}^2),\quad {\widetilde{\mathcal{N}}_{r-1}}^{(i)}:=D_\T(\Re(\theta_{r-1}^{(i)}(E^*)),\widehat{\delta}_{r-1}).$$   Then  $\operatorname{det}(H_{B_r}(\theta)-E^*)$ has  $|\mathcal{C}_{r-1}^{(i)}|$ many $\theta$-roots in $ {\rm Db}(\mathcal{N}_{r-1}^{(i)}):=D_\T(\Re(\theta_{r-1}^{(i)}(E^*)),2\widehat{\delta}_{r-1}^2)$, denoted by $\{\theta_{r}^{(i,j)}(E^*)\}_{j=1}^{|\mathcal{C}_{r-1}^{(i)}|}$,  and has no other roots  in the extension ${\widetilde{\mathcal{N}}_{r-1}}^{(i)}$.  The   numbers   $\{\theta_{r}^{(i)}(E^*)\}_{i=1}^{m_r}$ are obtained   by ordering  these roots  $\bigcup_{i=1}^{\kappa_{r-1}}\{\theta_{r}^{(i,j)}(E^*)\}_{j=1}^{|\mathcal{C}_{r-1}^{(i)}|}$. 
	Moreover, for any $\theta\in  {\widetilde{\mathcal{N}}_{r-1}}^{(i)}$, $H_{B_r}(\theta)$ has $\iota_{r-1}^{(i)}(\leq |\mathcal{C}_{r-1}^{(i)}|) $ many eigenvalues in $D(E^*,3\widetilde{\delta}_{r-1}^2)$, denoted by  $\{E_r^{(i,j)}(\theta)\}_{j=1}^{\iota_{r-1}^{(i)}}$,  and has   no other eigenvalue  in the extension $ D(E^*,\widetilde{\delta}_{r-1})$. 		
	In particular, it follows from the above statement  that 
	\begin{equation}\label{geshu}
	\sum_{i=1}^{\widetilde{m}_{r-1}}|\mathcal{C}_{r-1}^{(i)}|=m_{r-1}, \ \  \sum_{i=1}^{\kappa_{r-1}}|\mathcal{C}_{r-1}^{(i)}|=m_{r} \ \ \text{and} \ \  m_{r}\leq m_{r-1}\leq M.
\end{equation}
\end{hp}
Before stating the other  inductive Hypotheses, we introduce the definitions of  resonant points set, regular set and nonresonant set. 
\begin{defn}\label{swan}
	Given $\theta^*$, $E^*$ and  $\delta\in [h(\delta_{r}),\delta_{r-1}]$,  we define     	
	\begin{align*}
		S_r(\theta^*,E^*,\delta):=&\{x\in \Z^d:\ \exists i\in [1,\kappa_{r-1}],  \text{ {\rm s.t.}, }\theta^*+x\cdot \omega \in \widetilde{\mathcal{N}}_{r-1}^{(i)}
		\\\text{ and } &\min_{1\leq j\leq |\mathcal{C}_{r-1}^{(i)}|}\|\theta^*+x\cdot \omega-\theta_r^{(i,j)}(E^*)\|_\T<\delta\}, 
	\end{align*}
	and for  $E\in D(E^*,h(\delta_{r-1}))$, we define 
	\begin{align*}
		\widetilde{S}_r(\theta^*,E,\delta):=&\{x\in \Z^d:\ \exists i\in [1,\kappa_{r-1}],  \text{ {\rm s.t.}, }\theta^*+x\cdot \omega \in \widetilde{\mathcal{N}}_{r-1}^{(i)}
		\\ \text{ and } &\min_{1\leq j\leq \iota_{r-1}^{(i)}}|E_r^{(i,j)}(\theta^*+x\cdot\omega )-E|<\delta\}.
	\end{align*}
	We say  a finite set $\Lambda\subset\Z^d$ is 
	$(\theta^*,E^*)$-$r$-regular if
	$$x\in S_{k-1}(\theta^*,E^*,g(\delta_{k-1}))\cap \Lambda\Rightarrow (Q_{5l_k}+x)\subset \Lambda, \ \  \forall  k\in [1,r]. $$
	We say  $\Lambda$ is  $\theta$-type-$(\theta^*,E^*,\delta)$-$r$-nonresonant   if $\Lambda\cap S_r(\theta^*,E^*,\delta)=\emptyset$. 
	For $E\in D(E^*,h(\delta_{r-1}))$, we say  $\Lambda$  is $E$-type-$(\theta^*,E,\delta)$-$r$-nonresonant   if $\Lambda\cap\widetilde{S}_r(\theta^*,E,\delta)=\emptyset$.
\end{defn}

\begin{hp}[Green's function]\label{h2}
	Let $\theta^*\in \T_{\widehat{\delta}_{r-1}}$, $\delta\in [h(\delta_r),\delta_{r-1}]$  and   $\Lambda$ be   $(\theta^*,E^*)$-$r$-regular. Then we have 
	\begin{itemize}
		\item[\textbf{(1)}] If  $\Lambda$ is  $\theta$-type-$(\theta^*,E^*,\delta)$-$r$-nonresonant,  then  for any $(\theta, E)\in D_\T(\theta^*,h(\delta))\times D(E^*,h(\delta))$, we have 
		\begin{align}
			\|G_\Lambda^{\theta,E}\|&\leq \delta_{r-1}^{-5}\delta^{-M},\label{12z}\\
			|G_\Lambda^{\theta,E}(x,y)|&\leq e^{-\gamma_r\|x-y\|_1}, \quad \forall \|x-y\|_1\geq l_r^\frac{2}{\alpha+1}.\label{13z}
		\end{align}	
		\item[\textbf{(2)}]  Let $E\in D(E^*,h(\delta_{r-1}))$. If  $\Lambda$ is  $E$-type-$(\theta^*,E,\delta)$-$r$-nonresonant, then  \eqref{12z} and  \eqref{13z} hold for $\theta=\theta^*$. 
	\end{itemize}
\end{hp}

\begin{hp}[Block and  Schur complement]\label{h3}
	The block  $B_r$  satisfies  $Q_{l_r}\subset B_r\subset Q_{l_r+50M^2l_{r-1}}$ and  \begin{equation*}\label{shenm}
		\text{$B_r$ is chosen to be  $(\theta,E^*)$-$(r-1)$-regular for any $\theta\in \widetilde{\mathcal{N}}_{r-1}^{(i)}$,  $ i\in [1,\kappa_{r-1}]$.}
	\end{equation*}  Moreover, for each $ i\in [1,\kappa_{r-1}]$,  $B_r$ contains a further subset $A_r^{(i)}\subset \operatorname{Core}(B_r):= Q_{2l_r^{1/\alpha}}$ such that $|A_r^{(i)}|\leq M^r$ and \begin{equation}\label{shenm1}
		\text{$B_r\setminus A_r^{(i)}$ is   $(\theta,E^*)$-$(r-1)$-regular,  $\forall\theta\in \widetilde{\mathcal{N}}_{r-1}^{(i)}$,}
	\end{equation} \begin{equation}\label{shenm2}
		\text{ $B_r\setminus A_r^{(i)}$ is $\theta$-type-$(\theta,E^*,g^{(2)}(\widehat{\delta}_{r-1}))$-$(r-1)$-nonresonant,  $\forall\theta\in \widetilde{\mathcal{N}}_{r-1}^{(i)}$.}
	\end{equation} We denote $ \widehat{B}^{(i)}_r:= B_r\setminus A_r^{(i)} $ and write 
	$$	E-H_{B_r}(\theta)=\begin{pmatrix}
		E-H_{A_r^{(i)}}(\theta) & \Gamma^{{\rm T}} \\
		\Gamma & E-H_{\widehat{B}^{(i)}_r}(\theta)
	\end{pmatrix}.$$
Define  the Schur complement 
	$$S_{A_r^{(i)}}(\theta,E):=	E-H_{A_r^{(i)}}(\theta)-\Gamma^{{\rm T}} G_{\widehat{B}^{(i)}_r}^{\theta,E} \Gamma.$$
	Then $s_{A_r^{(i)}}(\theta,E):=\operatorname{det}S_{A_r^{(i)}}(\theta,E)$ becomes an analytic function of   $(\theta,E)\in \widetilde{\mathcal{N}}_{r-1}^{(i)}\times D(E^*,\widetilde{\delta}_{r-1})$ and we have  
	\begin{align}
		\label{SS}	s_{A_r^{(i)}}(\theta,E^*)&\overset{\delta_{r-1}}{\sim} \prod_{j=1}^{|\mathcal{C}_{r-1}^{(i)}|}\left(\theta-\theta^{(i,j)}_{r}(E^*) \right),\quad \forall \theta\in \widetilde{\mathcal{N}}_{r-1}^{(i)},\\
		\label{SE}		s_{A_r^{(i)}}(\theta,E)&\overset{\delta_{r-1}}{\sim} \prod_{j=1}^{\iota_{r-1}^{(i)}}\left(E-E_r^{(i,j)}(\theta) \right), \ \   \forall \theta\in \widetilde{\mathcal{N}}_{r-1}^{(i)}, \  E \in D(E^*,\widetilde{\delta}_{r-1}) .
	\end{align}
\end{hp}

\begin{hp}[Transversality]\label{h4}
	We define the $E$-polynomial 
	$$	P_r^{(i)}(\theta,E):=	\prod_{j=1}^{\iota_{r-1}^{(i)}}\left(E-E_r^{(i,j)}(\theta) \right).$$
	For  $\theta\in \widetilde{\mathcal{N}}_{r-1}^{(i)}\cap\T$ and $\phi\in \T$ such that $\theta+\phi\in \widetilde{\mathcal{N}}_{r-1}^{(i')}\cap\T$, we define the   $E$-resultant of $ P_r^{(i)}(\theta,E)$ and  $ P_r^{(i')}(\theta+\phi,E)$   by  
	$$		R_r^{(i,i')}(\theta,\phi):= \operatorname{Res}_E(P_r^{(i)}(\theta,E),P_r^{(i')}(\theta+\phi,E)).$$
	Then we have   derivative	estimates: 
	\item[\textbf{(1)}] Upper bound:  \begin{equation}\label{hup}
		|\partial_\theta^lR_r^{(i,i')}(\theta,\phi)|\leq  C l! \widehat{\delta}_{r-1}^{-l}, \quad \forall l\geq 0.
	\end{equation}  
	\item[\textbf{(2)}] Lower  bound: Let $$t_r:=|\{k\in [1,r]:\ \exists i\in [0,\kappa_{k-1}],  \text{ {\rm s.t.},  } |\mathcal{C}_{k-1}^{(i)}|>1 \}|.$$ 
	Then  we have   $$t_r\leq M^2+1+\ln r .$$
		As a corollary,  \begin{equation}\label{cishur}
		s M^{2t_{r}}\leq Cr^C \text{\ \  with $C=C(s,M)$.}
	\end{equation}
	\begin{itemize}
		\item  	If $i\neq i'$, then 
		\begin{equation}\label{lowbu=}
			\max _{0 \leq l \leq s M^{2t_r}}| \partial^l_\theta 	R_r^{(i,i')}(\theta,\phi)| 	\geq \delta_{r-1}.
		\end{equation}
		\item 	If $i=i'$, then 
		\begin{equation}\label{low=}
			\max_{0 \leq l \leq sM^{2t_r}}| \partial^l_\theta 	R_r^{(i,i)}(\theta,\phi)| 	\geq \delta_{r-1}\|\phi\|_\T^{\iota_{r-1}^{(i)}}.
		\end{equation}
	\end{itemize}
	
\end{hp}
Combining the  conclusions of  \S    \ref{n=1}  yields 
\begin{prop}\label{1832}There exists $\varepsilon_0=\varepsilon_0(\eta,C_v,M, \widetilde{C}_v, c, s  ,d,\tau,\gamma)>0$ such that for all $0\leq \varepsilon\leq \varepsilon_0$, the inductive Hypotheses \ref{h1}--\ref{h4} hold  at step $n=1$.
\end{prop}

\subsection{The general inductive step}\label{n=n} In this section, we will prove the following Theorem \ref{key2}: 
\begin{thm}\label{key2}
	There exists $\varepsilon_0=\varepsilon_0(\eta,C_v,M, \widetilde{C}_v, c, s  ,d,\tau,\gamma)>0$ such that for all $0\leq \varepsilon\leq \varepsilon_0$ and $n\geq 1$, if the  inductive Hypotheses  \ref{h1}--\ref{h4} hold  at  step $n$, then they  remain valid  at   step $n+1$.
\end{thm}

\subsubsection{Choice of scale}\label{CN}
\begin{prop}\label{scalen}
	There exist a resonance parameter  $\bar{\delta}_n\in [g(\delta_n), g^{(4M^8+1)}(\delta_n)]$ and a  scale parameter  $l_{n+1}\in [l_n^4,l_n^8]$ such that the following statements hold: For each $\theta_n^{(i)}(E^*)$, there is a set $\Omega_n^{(i)}\subset Q_{l_{n+1}^{1/\alpha}}$ with $|\Omega_n^{(i)}|\leq M$ such that
	\begin{itemize}
		\item[\textbf{(1)}] For each  $x\in \Omega_n^{(i)}$,   there exists $ i'\in [1,m_n]$ such that \begin{equation}\label{tiaon}
			\|\theta_n^{(i)}(E^*)+x\cdot\omega-\theta_n^{(i')}(E^*)\|_\T\leq \bar{\delta}_n.
		\end{equation} 
		Moreover, if  $\theta_n^{(i'')}(E^*)$ does not  satisfy \eqref{tiaon}, then a   larger  separation holds: 	\begin{equation}\label{fenkn}
			\|\theta_n^{(i)}(E^*)+x\cdot\omega-\theta_n^{(i'')}(E^*)\|_\T>  g^{(M^4)}(\bar{\delta}_n).
		\end{equation}
		\item[\textbf{(2)}]   If $x\in Q_{l_{n+1}^{\alpha}}\setminus \Omega_n^{(i)}$, then we have   \begin{equation}\label{fanz}
			\min_{1\leq i'\leq m_n}\|\theta_n^{(i)}(E^*)+x\cdot\omega-\theta_n^{(i')}(E^*)\|_\T>  g^{(M^4)}(\bar{\delta}_n).
		\end{equation}
		\item[\textbf{(3)}]  For a fixed  $ i'\in [1,m_n]$, there is at most one $x\in \Omega_n^{(i)}$ satisfying \eqref{tiaon}.
		\item  [\textbf{(4)}] The following \eqref{classn} defines an equivalence relation on  $\{\theta_n^{(i)}(E^*)\}_{i=1}^{m_n}$, which can be partitioned  into $\widetilde{m}_n$ different equivalence classes $\mathcal{C}_n^{(1)},\cdots,\mathcal{C}_n^{( \widetilde{m}_n)}$. 
		\begin{equation}\label{classn}
			\theta_n^{(i)}(E^*) \sim \theta_n^{(i')}(E^*) \Leftrightarrow \exists x\in Q_{l_{n+1}^{1/\alpha}},  \text{ {\rm s.t.}, } \| \theta_n^{(i)}(E^*)+x\cdot\omega-\theta_n^{(i')}(E^*)\|_\T\leq \bar{\delta}_n. 
		\end{equation}
		\item  [\textbf{(5)}] For any  $x\neq y\in \Omega_n^{(i)}$,  we have $\|x-y\|_1> l_n^\alpha$. 
	\end{itemize} 
\end{prop}
\begin{proof}
	The proofs of \textbf{(1)}--\textbf{(4)}  are the same as those of Proposition \ref{scale} and \ref{clas}, with $0$ replaced by $n$ and $1$ replaced by $n+1$. Now we prove  \textbf{(5)}. Suppose  that  \textbf{(5)} fails.  Then there exist $ i',i''\in [1,m_n]$ and  $x\neq y\in \Omega_n^{(i)}$ with $\|x-y\|_1\leq  l_n^\alpha $ such that
	\begin{equation}\label{laoban}
		\|\theta_n^{(i)}(E^*)+x\cdot\omega-\theta_n^{(i')}(E^*)\|_\T\leq \bar{\delta}_n, \ \ 	\|\theta_n^{(i)}(E^*)+y\cdot\omega-\theta_n^{(i'')}(E^*)\|_\T\leq \bar{\delta}_n.
	\end{equation}
	Recalling Hypothesis \ref{h1}, we assume $\theta_n^{(i')}(E^*)$ is obtained from $\widetilde{\mathcal{N}}_{n-1}^{(i_1)}$, the neighborhood  of the representative $ \theta_{n-1}^{(i_1)}(E^*)\in\mathcal{C}_{n-1}^{(i_1)} $,   and $\theta_n^{(i'')}(E^*)$ is obtained from $\widetilde{\mathcal{N}}_{n-1}^{(i_2)}$, the neighborhood  of the representative $ \theta_{n-1}^{(i_2)}(E^*)\in\mathcal{C}_{n-1}^{(i_2)} $. 
By  $\|\theta_n^{(i')}(E^*) - \theta_{n-1}^{(i_1)}(E^*)\|_\T, \|\theta_n^{(i'')}(E^*) - \theta_{n-1}^{(i_2)}(E^*)\|_\T\leq 2\widehat{\delta}_{n-1}  $ and \eqref{laoban}, we have 
	\begin{equation}\label{fanzheng}
		\|\theta_{n-1}^{(i_1)}(E^*)+(y-x)\cdot\omega-\theta_{n-1}^{(i_2)}(E^*)\|_\T\leq 5\widehat{\delta}_{n-1}< g^{(M^4)}(\bar{\delta}_{n-1}). 
	\end{equation}
	It follows from  \eqref{fanzheng} that 	$\|\theta_{n-1}^{(i_1)}(E^*)-\theta_{n-1}^{(i_2)}(E^*)\|_\T>\|(x-y)\cdot\omega\|_\T- g^{(M^4)}(\bar{\delta}_{n-1})>0 $  since   $\omega\in 	{\rm DC}_{\tau, \gamma}$  and $0\neq \|x-y\|_1\leq l_n^\alpha$.  
	Thus $\theta_{n-1}^{(i_1)}(E^*)\neq \theta_{n-1}^{(i_2)}(E^*)$ must belong to  different equivalence classes. 
	On the other hand, by \eqref{fanz} with $n$ replaced 
	by $n-1$ together with \eqref{fanzheng}, we have  $y-x\in \Omega_{n-1}^{(i_1)}$ and  $\theta_{n-1}^{(i_1)}(E^*)\sim \theta_{n-1}^{(i_2)}(E^*)$.   This is a contradiction.
\end{proof}

\begin{prop}\label{realn}
	There exists some  $\widehat{\delta}_n\in [g^{(M)}(\bar{\delta}_n),g^{(M^3)}(\bar{\delta}_n)]$ such that each equivalence class $\mathcal{C}_n^{(i)}$  satisfies either  \textbf{(1)} or \textbf{(2)}: \begin{itemize}
		\item [\textbf{(1)}] All the elements in $\mathcal{C}_n^{(i)}$ belong to $\T_{h(\widehat{\delta}_n)}$.
		\item [\textbf{(2)}] All the elements in $\mathcal{C}_n^{(i)}$ don't belong  to $\T_{g^{(M)}(\widehat{\delta}_n)}$. 
	\end{itemize} 
\end{prop}

\begin{proof}
	The proof is the same as that of  Proposition \ref{real}.
\end{proof}
We assume  that the  equivalence classes satisfying \textbf{(1)} are the first $\kappa_n$ classes  and the representative of $\mathcal{C}_n^{(i)}$ is $\theta_n^{(i)}(E^*)$ for $1\leq i\leq \kappa_n$.  For each  $1\leq i\leq \kappa_n$,  we  define  $$\mathcal{N}_n^{(i)}:=D_\T(\Re(\theta_n^{(i)}(E^*)),{\widehat{\delta}_n}^2), \quad {\widetilde{\mathcal{N}}}_n^{(i)}:=D_\T(\Re(\theta_n^{(i)}(E^*)),\widehat{\delta}_n). $$
We will carry out  the resonance analysis in each  ${\widetilde{\mathcal{N}}}_n^{(i)}$, $1\leq i\leq \kappa_n$.

\subsubsection{Construction $B_{n+1}$ and $A_{n+1}^{(i)}$}\label{Bn}
The first issue is to construct the $B_{n+1}$ and  $A_{n+1}^{(i)}$ satisfying the inductive Hypothesis \ref{h3}, that is, we  need to construct a block  $B_{n+1}$ such that	$Q_{l_{n+1}}\subset B_{n+1}\subset Q_{l_{n+1}+50M^2l_{n}}$ and for any  $\theta\in \widetilde{\mathcal{N}}_{n}^{(i)}$ and  $ i\in [1,\kappa_{n}]$,    
\begin{equation}\label{r1}
	x\in S_{k-1}(\theta,E^*,g(\delta_{k-1}))\cap B_{n+1}\Rightarrow (Q_{5l_k}+x)\subset B_{n+1}, \ \  \forall k\in [1,n]. 
\end{equation}
Moreover,  for each $ i\in [1,\kappa_{n}]$, we need to construct a set $A_{n+1}^{(i)}$ such that $A_{n+1}^{(i)}\subset Q_{l_{n+1}^{1/\alpha}}$, $|A_{n+1}^{(i)}|\leq M^{n+1}$ and  $ \widehat{B}_{n+1}^{(i)}=:B_{n+1}\setminus A_{n+1}^{(i)}$ satisfying  for any  $\theta\in \widetilde{\mathcal{N}}_{n}^{(i)}$, 
\begin{equation}\label{r2}
	x\in S_{k-1}(\theta,E^*,g(\delta_{k-1}))\cap \widehat{B}_{n+1}^{(i)}\Rightarrow (Q_{5l_k}+x)\subset \widehat{B}_{n+1}^{(i)}, \ \  \forall k\in [1,n],
\end{equation} 
\begin{equation}\label{n1}
	S_{n}(\theta,E^*,g^{(2)}(\widehat{\delta}_{n})) \cap \widehat{B}_{n+1}^{(i)} =\emptyset.
\end{equation}
The construction of $B_{n+1}$ will follow from the following Lemmas \ref{set} and \ref{sep} as well as   the separation of resonant points guaranteed by  $\omega\in 	{\rm DC}_{\tau, \gamma}$. 

\begin{lem}\label{set}
	Let $L,K\in \N^*$.  Assume that   a set $S\subset\Z^d$ satisfies 
	\begin{equation}\label{3211}
		|S\cap( Q_{4LK}+x)|\leq  K , \quad \forall x\in \Z^d.   
	\end{equation}
	Then for any set $\Lambda\subset \Z^d$, there exists an extension  $\widetilde{\Lambda}$  satisfying the following two  properties: 
	\begin{itemize}
		\item $\Lambda \subset \widetilde{\Lambda }\subset \{x\in \Z^d: \ \operatorname{dist}_1(x,\Lambda )\leq 2KL\}$.
		\item    $(Q_L+x)\subset \widetilde{\Lambda }$ for any $x\in S$ such that  $(Q_L+x)\cap \widetilde{\Lambda }\neq \emptyset$. 
	\end{itemize}
\end{lem}
\begin{proof}
	Let $\Lambda_0:=\Lambda$. For $n\geq 0$, 	we inductively define the sequence 
	\begin{equation*}
		\Lambda_{n+1}:=\Lambda_n \bigcup\left(\bigcup_{x\in S:\, (Q_L+x)\cap \Lambda_n \neq \emptyset }(Q_L+x)\right).
	\end{equation*}
	By \eqref{3211}, this sequence  must satisfy for some  $ n_0\leq K$ such that $	\Lambda_{n_0+1}=\Lambda_{n_0}$. Then one can easily check $\widetilde{\Lambda}=\Lambda_{n_0}$ is the desired set. 
\end{proof}
\begin{lem}\label{sep} For any $k\in[1,n]$,  we have 
	$$ 	\left| \left(  \bigcup_{i=1}^{\kappa_{n}} \bigcup_{ \theta\in \widetilde{\mathcal{N}}_{n}^{(i)}}S_{k-1}(\theta,E^*,g(\delta_{k-1}))\right) \cap( Q_{40l_{k}M^2}+x)\right|\leq  M^2 , \quad \forall x\in \Z^d.   $$
\end{lem}
\begin{proof}
	If $y\in    \bigcup_{i=1}^{\kappa_{n}} \bigcup_{ \theta\in \widetilde{\mathcal{N}}_{n}^{(i)}}S_{k-1}(\theta,E^*,g(\delta_{k-1}))$, then there exist $i\in [1,\kappa_n]$ and $i'\in [1,m_{k-1}]$ such that  \begin{equation}\label{shufu}
		\|\Re(\theta_{n}^{(i)}(E^*))+y\cdot\omega-\theta_{k-1}^{(i')}(E^*)\|_\T<g(\delta_{k-1})
		+\widehat{\delta}_n<2g(\delta_{k-1}).
	\end{equation}
	For  any fixed pair $(i,i')$ and $x\in \Z^d$,  there is at most one $y\in  (Q_{40l_{k}M^2}+x)$ satisfying \eqref{shufu} since $\omega\in 	{\rm DC}_{\tau, \gamma}$ and $g(\delta_{k-1})\ll l_{k}^{-1}$. The conclusion follows immediately since there are at most $M^2$ many $(i,i')$-pairs.  
\end{proof}
\begin{proof}[Construction of $B_{n+1}$]
We start with $B^{(0)}:=Q_{l_{n+1}}$.	Setting $L=10l_{n}$, $K=M^2$ and  $S=\bigcup_{i=1}^{\kappa_{n}} \bigcup_{ \theta\in \widetilde{\mathcal{N}}_{n}^{(i)}}S_{n-1}(\theta,E^*,g(\delta_{n-1}))$ in  Lemma \ref{set} yields  that there exists $B^{(1)}$ satisfying 
	\begin{itemize}
		\item  $B^{(0)}\subset B^{(1)}\subset \{x\in \Z^d: \ \operatorname{dist}_1(x,B^{(0)})\leq 20M^2l_n\} .$
		\item   $(Q_{10l_n}+x)\subset B^{(1)}$ for any $x\in\bigcup_{i=1}^{\kappa_{n}} \bigcup_{ \theta\in \widetilde{\mathcal{N}}_{n}^{(i)}}S_{n-1}(\theta,E^*,g(\delta_{n-1}))$ such that  $(Q_{10l_n}+x)\cap B^{(1)}\neq \emptyset$.
	\end{itemize}
	By setting $L=10l_{n-k}$, $K=M^2$ and  $S=\bigcup_{i=1}^{\kappa_{n}} \bigcup_{ \theta\in \widetilde{\mathcal{N}}_{n}^{(i)}}S_{n-k-1}(\theta,E^*,g(\delta_{n-k-1}))$ in Lemma \ref{set},  we inductively	define  $B^{(k+1)}$ ($1\leq k\leq n-1$), such that
	\begin{itemize}
		\item 	$B^{(k)}\subset B^{(k+1)}\subset \{x\in \Z^d: \ \operatorname{dist}_1(x,B^{(k)})\leq 20l_{n-k}\},$
		\item[$\bigstar$]  	 	 $(Q_{10l_{n-k}}+x)\subset B^{(k+1)}$ for any $x\in \bigcup_{i=1}^{\kappa_{n}} \bigcup_{ \theta\in \widetilde{\mathcal{N}}_{n}^{(i)}}S_{n-k-1}(\theta,E^*,g(\delta_{n-k-1}))$ such that $(Q_{10l_{n-k}}+x)\cap B^{(k+1)}\neq \emptyset$. 
	\end{itemize}
	We claim that $B^{(n)}$ is the desired block $B_{n+1}$ satisfying \eqref{r1}. It follows from  $\sum_{k=m}^{n-1}20l_{n-k}\leq 30l_{n-m}$ that 
	\begin{equation}\label{1316?}
		B^{(n)}\subset \{x\in \Z^d :\ \operatorname{dist}_1(x,B^{(m)})\leq 30M^2l_{n-m}\}.
	\end{equation}
	Thus  $$B^{(n)}\subset \{x\in \Z^d :\ \operatorname{dist}_1(x,B^{(0)})\leq  30M^2l_n\}\subset  Q_{l_{n+1}+50M^2l_{n}}.$$
	Now assume  $$ \text{$x\in  \left(\bigcup_{i=1}^{\kappa_{n}} \bigcup_{ \theta\in \widetilde{\mathcal{N}}_{n}^{(i)}} S_{k-1}(\theta,E^*,g(\delta_{k-1}))\right)\cap B^{(n)}$   for some $ k\in [1,n]  $. }$$  By \eqref{1316?} with $m=n-k+1$, we have
	$$
	B^{(n)}\subset \{x\in \Z^d :\ \operatorname{dist}_1(x,B^{(n-k+1)})\leq 30M^2l_{k-1}\}.
	$$
	It follows from   $30M^2l_{k-1}\leq l_{k}$ that 	$(Q_{10l_{k}}+x)\cap B^{(n-k+1)}\neq \emptyset.$
	Thus by $\bigstar$,    $(Q_{10l_{k}}+x)\subset  B^{(n-k+1)}\subset B^{(n)}.$ This proves that  $B_{n+1}=B^{(n)}$ satisfies  \eqref{r1}.
\end{proof}
\begin{rem}\label{gouzao}
	 In fact, from  the above proof we  also obtain  that 
	  for any finite set $\Lambda\subset\Z^d$ and any fixed $\theta\in\C/\Z $,  there exists an extension $\widetilde{\Lambda}$ such that 
	  \begin{itemize}
	  	\item $\Lambda\subset\widetilde{\Lambda}\subset \{x\in \Z^d:\ \operatorname{dist}_1(x,\Lambda)\leq 50M^2l_n\}.$
	  	\item  $\widetilde{\Lambda}$ is $(\theta,E^*)$-$n$-regular. 
	  \end{itemize}
\end{rem}

\begin{proof}[Construction of $A_{n+1}^{(i)}$]
	Recall that each $\theta_n^{(i)}(E^*)$ is obtained from some $\mathcal{N}_{n-1}^{(i')}$,  a neighborhood  in the previous step (cf. Hypothesis \ref{h1}). Thus  for $\theta\in {\widetilde{\mathcal{N}}}_n^{(i)}$ and $y\in \Omega_n^{(i)}$, it follows from \eqref{tiaon} that \begin{equation}\label{!}
		\text{	$\theta+y\cdot\omega\in \widetilde{\mathcal{N}}_{n-1}^{(i_y)}$  for some $i_y\in [1,\kappa_{n-1}]$. }
	\end{equation}
	Recall the set $ A_n^{(i_y)}\subset\operatorname{Core}( B_n)$ in Hypothesis \ref{h3}. We define (for $n=0$, $A_0^{(i_y)}:=\{o\}$) 
	\begin{equation}\label{yuns}
		A_{n+1}^{(i)}:=\bigcup_{y\in  \Omega_n^{(i)}}(A_n^{(i_y)}+y). 
	\end{equation}
	We need to  prove that $A_{n+1}^{(i)}$ satisfies 
	$|A_{n+1}^{(i)}|\leq M^{n+1}$ and $\widehat{B}_{n+1}^{(i)}=B_{n+1}\setminus A_{n+1}^{(i)} $ satisfies \eqref{r2} and  \eqref{n1}.
	Clearly, $|A_{n+1}^{(i)}| \leq M^{n+1}$ follows from  $|A_n^{(i_x)}|\leq M^n$ together with  $|\Omega_n^{(i)}|\leq M$. By Proposition \ref{scalen} \textbf{(5)} and $l_n^\alpha \gg l_n$, \begin{equation}\label{fankai}
		\text{the blocks $
			\{B_n+y\}_{y\in  \Omega_n^{(i)}}$  are separated from each other by at least  $\ l_n^\alpha/2$.}
	\end{equation} Thus \eqref{yuns} is a disjoint union. 
	By $\Omega_n^{(i)}\subset Q_{l_{n+1}^{1/\alpha}}$ and $l_n+50M^2l_{n-1}\leq l_{n+1}^{1/\alpha}$,  we have $(A_n^{(i_y)}+y)\subset (B_n+y)\subset Q_{2l_{n+1}^{1/\alpha}}$ for any  $ y\in \Omega_n^{(i)}$. 
	Thus by \eqref{yuns}, we have 
	$$	A_{n+1}^{(i)}\subset Q_{2l_{n+1}^{1/\alpha}}= \operatorname{Core}( B_{n+1}).  $$
Now we prove  \eqref{r2}. 
	Assume \begin{equation}\label{?}
		\text{	$1\leq k\leq n-1$ and $x\in \left(\bigcup_{\theta\in \widetilde{\mathcal{N}}_{n}^{(i)}}S_{k-1}(\theta,E^*,g(\delta_{k-1}))\right)\cap \widehat{B}_{n+1}^{(i)}$.}
	\end{equation}   Then $(Q_{5l_k}+x)\subset B_{n+1}$ by \eqref{r1}. If \begin{equation}\label{md}
		\text{$(Q_{5l_k}+x)\cap(A_n^{(i_y)}+y)\neq \emptyset$   for some $y\in  \Omega_n^{(i)}$,}
	\end{equation} then  
	$\|x-y\|_1\leq 5l_k+2l_n^{1/\alpha}<l_n-5l_k$. Thus $(Q_{5l_k}+x-y)\subset B_n$ and 
	$x-y\in B_n\cap (\widehat{B}_{n+1}^{(i)}-y)= (B_n\setminus A_n^{(i_y)})=\widehat{B}_n^{(i_y)}$. 
	It follows from \eqref{!} and \eqref{?} that  $x-y\in  \bigcup_{\theta\in \widetilde{\mathcal{N}}_{n-1}^{(i_y)}}S_{k-1}(\theta,E^*,g(\delta_{k-1})) $. Thus  $(Q_{5l_k}+x-y)\subset \widehat{B}_n^{(i_y)}$, i.e., $(Q_{5l_k}+x)\subset( B_n+y)\setminus (A_n^{(i_y)}+y)$    by Hypothesis \ref{h3} (cf. \eqref{shenm1}).  This  contradicts  \eqref{md}. Assume  \begin{equation}\label{??}
		\text{	$ k= n$ and $x\in \left(\bigcup_{\theta\in \widetilde{\mathcal{N}}_{n}^{(i)}}S_{n-1}(\theta,E^*,g(\delta_{n-1}))\right)\cap \widehat{B}_{n+1}^{(i)}$.}
	\end{equation}
	Then $(Q_{5l_n}+x)\subset B_{n+1}$ by \eqref{r1}.  If 
	$(Q_{5l_n}+x)\cap(A_n^{(i_y)}+y)\neq \emptyset $ for some $y\in  \Omega_n^{(i)}$,
	then  $\|x-y\|_1\leq 6l_n$.   It follows from \eqref{!} and \eqref{??}  that $$x-y\in  \bigcup_{\theta\in \widetilde{\mathcal{N}}_{n-1}^{(i_y)}}S_{n-1}(\theta,E^*,g(\delta_{n-1})) ,$$ which implies that there exists some  $i'\in [1,m_{n-1}] $ such  that  $$\|\theta_{n-1}^{(i_y)}(E^*)+(x-y)\cdot\omega-\theta_{n-1}^{(i')}(E^*)\|_\T\leq \widehat{\delta}_{n-1}+g(\delta_{n-1})<g^{(M^4)}(\bar{\delta}_{n-1}).$$
	It follows from  \eqref{fanz} with $n$ replaced by $n-1$ and $\|x-y\|_1\leq 6l_n<l_n^\alpha$  that  $x-y \in \Omega_{n-1}^{(i_y)}\subset A_n^{(i_y)}$ ($o\in \Omega_k^{(i)}$ holds  for any $k,i$ since $\theta_k^{(i)}(E^*)$ is always equivalent to itself.  By  \eqref{yuns} with $n$ replaced by $0,\cdots , n-1$, we obtain   $\Omega_{n-1}^{(i_y)}\subset A_n^{(i_y)}$). Thus $x\in A_n^{(i_y)}+y$, which  contradicts  \eqref{??}. This proves   \eqref{r2}. 
	
	 Finally,  \eqref{n1} follows from $\Omega_n^{(i)}\subset A_{n+1}^{(i)}$,  \eqref{fanz} and  $\widehat{\delta}_n+g^{(2)}(\widehat{\delta}_n)<g^{(M^4)}(\bar{\delta}_{n})$.
	
\end{proof}

\subsubsection{Resonance analysis via Schur complement}\label{RN}
For $x\in \Omega_n^{(i)}$, we define  the set $$\mathcal{C}_n^{(i),x}:=\{\theta_n^{(i')}(E^*)\in \mathcal{C}_n^{(i)}:\ \|\theta_n^{(i)}(E^*)+x\cdot \omega-\theta_n^{(i')}(E^*)\|_\T\leq \bar{\delta}_n\}. $$
By Proposition  \ref{scalen} \textbf{(3)}, we have \begin{equation}\label{jishun}
	\sum_{x\in\Omega_n^{(i)}}|\mathcal{C}_n^{(i),x}|=|\mathcal{C}_n^{(i)}|.
\end{equation}
For $\theta\in \widetilde{\mathcal{N}}_{n}^{(i)} $, 	by  \eqref{!}  ($\theta+x\cdot\omega\in \widetilde{\mathcal{N}}_{n-1}^{(i_x)}$  for some $i_x\in [1,\kappa_{n-1}]$)  and \eqref{SS}, we have  
\begin{equation}\label{5656}
	s_{A_n^{(i_x)}}(\theta+x\cdot\omega,E^*)\overset{\delta_{n-1}}{\sim} \prod_{j=1}^{|\mathcal{C}_{n-1}^{(i_x)}|}\|\theta+x\cdot\omega -\theta^{(i_x,j)}_{n}(E^*) \|_\T,\quad \forall \theta\in \widetilde{\mathcal{N}}_{n}^{(i)}. 
\end{equation}
\begin{prop}\label{SLn}The following statements hold: 
	\begin{itemize}
		\item[\textbf{(1)}]  For $x\in \Omega_n^{(i)}$ and $\theta_n^{(i')}(E^*)\in \mathcal{C}_n^{(i),x}$, we have  $\theta_n^{(i')}(E^*)-x\cdot \omega  \in \mathcal{N}_n^{(i)}$ (considered in $\C/\Z$).  
		\item  [\textbf{(2)}] For $x\in \Omega_n^{(i)}$ and $\theta\in  {\widetilde{\mathcal{N}}_n}^{(i)}$, we have 
		\begin{equation}\label{bianyi}
			s_{A_n^{(i_x)}}(\theta+x\cdot\omega,E^*)\overset{g^{(2)}(\widehat{\delta}_n)}{\sim} \prod_{\theta_n^{(i')}(E^*)\in \mathcal{C}_n^{(i),x}}\|\theta-(\theta_n^{(i')}(E^*)-x\cdot \omega)\|_\T.
		\end{equation}
	\end{itemize}
\end{prop}

\begin{proof}
Recall that $\widehat{\delta}_n\in [g^{(M)}(\bar{\delta}_n),g^{(M^3)}(\bar{\delta}_n)]$ (cf. Proposition \ref{realn}).  Now \textbf{(1)} follows from $\bar{\delta}_n\ll h(\widehat{\delta}_n)\ll\widehat{\delta}_n^2$;	\textbf{(2)} follows from  \eqref{5656} and \eqref{fenkn} (\eqref{fenkn} implies  $	\|\theta_n^{(i)}(E^*)+x\cdot\omega-\theta_n^{(i')}(E^*)\|_\T> g^{(M^4)}(\bar{\delta}_n)\gg g^{(3)}(\widehat{\delta}_n)$ for $\theta_n^{(i')}(E^*)\notin \mathcal{C}_n^{(i),x} $). 
\end{proof}
Write 
$$	E-H_{B_{n+1}}(\theta)=\begin{pmatrix}
	E-H_{A_{n+1}^{(i)}}(\theta) & \Gamma^{{\rm T}} \\
	\Gamma & E-H_{\widehat{B}^{(i)}_{n+1}}(\theta)
\end{pmatrix}.$$
Now restrict  $(\theta,E)\in  \widetilde{\mathcal{N}}_n^{(i)}\times D(E^*,g(\widehat{\delta}_n)).$  We define the Schur complement 
\begin{equation}\label{su}
	S_{A_{n+1}^{(i)}}(\theta,E):=	E-H_{A_{n+1}^{(i)}}(\theta)-\Gamma^{{\rm T}} G_{\widehat{B}^{(i)}_{n+1}}^{\theta,E} \Gamma.
\end{equation}
Since $\widehat{B}^{(i)}_{n+1}$ is $(\theta, E^*)$-$n$-regular and $\theta$-type-$(\theta,E^*,g^{(2)}(\widehat{\delta}_{n}))$-$n$-nonresonant, by Hypothesis \ref{h2} (cf. \eqref{12z}),  we have 
\begin{equation}\label{n+1g}
	\|G_{\widehat{B}^{(i)}_{n+1}}^{\theta,E}\|\leq \delta_{n-1}^{-5} g^{(2)}(\widehat{\delta}_{n})^{-M}\leq \widehat{\delta}_n^{-1}.
\end{equation}
Recalling \eqref{!}, we let $\check{B}_{n+1}^{(i)}:=\bigcup_{z\in \Omega_n^{(i)}}(\widehat{B}_n^{(i_z)}+z),$ which is a disjoint union by \eqref{fankai}.
Thus  \begin{align}
	\label{G}	G_{\check{B}_{n+1}^{(i)}}^{\theta,E}&=\bigoplus_{z\in \Omega_n^{(i)}}G_{\widehat{B}_n^{(i_z)}+z}^{\theta,E},\\ 
	\label{Aaa}	H_{A_{n+1}^{(i)}}(\theta)&=\bigoplus_{z\in \Omega_n^{(i)}}H_{A_n^{(i_z)}+z}(\theta).
\end{align}   
Moreover, since   $\theta+z\cdot\omega\in \widetilde{\mathcal{N}}_{n-1}^{(i_z)}$ for $z\in \Omega_n^{(i)}$,   it follows  from  \eqref{shenm1}, \eqref{shenm2} and Hypothesis \ref{h2} (cf. \eqref{13z}) that  
\begin{equation}\label{ngg}
	|G_{\widehat{B}_n^{(i_z)}+z}^{\theta,E}(x,y)|\leq e^{-\gamma_{n-1}\|x-y\|_1}, \quad \forall \|x-y\|_1\geq l_{n-1}^\frac{2}{\alpha+1}.
\end{equation}
By the  resolvent identity, for $x\in \check{B}_{n+1}^{(i)}$, we have 
\begin{equation}\label{yujie}
	G_{\widehat{B}^{(i)}_{n+1}}^{\theta,E}(x,y)=  G_{\check{B}_{n+1}^{(i)}}^{\theta,E}(x,y) \chi_{\check{B}_{n+1}^{(i)}}(y)- \varepsilon\sum_{(w,w')\in \partial_{\widehat{B}^{(i)}_{n+1}}\check{B}_{n+1}^{(i)}} G_{\check{B}_{n+1}^{(i)}}^{\theta,E}(x,w)  G_{\widehat{B}^{(i)}_{n+1}}^{\theta,E}(w',y).
\end{equation}
Recall \eqref{yuns} and $(A_n^{(i_z)}+z)\subset\operatorname{Core}( B_n+z)$. Then for  $$\text{$x,y\in \check{B}^{(i)}_{n+1}$ such that $\operatorname{dist}_1(x, A_{n+1}^{(i)}) =\operatorname{dist}_1(y, A_{n+1}^{(i)})=1,$} $$ we have $\operatorname{dist}_1(x, \partial^-_{\widehat{B}^{(i)}_{n+1}}\check{B}_{n+1}^{(i)})>l_{n}-2l_{n}^{1/\alpha}-2 \geq l_{n}/2$. 
It follows from \eqref{n+1g}, \eqref{G} and \eqref{ngg} that  the second term of the right hand side in \eqref{yujie} is bounded above by 
$10dl_n^d\widehat{\delta}_n^{-1}e^{-5l_n}\leq e^{-4l_n}$. 
Substituting  \eqref{yujie}, \eqref{G} and  \eqref{Aaa} into  \eqref{su} yields 
\begin{equation}\label{kuaile}
	S_{A_{n+1}^{(i)}}(\theta,E) =\bigoplus_{x\in \Omega_n^{(i)}}S_{(A_n^{(i_x)}+x)}(\theta,E)- R(\theta,E),
\end{equation}
where $R(\theta,E)$ is an $|A_{n+1}^{(i)}|\times|A_{n+1}^{(i)}|$ analytic matrix with elements bounded by $4d^2\varepsilon^2e^{-4l_n}$. 
Define 
$$s_{A_{n+1}^{(i)}}(\theta,E):=\operatorname{det}S_{A_{n+1}^{(i)}}(\theta,E) .$$ Hence $s_{A_{n+1}^{(i)}}(\theta,E)$ is an analytic function in $(\theta,E)\in  \widetilde{\mathcal{N}}_n^{(i)}\times D(E^*,g(\widehat{\delta}_n)).$ 
By \eqref{kuaile}, we have \begin{equation}\label{1712n}
	s_{A_{n+1}^{(i)}}(\theta,E)=\mathcal{Q}_{n+1}^{(i)}(\theta,E)-r(\theta,E), 
\end{equation} where  \begin{equation}\label{pn+1}
	\mathcal{Q}_{n+1}^{(i)}(\theta,E):=	\prod_{x\in \Omega_n^{(i)}}s_{A_n^{(i_x)}}(\theta+x\cdot\omega ,E)
\end{equation} and $|r(\theta,E)|\leq M^{n+1} (2C_v)^{M^{n+1}} 4d^2\varepsilon^2e^{-4l_n}\leq e^{-l_n} $      by $|A_{n+1}^{(i)}|\leq M^{n+1}$ and Remark \ref{dif}, noting   the row and column summation  of    $S_{A_{n}^{(i_x)}}(\theta,E)$  ($k$-step differs at most $e^{-l_{k-1}}$ from the $(k-1)$-step)  is bounded by $2C_v$.  

\begin{prop}\label{654n}
	Fix $E=E^*$. Then the $\theta$-variable analytic function  $s_{A_{n+1}^{(i)}}(\theta,E^*)$ has $|\mathcal{C}_n^{(i)}|$ many  roots  in $D_\T(\Re(\theta_n^{(i)}(E^*)),2\widehat{\delta}_n^2)$, denoted by $\{\theta_{n+1}^{(i,j)}(E^*)\}_{j=1}^{|\mathcal{C}_n^{(i)}|}$,   and has  no other root in the extension $\widetilde{\mathcal{N}}_n^{(i)}$.  Moreover, we have 
	$$s_{A_{n+1}^{(i)}}(\theta,E^*)\overset{\delta_n}{\sim} 
	\prod_{j=1}^{|\mathcal{C}_n^{(i)}|}\left(\theta-\theta_{n+1}^{(i,j)}(E^*) \right),\quad \forall \theta\in \widetilde{\mathcal{N}}_n^{(i)}.$$
\end{prop}

\begin{proof}
	For $\theta\in \widetilde{\mathcal{N}}_n^{(i)}$,	by Proposition \ref{SLn} \textbf{(1)}, \textbf{(2)} and \eqref{jishun},  we have 
	\begin{align}
		\nonumber		\mathcal{Q}_{n+1}^{(i)}(\theta,E^*)	&\overset{g^{(2)}(\widehat{\delta}_n)^M}{\sim} \prod_{x\in \Omega_n^{(i)}} \  \prod_{\theta_n^{(i')}(E^*)\in \mathcal{C}_n^{(i),x}}\|\theta-(\theta_n^{(i')}(E^*)-x\cdot \omega)\|_\T\\
		&\ \ \	\sim \prod_{\nu=1}^{|\mathcal{C}_n^{(i)}|} \left(\theta-\widetilde{\theta}_n^{(i,\nu)}(E^*)\right), \label{123456}
	\end{align}
	where $\{\widetilde{\theta}_n^{(i,\nu)}(E^*)\}_{\nu=1}^{|\mathcal{C}_n^{(i)}|}=\{\theta^{(i')}_n(E^*)-x\cdot \omega:\ x\in \Omega_n^{(i)},  \theta_n^{(i')}(E^*)\in \mathcal{C}_n^{(i),x}\}\subset  \mathcal{N}_n^{(i)}$.
	Now  the proof follows from the same argument as in the proof of Proposition \ref{654}, applying Rouch\'e  argument to \eqref{1712n} for $\theta$-variable  and using    $e^{-l_n}\ll \delta_n^M$. 
\end{proof}

\begin{prop}\label{sen}
	There exists some $\widetilde{\delta}_n\in  [\widehat{\delta}_n^{\frac{1}{4}}, g(\widehat{\delta}_n )]$ such that for any $i\in [1,\kappa_n]$ and   $\theta^*\in \widetilde{\mathcal{N}}_n^{(i)}$, the $E$-variable analytic function $s_{A_{n+1}^{(i)}}(\theta^*,E)$  has $\iota_{n}^{(i)}\in [0, |\mathcal{C}_n^{(i)}|]$ many  roots  in $ D(E^*,3\widetilde{\delta}_n^2)$, denoted by \{$E_{n+1}^{(i,j)}(\theta^*)\}_{j=1}^{\iota_{n}^{(i)}}$,   and has  no other root in the extension $E \in D(E^*,\widetilde{\delta}_n)$. By Lemma \ref{Su}, these roots are the eigenvalues of $H_{B_{n+1}} (\theta^*)$ in $D(E^*,\widetilde{\delta}_n)$.  Moreover,  we have 
	$$s_{A_{n+1}^{(i)}}(\theta^*,E) \overset{\delta_n}{\sim}   \prod_{j=1}^{\iota_{n}^{(i)}}\left(E-E_{n+1}^{(i,j)}(\theta^*) \right),\quad \forall E \in D(E^*,\widetilde{\delta}_n) .$$
\end{prop}

\begin{proof} 
We  need the following    Lemma \ref{line}  on  variations of normal matrix from \cite{line}.  
	\begin{lem}[{\cite[Theorem 1.1]{line}}]\label{line}
		Let $A$ and $B$ be two $N\times N$ matrices, where $A$ is normal and $B$ is arbitrary, with $\operatorname{Spec}(A)=\left\{\lambda_1, \cdots, \lambda_N\right\}$ and $\operatorname{Spec}(B)=\left\{\mu_1, \cdots, \mu_N\right\}$. Then there exists a permutation $\pi$ of $\{1,2, \cdots, N\}$ such that
		$$
		\sqrt{\sum_{j=1}^N\left|\mu_{\pi(j)}-\lambda_j\right|^2} \leq \sqrt{N}\|B-A\|_F,
		$$
		where $\|\cdot\|_F$ is the Frobenius norm of a  matrix:   $\|A\|_F=\sqrt{\sum_{i=1}^N \sum_{j=1}^N\left|A_{i j}\right|^2} .$
	\end{lem}
	We denote the set  $$\widetilde{B}_{n+1}^{(i)}:=\bigcup_{x\in \Omega_n^{(i)}} (B_n+x). $$
	Substituting \eqref{SE} into \eqref{pn+1}, we have for any  $E \in D(E^*,\widetilde{\delta}_{n-1})$, 
	\begin{align}
		\nonumber 	\mathcal{Q}_{n+1}^{(i)}(\theta^*,E) &\overset{\delta_{n-1}^M}{\sim} 	\prod_{x\in \Omega_n^{(i)}} \prod_{j=1}^{\iota_{n-1}^{(i_x)}}\left(E-E_n^{(i_x,j)}(\theta^*+x\cdot\omega) \right)\\ 
		\label{123}		& \ \ =: \prod_{\nu=1}^{J^{(i)}}\left(E-\widetilde{E}_n^{(i,\nu)}(\theta^*) \right), 
	\end{align}
	where $J^{(i)}=\sum_{x\in \Omega_n^{(i)}}\iota_{n-1}^{(i_x)}$ and  $\{\widetilde{E}_n^{(i,\nu)}(\theta^*)\}_{\nu=1}^{J^{(i)}} =\bigcup_{x\in \Omega_n^{(i)}}\{E_n^{(i_x,j)}(\theta^*+x\cdot\omega)\}_{j=1}^{\iota_{n-1}^{(i_x)}}$ are the eigenvalues of $H_{\widetilde{B}_{n+1}^{(i)}}(\theta^*)$ in  $D(E^*,\widetilde{\delta}_{n-1})$ by Lemma \ref{Su}.  
	\begin{claim}\label{cm}
		For $\theta^*=\Re(\theta^{(i)}_n(E^*))$,   there are at most $ |\mathcal{C}_n^{(i)}|$ many  $\nu$ in \eqref{123}  such that   $|E^*-\widetilde{E}_n^{(i,\nu)}(\Re(\theta^{(i)}_n(E^*)))|\leq g(\widehat{\delta}_n ).$
	\end{claim}

	\begin{proof}[Proof of Claim]
	Suppose  the Claim fails.  Then there are  \begin{equation}\label{mm}
			\text{ at least $( |\mathcal{C}_n^{(i)}|+1)$ many  $\nu$   such that   $|E^*-\widetilde{E}_n^{(i,\nu)}(\Re(\theta^{(i)}_n(E^*)) )|\leq g(\widehat{\delta}_n ).$}
		\end{equation}
		Choose a $\phi\in \C$ such that $|\phi-\Re(\theta^{(i)}_n(E^*))| = g(\widehat{\delta}_n)^\frac12 $.   Since  $\phi +x\cdot\omega \in  \widetilde{\mathcal{N}}_{n-1}^{(i_x) }$  by  $g(\widehat{\delta}_n)^{\frac12 }\ll \delta_{n-1}$, \eqref{123} still  holds for $\theta^*=\phi$ and $E=E^*$.
		Note that $ H_{\widetilde{B}_{n+1}^{(i)}} (\Re(\theta^{(i)}_n(E^*)))$ is self-adjoint,   hence normal  since $v$ is real-valued on $\T$. Applying Lemma  \ref{line} to $A=H_{\widetilde{B}_{n+1}^{(i)}} (\Re(\theta^{(i)}_n(E^*))) $  and $B=  H_{\widetilde{B}_{n+1}^{(i)}} (\phi )  $,  we obtain that up to a permutation,  the eigenvalues of   $ H_{\widetilde{B}_{n+1}^{(i)}} (\Re(\theta^{(i)}_n(E^*)))$ and  $ H_{\widetilde{B}_{n+1}^{(i)}} (\phi )$ differ at most $ C l_{n+1}^{2d} g(\widehat{\delta}_n)^{\frac12 } $. Thus by \eqref{mm}, there are
		\begin{equation}\label{mma}
			\text{ at least $( |\mathcal{C}_n^{(i)}|+1)$ many  $\nu$   such that   $|E^*-\widetilde{E}_n^{(i,\nu)}(\phi  )|\leq C l_{n+1}^{2d} g(\widehat{\delta}_n)^{\frac12 }.$}
		\end{equation}
		By \eqref{123} and \eqref{mma},  we have 
		\begin{equation}\label{1234}
			|	\mathcal{Q}_{n+1}^{(i)}(\phi ,E^*)|	\leq  \delta_{n-1}^{-2M}  g(\widehat{\delta}_n)^{\frac12(|\mathcal{C}_n^{(i)}|+1)  }< g(\widehat{\delta}_n)^{\frac12|\mathcal{C}_n^{(i)}|+\frac14   } .
		\end{equation}
		On the other hand,  it follows from  \eqref{fenkn}  and  $g^{(M^4)}(\bar{\delta}_n)\gg g^{(3)}(\widehat{\delta}_n)$ that   \eqref{bianyi} hence \eqref{123456} still holds for $\theta=\phi$.  
		By \eqref{123456} and $\{\widetilde{\theta}_n^{(i,\nu)}(E^*)\}_{\nu=1}^{|\mathcal{C}_n^{(i)}|} \subset\mathcal{N}_n^{(i)} $,  we have  
		$$	\mathcal{Q}_{n+1}^{(i)}(\phi ,E^*)  \overset{g^{(2)}(\widehat{\delta}_n)^M}{\sim}  \prod_{\nu=1}^{|\mathcal{C}_n^{(i)}|} \left(\phi-\widetilde{\theta}_n^{(i,\nu)}(E^*)\right)\sim g(\widehat{\delta}_n)^{\frac{1}{2}|\mathcal{C}_n^{(i)}|}.$$
		Thus \begin{equation*}
			|\mathcal{Q}_{n+1}^{(i)}(\phi ,E^*)|\gtrsim  g^{(2)}(\widehat{\delta}_n)^M g(\widehat{\delta}_n)^{\frac12 |\mathcal{C}_n^{(i)}|}\geq g(\widehat{\delta}_n)^{\frac12|\mathcal{C}_n^{(i)}|+\frac14 } ,
		\end{equation*}
		which contradict \eqref{1234}.  This proves Claim \ref{cm}. 
	\end{proof}
	Partition the   interval $( g(\widehat{\delta}_n)^{4^{M+1}},  g(\widehat{\delta}_n)]$ into $(M+1)$ many  disjoint intervals $$( g(\widehat{\delta}_n)^{4^{k+1}},  g(\widehat{\delta}_n)^{4^{k}}  ], \quad    0\leq k\leq M.$$ 
	It follows from   Claim \ref{cm},  $\sum_{i=1}^{\kappa_n}   |\mathcal{C}_n^{(i)}| \leq m_n\leq M $  and  the  pigeonhole principle that  there exists an interval  $J_{k_0}:=( g(\widehat{\delta}_n)^{4^{k_0+1}},  g(\widehat{\delta}_n)^{4^{k_0}}  ]$ such that 
	$$J_{k_0}\cap\{|E^*-\widetilde{E}_n^{(i,\nu)}(\Re(\theta^{(i)}_n(E^*)))|\}_{\nu=1}^{J^{(i)}}=\emptyset, \quad \forall i\in [1,\kappa_n]. $$ 
	We define 
	$$\widetilde{\delta}_n:=g(\widehat{\delta}_n)^{2\cdot 4^{k_0}} . $$ 
	Clearly,  $\widetilde{\delta}_n\in  [\widehat{\delta}_n^{\frac{1}{4}}, g(\widehat{\delta}_n )]$ satisfies that $ H_{\widetilde{B}_{n+1}^{(i)}} (\Re(\theta^{(i)}_n(E^*)))$ has no eigenvalue in $\{E\in \C:\   \widetilde{\delta}_n^2<| E-E^*|<\widetilde{\delta}_n^{\frac12} \}$ and has $\iota_{n}^{(i)}\in [0, |\mathcal{C}_n^{(i)}|]$ many  eigenvalues  in $D(E^*,\widetilde{\delta}_n^2)$. It follows from Lemma  \ref{line} and $l_{n+1}^{2d}\widehat{\delta}_n\ll \widetilde{\delta}_n^2$ that  for any $\theta^*\in {\widetilde{\mathcal{N}}}_n^{(i)}$, 
	\begin{equation}\label{meigen}
		\text{		$ H_{\widetilde{B}_{n+1}^{(i)}} (\theta^*)$ has no eigenvalue in $\{E\in \C:\   2\widetilde{\delta}_n^2<| E-E^*|<\frac12\widetilde{\delta}_n^{\frac12} \}$}, 
	\end{equation}
	and has  $\iota_{n}^{(i)}$ many  eigenvalues  in $D(E^*,\widetilde{\delta}_n)$,	denoted by  $\{\widetilde{E}_n^{(i,j)}(\theta^*)\}_{j=1}^{\iota_{n}^{(i)}}$.  By \eqref{meigen}, the other eigenvalues are not in  $D(E^*,2\widetilde{\delta}_n)$. Thus\begin{equation}\label{shediao}
		|\widetilde{E}_n^{(i,\nu)}(\theta^*)-E| \geq \widetilde{\delta}_n, \quad \forall E\in D(E^*,\widetilde{\delta}_n),\  \nu \geq \iota_{n}^{(i)}+1.   
	\end{equation}
By  \eqref{123} and \eqref{shediao},  we have 
	\begin{equation}\label{zuihou}
		\mathcal{Q}_{n+1}^{(i)}(\theta^*,E) \overset{\delta_{n-1}^M}{\sim} 	 
		\prod_{\nu=1}^{{J^{(i)}}}\left(E-\widetilde{E}_n^{(i,\nu)}(\theta^*) \right)  \overset{\widetilde{\delta}_n^{{M^2}}}{\sim} 	\prod_{\nu=1}^{\iota_n^{(i)} }\left(E-\widetilde{E}_n^{(i,\nu)}(\theta^*) \right) .
	\end{equation}
	Substituting \eqref{zuihou} into \eqref{1712n}, we have 
	\begin{equation}\label{1712na}
		s_{A_{n+1}^{(i)}}(\theta^*,E)\overset{\delta_{n}}{\sim} \prod_{\nu=1}^{\iota_n^{(i)} }\left(E-\widetilde{E}_n^{(i,\nu)}(\theta^*) \right)-\widetilde{r}(\theta^*,E)=:L_{n+1}^{(i)}(\theta^*,E)
	\end{equation}
	with $|\widetilde{r}(\theta^*,E)|\leq \delta_n^{-1}e^{-l_n}\leq e^{-\frac45 l_n}.   $
	Now 	the  conclusion follows from the same argument as in the proof of Proposition \ref{654} by applying Rouch\'e  argument to \eqref{1712na} for the $E$-variable. 
\end{proof}
\begin{rem}\label{kuodan}  
	Similar to  Remark \ref{kuoda},  the conclusions in  Propositions \ref{SLn}, \ref{654n}   \ref{sen} still hold with $ \widetilde{\mathcal{N}}_n^{(i)}$  replaced  by    ${\rm Qdp}(\widetilde{\mathcal{N}}_n^{(i)}):=D_\T(\Re(\theta_n^{(i)}(E^*)),4\widehat{\delta}_n)$.
\end{rem} 

\subsubsection{Green's function estimates for $(n+1)$-nonresonant  sets}\label{GN}
\begin{lem}\label{1833n}
	For $(\theta,E)\in \widetilde{\mathcal{N}}_n^{(i)}\times D(E^*,\widetilde{\delta}_n)$, we have $$\|G_{B_{n+1}}^{\theta,E}\|\leq \delta_n^{-1}|s_{A_{n+1}^{(i)}}(\theta,E)|^{-1}.$$
\end{lem}
\begin{proof}
By Lemma \ref{Su},  we have $\|G_{B_{n+1}}^{\theta,E}\|\lesssim \|S_{A_{n+1}^{(i)}}(\theta,E)^{-1}\|\cdot   \|G_{\widehat{B}^{(i)}_{n+1}}^{\theta,E}\|^2$. 
	Now the conclusion follows from  $\|G_{\widehat{B}^{(i)}_{n+1}}^{\theta,E}\|\leq \widehat{\delta}_n^{-1}$ (cf. \eqref{n+1g})  and $\|S_{A_{n+1}^{(i)}}(\theta,E)^{-1}\|=|s_{A_{n+1}^{(i)}}(\theta,E)|^{-1}\cdot \|S_{A_{n+1}^{(i)}}(\theta,E)^{\#}\|\leq |s_{A_{n+1}^{(i)}}(\theta,E)|^{-1} \cdot M^{n+1}C^{M^{n+1}}\leq  \delta_n^{-\frac12 } |s_{A_{n+1}^{(i)}}(\theta,E)|^{-1} $ since $|A_{n+1}^{(i)}|\leq M^{n+1} $ and the row and column summation  of    $S_{A_{n+1}^{(i)}}(\theta,E)$  ($k$-step differs at most $e^{-l_{k-1}}$ from the $(k-1)$-step)  is bounded by $C$. 
\end{proof}
Combining Propositions \ref{654n}, \ref{sen} and Lemma \ref{1833n} yields 
\begin{prop} \begin{itemize}
		\item [\textbf{(1)}]
		For $\theta\in \widetilde{\mathcal{N}}_n^{(i)}$, we have 		 \begin{equation}\label{123321}
			\|G_{B_{n+1}}^{\theta,E^*}\|\leq \delta_n^{-3}\prod_{j=1}^{|\mathcal{C}_n^{(i)}|}\|\theta-\theta^{(i,j)}_{n+1}(E^*)\|_\T^{-1}. 
		\end{equation}
		\item [\textbf{(2)}]For $(\theta,E)\in \widetilde{\mathcal{N}}_n^{(i)}\times D(E^*,\widetilde{\delta}_n)$, we have 	\begin{equation}\label{456654}
			 \|G_{B_{n+1}}^{\theta,E}\|\leq \delta_n^{-3}\prod_{j=1}^{\iota_{n}^{(i)}}|E-E_{n+1}^{(i,j)}(\theta)|^{-1}. 
\end{equation} 
		\end{itemize}
\end{prop}

\begin{prop}\label{1gon}
	Let $\theta^*\in \T_{\widehat{\delta}_{n}}$, $\delta\in [h(\delta_{n+1}),\delta_{n}]$  and   $\Lambda$ be   $(\theta^*,E^*)$-$(n+1)$-regular. Then we have 
	\begin{itemize}
		\item[\textbf{(1)}] If  $\Lambda$ is  $\theta$-type-$(\theta^*,E^*,\delta)$-$(n+1)$-nonresonant,  then  for any $(\theta, E)\in D_\T(\theta^*,h(\delta))\times D(E^*,h(\delta))$, we have 
		\begin{align}
			\|G_\Lambda^{\theta,E}\|&\leq \delta_{n}^{-5}\delta^{-M},\label{12zn}\\
			|G_\Lambda^{\theta,E}(x,y)|&\leq e^{-\gamma_{n+1}\|x-y\|_1}, \quad \forall \|x-y\|_1\geq l_{n+1}^\frac{2}{\alpha+1}.\label{13zn}
		\end{align}	
		\item[\textbf{(2)}]  Let $E\in D(E^*,h(\delta_{n}))$. If  $\Lambda$ is  $E$-type-$(\theta^*,E,\delta)$-$(n+1)$-nonresonant, then  \eqref{12zn}, \eqref{13zn} hold for $\theta=\theta^*$. 
	\end{itemize}
\end{prop}
\begin{proof} If $ S_n(\theta^*,E^*,\delta_n)\cap \Lambda=\emptyset$, the conclusion follows from the inductive   Hypothesis \ref{h2} with $r=n$. Now assume $ S_n(\theta^*,E^*,\delta_n)\cap \Lambda\neq \emptyset$ but  $ S_{n+1}(\theta^*,E^*,\delta)\cap \Lambda\neq \emptyset$ or  $ \widetilde{S}_{n+1}(\theta^*,E,\delta)\cap \Lambda\neq \emptyset$.  
	Since $\Lambda$ is $(\theta^*,E^*)$-$(n+1)$-regular, similar to  Claim \ref{1goc}, 	one can  claim   that for  every $x\in S_n(\theta^*,E^*,\delta_n)\cap \Lambda$, there exists a block $B_{n+1}(y)(=B_{n+1}+y)$  satisfying $x\in \operatorname{Core}(B_{n+1}(y))\subset B_{n+1}(y) \subset \Lambda$ and  
\begin{align}
	\|G_{B_{n+1}(y)}^{\theta,E}\|&\leq 2\delta_n^{-3}\delta^{-M},\label{121}\\
|G_{B_{n+1}(y)}^{\theta,E}(x',y')|&\leq e^{-\gamma_{n+1}'\|x'-y'\|_1}, \quad \forall \|x'-y'\|_1\geq l_{n+1}^\frac{3}{2\alpha+1}, \label{1211}
\end{align}	
where  $\gamma_{n+1}':=(1-l_{n+1}^{-2(\alpha-1)^8})\gamma_n$.

Proposition \ref{1gon}  then   follows from  \eqref{121}, \eqref{1211}  and the inductive Hypothesis \ref{h2} for $r\leq n$ by a  standard resolvent identity argument, using the fact that $\ln\delta^{-1}\leq  \ln h(\delta_{n+1})^{-1}=l_{n+1}^{\frac{2}{3}\alpha}\ll l_{n+1}$. We refer   to \cite[Lemma 2.41 and Proposition 2.42]{CSZ25}    for   details. 
\end{proof}

\subsubsection{Transversality}\label{TN}
We view the $E$-polynomial  \begin{equation}\label{an}
	P_{n+1}^{(i)}(\theta,E):=	\prod_{j=1}^{\iota_n^{(i)}}\left(E-E_{n+1}^{(i,j)}(\theta) \right)=:\sum_{k=0}^{\iota_n^{(i)}} a_{k}^{(i)}(\theta)E^{\iota_n^{(i)}-k}
\end{equation}
as a perturbation of 
\begin{equation}\label{bn}
	Q_{n+1}^{(i)}(\theta,E):=\prod_{\nu=1}^{\iota_n^{(i)} }\left(E-\widetilde{E}_n^{(i,\nu)}(\theta) \right) \text{(cf. \eqref{zuihou})} =:\sum_{k=0}^{\iota_n^{(i)}} b_{k}^{(i)}(\theta)E^{\iota_n^{(i)}-k}.
\end{equation}
We still use the notation $a_k^{(i)}(\theta)$ and $b_k^{(i)}(\theta)$ as in  \S \ref{trans} since there is no confusion. But they are defined by \eqref{an} and \eqref{bn} in this section. 
\begin{prop}\label{coen}
	The polynomial coefficients $a_k^{(i)}(\theta),b_k^{(i)}(\theta)$ are analytic functions in  
	$\theta\in {\rm Qdp} (\widetilde{\mathcal{N}}_n^{(i)})$. Moreover, for 
	$\theta\in {\rm Db}(\widetilde{\mathcal{N}}_n^{(i)}):=D_\T(\Re(\theta_n^{(i)}(E^*)),2\widehat{\delta}_n)$, we have  
	$$|\partial_\theta^l(a_k^{(i)}(\theta)-b_k^{(i)}(\theta))|\leq l! \widehat{\delta}_n^{-l}e^{-\frac34  l_n}, \quad \forall  \theta\in {\rm Db}(\widetilde{\mathcal{N}}_n^{(i)}),  0\leq k\leq \iota_n^{(i)}, l\geq 0. $$ 
\end{prop}
\begin{proof}
	The proof is the same as that of Proposition  \ref{coe} by the  Weierstrass preparation argument and Cauchy integral formula, noting  that   $\{E_{n+1}^{(i,j)}(\theta) \}_{j=1}^{ \iota_n^{(i)}}$ (resp. $\{\widetilde{E}_{n}^{(i,\nu)}(\theta) \}_{\nu=1}^{ \iota_n^{(i)}}$)  are the $E$-roots of $L_{n+1}^{(i)}(\theta,E)$ (cf. \eqref{1712na})    (resp.  $Q_{n+1}^{(i)}(\theta,E)$) in  $ D(E^*,3\widetilde{\delta}_n^2)$,  and neither  $L_{n+1}^{(i)}(\theta,E)$ nor  $Q_{n+1}^{(i)}(\theta,E)$ has 
	any  other root in $ D(E^*,\widetilde{\delta}_n)$   with $|Q_{n+1}^{(i)}(\theta,E)-L_{n+1}^{(i)}(\theta,E)|\leq  e^{-\frac45 l_n}\ll \delta_n^C  e^{-\frac34 l_n}$.

\end{proof}
For  $\theta\in \widetilde{\mathcal{N}}_n^{(i)}\cap\T$ and $\phi\in \T$ such that $\theta+\phi\in \widetilde{\mathcal{N}}_n^{(i')}\cap\T$, we define the   $E$-resultant of $ P_{n+1}^{(i)}(\theta,E)$ and  $ P_{n+1}^{(i')}(\theta+\phi,E)$   by  
\begin{align*}
	R_{n+1}^{(i,i')}(\theta,\phi)&:= \operatorname{Res}_E(P_{n+1}^{(i)}(\theta,E),P_{n+1}^{(i')}(\theta+\phi,E)).
\end{align*}
\begin{prop}\label{h4n}
	We have   derivative	estimates for 	$R_{n+1}^{(i,i')}(\theta,\phi)$: 
	\item[\textbf{(1)}] Upper bound:  \begin{equation}\label{upp}
		|\partial_\theta^lR_{n+1}^{(i,i')}(\theta,\phi)|\leq  C l! \widehat{\delta}_{n}^{-l}, \quad \forall l\geq 0.
	\end{equation}  
	\item[\textbf{(2)}] Lower  bound: Let $$t_{n+1}:=|\{k\in [1,n+1]:\ \exists i\in [1,\kappa_{k-1}],  \text{ {\rm s.t.}, } |\mathcal{C}_{k-1}^{(i)}|>1 \}|. $$
	Then  we have  \begin{equation}\label{time}
		t_{n+1}\leq M^2+1+\ln (n+1). 
	\end{equation}
	As a corollary,  \begin{equation}\label{cis}
		s M^{2t_{n+1}}\leq Cn^C \text{\ \  with $C=C(s,M)$.}
	\end{equation}
	\begin{itemize}
		\item  	If $i\neq i'$, then 
		\begin{equation}\label{jie1}
			\max _{0 \leq l \leq s M^{2t_{n+1}}}| \partial^l_\theta 	R_{n+1}^{(i,i')}(\theta,\phi)| 	\geq \delta_{n}.
		\end{equation}
		\item 	If $i=i'$, then 
		\begin{equation}\label{jie2}
			\max_{0 \leq l \leq sM^{2t_{n+1}}}| \partial^l_\theta 	R_{n+1}^{(i,i)}(\theta,\phi)| 	\geq \delta_{n}\|\phi\|_\T^{\iota_{n}^{(i)}}.
		\end{equation}
	\end{itemize}
	
\end{prop}

\begin{proof} 
	The proof of  \textbf{(1)} is the same as that of Proposition \ref{zhizhu} \textbf{(1)}. Note that   for such fixed  $\theta$ and $\phi$,   $R_{n+1}^{(i,i')}(w,\phi)$ is an analytic function in $w\in  D_\T(\theta,\widehat{\delta}_n )$ bounded above by $C$. 
	
	Next, we prove \textbf{(2)}.  In order to prove \eqref{time}, it suffices to prove 
	\begin{equation}\label{1607}
		|\{k\in [\ln(n+1)+1,n+1] :\  \exists i\in [1,\kappa_{k-1}], \text{ {\rm s.t.},   } |\mathcal{C}_{k-1}^{(i)}|>1\}  |\leq M^2. 
	\end{equation} 
Suppose that \eqref{1607} fails. Then by Proposition \ref{scalen} \textbf{(4)} \textbf{(5)} with $n+1$ replaced by $k$,  there are  $(M^2+1)$ many different $k \in [\ln(n+1)+1,n+1]$ such that  
	\begin{equation}\label{cunzai}
		\exists i_{k-1}\neq i_{k-1}',\   l_{k-1}^\alpha  <\|x_k\|_1\leq l_{k}^\frac{1}{\alpha},   \text{ {\rm s.t.}, }   \| \theta_{k-1}^{(i_{k-1})}(E^*)+x_k\cdot\omega-\theta_{k-1}^{(i_{k-1}')}(E^*)\|_\T\leq \bar{\delta}_{k-1} .
	\end{equation}   Since $[l_{k_1-1}^\alpha,  l_{k_1}^\frac{1}{\alpha}]\cap [l_{k_2-1}^\alpha,  l_{k_2}^\frac{1}{\alpha}]=\emptyset $ for  $k_1\neq k_2$, we have $x_{k_1}\neq x_{k_2}$  for  $k_1\neq k_2$. Let $\bar{k}:=\min\{ [\ln(n+1)+1,n+1]\cap \Z\}.$  From Hypothesis \ref{h1}, we know that  for $k\geq \bar{k}$, all the numbers $\theta_{k-1}^{(i)} $ belong to  $\bigcup_{1\leq i\leq \kappa_{\bar{k}-1}}\widetilde{\mathcal{N}}_{\bar{k}-1}^{(i)}  $, a union of at most $M$ many $\widehat{\delta}_{\bar{k}-1}$-neighborhood. By  the   pigeonhole principle, there exist $k_1\neq k_2 \in [\ln(n+1)+1,n+1]$ such that \eqref{cunzai} holds for $k_1$ and $k_2$ and moreover,   $\theta_{k_1-1}^{(i_{k_1-1})}(E^*)$ and $\theta_{k_2-1}^{(i_{k_2-1})}(E^*)$ belonging  to the same $\widehat{\delta}_{\bar{k}-1}$-neighborhood, together with  $\theta_{k_1-1}^{(i_{k_1-1}')}(E^*)$ and $\theta_{k_2-1}^{(i_{k_2-1}')}(E^*)$ belonging to the same $\widehat{\delta}_{\bar{k}-1}$-neighborhood, i.e., 
	$$\|\theta_{k_1-1}^{(i_{k_1-1})}(E^*)-\theta_{k_2-1}^{(i_{k_2-1})}(E^*)\|_\T , \ \ \|\theta_{k_1-1}^{(i'_{k_1-1})}(E^*)-\theta_{k_2-1}^{(i'_{k_2-1})}(E^*)\|_\T\leq       2\widehat{\delta}_{\bar{k}-1} .   $$
	This together with \eqref{cunzai} yields \begin{equation}\label{1113}
		\|(x_{k_1}-x_{k_2})\cdot\omega \|_\T\leq 5\widehat{\delta}_{\bar{k}-1}.
	\end{equation} Since  $\omega\in 	{\rm DC}_{\tau, \gamma}$,   $0\neq \|x_{k_1}-x_{k_2}\|_1\leq l_{n+1} \leq l_{\bar{k}-1}^{8^{n+2-\bar{k}}}, $ we have $$ \ln \|(x_{k_1}-x_{k_2})\cdot\omega \|_\T\geq \ln (\gamma/l_{n+1}^\tau)\geq  \ln \gamma- \tau 8^{n+2-\bar{k}}  \ln l_{\bar{k}-1}. $$ 
	But \text{$\ln (5\widehat{\delta}_{\bar{k}-1}) \leq \ln 5- l_{\bar{k}-1}^{1/3}$ and  $l_{\bar{k}-1}\geq (\ln|\delta_0|)^{4^{\bar{k}-1}}\geq   (\ln|\delta_0|)^{(n+1)^{\ln4 }/4}  \gg (8^{n+2})^3  $, } which contradicts  \eqref{1113}. Thus \eqref{1607} is proved. 
	
	Define  
	\begin{align}
		K^{(i)}_{n+1}(\theta,E)&:=\prod_{\nu=1}^{{J^{(i)}}}\left(E-\widetilde{E}_n^{(i,\nu)}(\theta) \right)  \text{ (cf. \eqref{123})},\nonumber \\
		F^{(i)}_{n+1}(\theta,E)&:= K^{(i)}_{n+1}(\theta,E)/Q_{n+1}^{(i)}(\theta,E)=\prod_{\nu=\iota_{n}^{(i)}+1}^{{J^{(i)}}}\left(E-\widetilde{E}_n^{(i,\nu)}(\theta) \right) ,\label{shang}
	\end{align}
	which are also analytic in $\theta\in {\rm Qdp}(\widetilde{\mathcal{N}}_n^{(i)})$. 
	For $\star,\bullet\in \{P,Q,K,F\}$,	we define the   $E$-resultant 
	\begin{align}
		R_{n+1}^{(i,i')}(\theta,\phi,\star,\bullet):&= \operatorname{Res}_E(\star_{n+1}^{(i)}(\theta,E),\bullet_{n+1}^{(i')}(\theta+\phi,E)).
	\end{align} Hence ${R}_{n+1}^{(i,i')}(\theta,\phi,P,P)={R}_{n+1}^{(i,i')}(\theta,\phi).$ 
	We will first prove the  lower bound estimates  for $	R_{n+1}^{(i,i')}(\theta,\phi,K,K)$  hence  for the factor $	R_{n+1}^{(i,i')}(\theta,\phi,Q,Q)$,  and then show  $|\partial^l_\theta( R_{n+1}^{(i,i')}(\theta,\phi)-	R_{n+1}^{(i,i')}(\theta,\phi,Q,Q))|$ is small. Thus the  estimates also hold for $R_{n+1}^{(i,i')}(\theta,\phi) $. 
	\begin{itemize}
		\item  $i\neq i'$: In this case, $\theta_n^{(i)}(E^*)$ and $\theta_n^{(i')}(E^*)$ belong to different equivalence classes. Thus by \eqref{fenkn} and \eqref{fanz}, we have   \begin{equation}\label{133n}
			\|\theta_n^{(i)}(E^*)+z\cdot\omega-\theta_n^{(i')}(E^*)\|_\T>g^{(M^4)}(\bar{\delta}_n),  \quad \forall z\in Q_{l_{n+1}^{\alpha}}.  
		\end{equation}   Since $\theta\in \widetilde{\mathcal{N}}_n^{(i)}\cap\T$ and $\theta+\phi\in \widetilde{\mathcal{N}}_n^{(i')}\cap\T$, we have \begin{equation}\label{14n}
			\|\theta_n^{(i)}(E^*)+\phi-\theta_n^{(i')}(E^*)\|_\T\lesssim \widehat{\delta}_n. 
		\end{equation}
		Since  $\Omega_n^{(i)},\Omega_n^{(i')}\subset Q_{l_{n+1}^{1/\alpha}}$ and  $g^{(M^4)}(\bar{\delta}_n)\gg  \widehat{\delta}_n$,  \eqref{133n} together with \eqref{14n} implies 
		\begin{equation}\label{fafn}
			\|\phi+(y-x)\cdot \omega\|_\T>\widehat{\delta}_n, \quad \forall  x\in \Omega_n^{(i)},y\in \Omega_n^{(i')}. 
		\end{equation}Recalling \eqref{!}, we assume $x\in \Omega_n^{(i)},y\in \Omega_n^{(i')} $ and  $\theta + x\cdot \omega \in \widetilde{\mathcal{N}}_{n-1}^{(i_x)}, \theta+\phi +y\cdot\omega \in \widetilde{\mathcal{N}}_{n-1}^{(i_y)} $.    Thus by inductive Hypothesis \eqref{lowbu=} and \eqref{low=},  we have 
		\begin{equation}\label{2120n}
			\max _{0 \leq l \leq sM^{2t_{n}}}|\partial^l_\theta 	R_{n}^{(i_x,i_y)}(\theta+x\cdot\omega,\phi+(y-x)\cdot\omega)|  \geq\delta_{n-1} \|\phi+(y-x)\cdot\omega\|_\T^M\geq   \widehat{\delta}_n^{M+1}.
		\end{equation}
		By inductive Hypothesis \eqref{hup}, we have 
		\begin{equation}\label{2121n}
			|\partial^l_\theta 	R_{n}^{(i_x,i_y)}(\theta+x\cdot\omega,\phi+(y-x)\cdot\omega)|   \leq  C l!\widehat{\delta}_{n-1}^{-l}, \quad  \forall l \geq 0. 
		\end{equation}
		Since $K^{(i)}_{n+1}(\theta,E)=\prod_{x\in \Omega_n^{(i)}}P_n^{(i_x)}(\theta+x\cdot\omega,E) \text{ (cf. \eqref{123})},  $ we have 
		\begin{equation}\label{701}
			{R}_{n+1}^{(i,i')}(\theta,\phi,K,K)=\prod_{x\in \Omega_n^{(i)}, \\ y\in \Omega_n^{(i')}} R_{n}^{(i_x,i_y)}(\theta+x\cdot\omega,\phi+(y-x)\cdot\omega).
		\end{equation}
		Applying Lemma \ref{eli} to \eqref{701} with  $ J=|\Omega_n^{(i)}|\cdot|\Omega_n^{(i')}|$, $u_j(\theta)= R_{n}^{(i_x,i_y)}(\theta+x\cdot\omega,\phi+(y-x)\cdot\omega)$, $m_j=sM^{2t_n}$,  $N=s  M^{2t_n} |\Omega_n^{(i)}|\cdot|\Omega_n^{(i')}| \leq sM^{2t_{n+1}}$, $C_1=C(sM^{2t_{n+1}})!\widehat{\delta}_{n-1}^{-sM^{2t_{n+1}}}$ and  $\beta=   \widehat{\delta}_n^{M+1}$,  and recalling \eqref{cis} yields   
		\begin{align}
			\nonumber \max _{0 \leq l \leq sM^{2t_{n+1}}}| \partial^l_\theta 			{R}_{n+1}^{(i,i')}(\theta,\phi,K,K)| &	\geq  \left(\frac{1}{e}\right)^{M^{Cn^C}} \left(\frac{  \widehat{\delta}_n^{M+1}}{((Cn^C)!\widehat{\delta}_{n-1}^{-Cn^C})^2}\right)^{M^{Cn^C}} \\
			& \geq    \delta_n^{1/2}. \label{11n}
		\end{align}
		In \eqref{11n}, we used $M^{Cn^C} |\ln \widehat{\delta}_n|\leq M^{Cn^C}|\ln\delta_n|^\frac{1}{\alpha}\ll|\ln \delta_n|$ since $|\ln\delta_n|=l_n^{2/3}\geq |\ln\delta_0|^{\frac23 \cdot 4^n}$.	 
	 Similar to  \eqref{upp}, we have 
		\begin{equation}\label{fans}
			|\partial^l_\theta  R_{n+1}^{(i,i')}(\theta,\phi,\star,\bullet)|\leq   l! C_{v}'\widehat{\delta}_{n}^{-l}, \quad \forall l\geq 0,\  \star,\bullet\in \{Q,F\}
		\end{equation}  with 
		$C_{v}':=(M^2(2+2C_v M^2))^{2M^2}. $ 
	Combining 	 \eqref{11n},  \eqref{fans},  \eqref{shang} and the equation
	 ${R}_{n+1}^{(i,i')}(\theta,\phi,K,K)={R}_{n+1}^{(i,i')}(\theta,\phi,Q,Q)\cdot {R}_{n+1}^{(i,i')}(\theta,\phi,F,F)\cdot {R}_{n+1}^{(i,i')}(\theta,\phi,Q,F)\cdot {R}_{n+1}^{(i,i')}(\theta,\phi,F,Q)$   	 yields 
		\begin{align}
			\nonumber & \max _{0 \leq l \leq sM^{2t_{n+1}}}| \partial^l_\theta 			{R}_{n+1}^{(i,i')}(\theta,\phi,Q,Q)|	\\ 
			\nonumber	\geq&  \max _{0 \leq l \leq sM^{2t_{n+1}}}| \partial^l_\theta 			{R}_{n+1}^{(i,i')}(\theta,\phi,K,K)| ( 4^{Cn^C}  (Cn^C)! C_{v}'\widehat{\delta}_{n}^{-Cn^C})^{-1}\\
			\geq &   \delta_n^{1/2}\cdot ( 4^{Cn^C}  (Cn^C)! C_{v}'\widehat{\delta}_{n}^{-Cn^C})^{-1}\geq       2 \delta_n. \label{11n3}
		\end{align}
		By Proposition \ref{coen} and a similar argument to the proof  of \eqref{1018}, we have  $$|R_{n+1}^{(i,i')}(w,\phi)-R_{n+1}^{(i,i')}(w,\phi,Q,Q)|\leq C e^{-\frac34 l_n}, \quad \forall w\in   D_\T(\theta,\widehat{\delta}_n). $$
		Bounding  the derivatives by the  Cauchy integral formula	yields
		\begin{equation}\label{33n}
			\max _{0 \leq l \leq sM^{2t_{n+1}}}| \partial^l_\theta  (R_{n+1}^{(i,i')}(\theta,\phi)-R_{n+1}^{(i,i')}(\theta,\phi,Q,Q))| \leq \delta_n^{-2} e^{-\frac34 l_{n}}\leq \delta_n.
		\end{equation}
	Now \eqref{jie1} follows from 	\eqref{11n3} and \eqref{33n}.
		\item  $i= i'$:  In this case, $\|\phi\|_\T\leq 2\widehat{\delta}_n$.  Thus for $x\neq y\in \Omega_n^{(i)}$,  we have 	\begin{equation}\label{211n}
			\|\phi+(y-x)\cdot \omega\|_\T\geq \|(y-x)\cdot\omega \|_\T- \|\phi\|_\T \geq \widehat{\delta}_n.
		\end{equation} Assume $ \|\phi\|_\T\neq 0$.  Write 
		\begin{align}
		\label{958}	\|\phi\|_\T^{-J^{(i)}}R_{n+1}^{(i,i)}(\theta,\phi,K,K) =  \prod_{\substack{ x,y \in \Omega_n^{(i)}\\  x\neq   y} }R_{n}^{(i_x,i_y)}(\theta+x\cdot\omega,\phi+(y-x)\cdot\omega) \\
			\times 		  \prod_{x\in  \Omega_n^{(i)}} \left( \|\phi\|_\T^{-\iota_{n-1}^{(i_x)}}(R_{n}^{(i_x,i_x)}(\theta+x\cdot\omega,\phi))\right) .\label{apl}
		\end{align}
		By \eqref{211n}, the derivatives of  the  factors  in the first product (i.e., the right hand side of \eqref{958})  can be estimated   by  similar arguments to those of  \eqref{2120n} and \eqref{2121n}.   For the factors  in the second  product (i.e., \eqref{apl}),   by inductive Hypothesis \eqref{low=}, we have  
		$$	\max _{0 \leq l \leq sM^{2t_n}}|\|\phi\|_\T^{-\iota_{n-1}^{(i_x)}}\partial^l_\theta  R_{n}^{(i_x,i_x)}(\theta +x\cdot\omega,\phi)|\geq \delta_{n-1}.  $$ 
	Using  the same argument as in the proof of \eqref{ys}, we have 
		\begin{align*}
			\|\phi\|_\T^{{-\iota_{n-1}^{(i_x)}}} 	| R_{n}^{(i_x,i_x)}(w+x\cdot\omega,\phi) 
			| \lesssim \widehat{\delta}_n^{-M},   \quad \forall w\in   D_\T(\theta,\widehat{\delta}_n).  	\end{align*}
	Bounding  the derivatives by  the Cauchy integral formula yields \begin{equation} \label{1057}
			\|\phi\|_\T^{{-\iota_{n-1}^{(i_x)}}}|	\partial^l_\theta  R_{n}^{(i_x,i_x)}(\theta+x\cdot\omega,\phi)|\lesssim  l!\widehat{\delta}_n^{-(M+l)}, \quad  \forall l \geq 0. 
		\end{equation}
Using the  factors estimates above and applying Lemma \ref{eli} to \eqref{958} and \eqref{apl} yields
		\begin{equation}\label{111n}
			\|\phi\|_\T^{-J^{(i)}} \max _{0 \leq l \leq sM^{2t_{n+1}}}| \partial^l_\theta 			{R}_{n+1}^{(i,i)}(\theta,\phi,K,K)| \geq    \delta_n^{1/2}. 
		\end{equation}
		Similar to \eqref{1057}, we have 
		\begin{equation} \label{10571}
			\|\phi\|_\T^{{-(J^{(i)}-\iota_{n}^{(i)}})}|	\partial^l_\theta  R_{n+1}^{(i,i)}(\theta,\phi,F,F)|\lesssim  l!\widehat{\delta}_n^{-(M^2+l)}, \quad  \forall l \geq 0. 
		\end{equation}
		By \eqref{fans}, \eqref{10571}, \eqref{111n} and using   same argument as in  the proof of \eqref{11n3}, we have 
		\begin{equation}\label{111nf}
			\|\phi\|_\T^{-\iota_{n}^{(i)}} \max _{0 \leq l \leq sM^{2t_{n+1}}}| \partial^l_\theta 			{R}_{n+1}^{(i,i)}(\theta,\phi,Q,Q)| \geq    2\delta_n. 
		\end{equation}
Using   the same argument as in the proof of \eqref{333}, we have \begin{equation}\label{333n}
			\|\phi\|_\T^{-\iota_{n}^{(i)}} \max _{0 \leq l \leq sM^{2t_{n+1}}}| \partial^l_\theta (		(R_{n+1}^{(i,i)}(\theta,\phi)-	{R}_{n+1}^{(i,i)}(\theta,\phi,Q,Q))| \leq   \delta_n^{-1} e^{-l_n^{3/4}} \leq \delta_n .
		\end{equation} 
	Now  \eqref{jie2} follows from 	\eqref{111nf} and \eqref{333n}.
	\end{itemize}
\end{proof}
\subsection{Local stability of the scale parameters}\label{lc}
As mentioned in Remark \ref{123123}, the  $l_{n+1}, B_{n+1}$ in Proposition \ref{scalen} and \S \ref{Bn}  may vary with $E^*$.  In this section, we show that the $l_{n+1}, B_{n+1}$ associated with  $E^*$ can be used for the resonance analysis of  any $\widetilde{E}^*\in D(E^*,h(\delta_{n} ))$  as established  in the previous sections. The key observation is the  stable dependence  of the roots $\theta_n^{(i,j)}(E^*)$ on $E^*$  in    the following Propositions \ref{sta} and \ref{sta0}:  
\begin{prop}\label{sta}
Let  $n\geq 1$ and  let $B_n$ and  $A_n^{(i)}$  be  constructed in \S \ref{Bn} (with $n+1$ replaced by $n$) for   $E^*$. Then for any  fixed  \begin{equation}\label{759}
\widetilde{E}^*\in D(E^*,h(\delta_{n} )),
	\end{equation}
the $\theta$-analytic function 	$s_{A_n^{(i)}}(\theta,\widetilde{E}^*)$ has $|\mathcal{C}_{n-1}^{(i)}|$ many  roots  in $\widetilde{\mathcal{N}}_{n-1}^{(i)}$, denoted by $\{\theta_n^{(i,j)}(\widetilde{E}^*)\}_{j=1}^{|\mathcal{C}_{n-1}^{(i)}|}$. 	Moreover,  up to a permutation, we have 
\begin{equation}\label{1215}
	|\theta_n^{(i,j)}(\widetilde{E}^*)-\theta_n^{(i,j)}(E^*)|\leq \delta_n, \quad \forall 1\leq j\leq |\mathcal{C}_{n-1}^{(i)}|.
\end{equation}
\end{prop}
\begin{proof}
	From  \eqref{su}, \eqref{n+1g}, \eqref{759} and the  resolvent identity, 
	we have  $$  |s_{A_n^{(i)}}(\theta,\widetilde{E}^*)-s_{A_n^{(i)}}(\theta,E^*)|\leq C^{M^n}\widehat{\delta}_{n-1}^{-2}\cdot  h(\delta_n) \leq   h(\delta_n)^{\frac12 }.  $$
	Thus by Proposition \ref{654n}, we have 
\begin{equation}\label{gsg}
		 s_{A_{n}^{(i)}}(\theta,\widetilde{E}^*)\overset{\delta_{n-1}}{\sim} 
\prod_{j=1}^{|\mathcal{C}_{n-1}^{(i)}|}\left(\theta-\theta_{n}^{(i,j)}(E^*) \right)+O( h(\delta_n)^{\frac13 }),\quad \forall \theta\in \widetilde{\mathcal{N}}_{n-1}^{(i)}.
\end{equation}
The existence of $\{\theta_n^{(i,j)}(\widetilde{E}^*)\}_{j=1}^{|\mathcal{C}_{n-1}^{(i)}|}$ follows from Rouch\'e theorem.  Note that  $\theta_n^{(i,j)} \in D_\T(\Re(\theta_{n-1}^{(i)}(E^*)),2\widehat{\delta}_{n-1}^2)$. Applying  the Weierstrass preparation argument as in the proof of Proposition \ref{coe} to the right hand side of \eqref{gsg} yields that   the coefficients of the $\theta$-polynomials 
\begin{equation*}
	\mathcal{P}(\theta,\widetilde{E}^*):=	\prod_{j=1}^{|\mathcal{C}_{n-1}^{(i)}|}\left(\theta -\theta_{n}^{(i,j)}(\widetilde{E}^*) \right)=:\sum_{k=0}^{|\mathcal{C}_{n-1}^{(i)}|} a_{k}(\widetilde{E}^*)\theta ^{|\mathcal{C}_{n-1}^{(i)}|-k}
\end{equation*}
and $$		\mathcal{P}(\theta,E^*):=	\prod_{j=1}^{|\mathcal{C}_{n-1}^{(i)}|}\left(\theta -\theta_{n}^{(i,j)}(E^*) \right)=:\sum_{k=0}^{|\mathcal{C}_{n-1}^{(i)}|} a_{k}(E^*)\theta ^{|\mathcal{C}_{n-1}^{(i)}|-k}$$   satisfy  $$|a_{k}(\widetilde{E}^*)-a_{k}(E^*)|\leq h(\delta_n)^\frac14, \quad  \forall 0 \leq k\leq |\mathcal{C}_{n-1}^{(i)}|    .$$ 
Then  \eqref{1215} follows from  the following  lemma  from  \cite{rao}.
\begin{lem}[{\cite[Theorem 1]{rao}}]
	Let $f$ and $g$ be two monic polynomials of degree $n$ with complex coefficients $a_1, \cdots, a_n$ and $b_1, \cdots, b_n$ respectively. Let $\alpha_1, \cdots, \alpha_n$ and $\beta_1, \cdots, \beta_n$ be their respective roots:
	$$
	\begin{aligned}
		& f(z)=z^n+a_1 z^{n-1}+\cdots+a_n=\prod_{i=1}^n\left(z-\alpha_i\right), \\
		& g(z)=z^n+b_1 z^{n-1}+\cdots+b_n=\prod_{i=1}^n\left(z-\beta_i\right).
	\end{aligned}
	$$
Let 	$$\Gamma=\max _{1 \leq  k \leq  n}\left(\left|a_k\right|^{1 / k},\left|b_k\right|^{1 / k}\right), \quad \gamma=2 \Gamma .$$
Then the roots of $f$ and $g$ can be enumerated as $\alpha_1, \ldots, \alpha_n$ and $\beta_1, \ldots, \beta_n$ in such a way that
$$
\max _i\left|\alpha_i-\beta_i\right| \leqslant 4 \times 2^{-1 / n}\left\{\sum_{k=1}^n\left|a_k-b_k\right| \gamma^{n-k}\right\}^{1 / n} .
$$
\end{lem} 	 
\end{proof}
For $n=0$, using  the identification  $\theta\mapsto e^{2\pi i \theta}$ (see the proof of  Lemma \ref{gen} for a similar argument)  and  a similar argument to that in the proof of Proposition  \ref{sta} yields 
	\begin{prop}\label{sta0}
	For any 	fixed $\widetilde{E}^*\in D(E^*,h(\delta_{0} ))$, the $\theta$-analytic function 
	$v(\theta)-\widetilde{E}^*$ has 	$m_1= \sum _{i=1}^{\kappa_0}|\mathcal{C}_0^{(i)}| $ many roots   in $\T_{\widehat{\delta}_0}$,  denoted by $\{\theta_0^{(i)}(\widetilde{E}^*)\}_{i=1}^{m_1}$.   Moreover,  up to a permutation, we have 
		$$\|\theta_0^{(i)}(\widetilde{E}^*)-\theta_0^{(i)}(E^*)\|_\T \leq \delta_0, \quad \forall 1\leq i\leq m_1.$$  
	\end{prop}
\begin{rem}\label{TY}
	By Propositions \ref{sta0} and \ref{sta},  for  any  $\widetilde{E}^*\in  D(E^*,h(\delta_{n} ))$, up to a permutation,  the distance  between the  roots corresponding to  $\widetilde{E}^*$ and $E^*$ is smaller than $\delta_n\ll\bar{\delta}_n$, which is  negligible in the choice of  $l_{n+1}$ (cf. Propositions \ref{scale},  \ref{scalen}), $B_{n+1} $ and $A_{n+1}^{(i)}$ (cf. \S \ref{Bn})  hence the resonance analysis thereafter.    
	Thus in \S \ref{HP}, fixing   the  $l_r(E^*),B_r(E^*)$  associated with  $E^*$ (cf. Remark \ref{123123})  for   $r\leq n$,  the inductive Hypotheses \ref{h1}--\ref{h4}   still hold   with $E^*$ replaced by any  $\widetilde{E}^*\in  D(E^*,h(\delta_{n-1} ))$.  
\end{rem}
\section{Anderson localization}\label{AL}
In this section, we present  the proof of Theorem \ref{maina}. Let $\varepsilon_0$ be sufficiently small so  that Proposition \ref{1832} and Theorem \ref{key2} hold,   and thus  the 	inductive Hypotheses \ref{h1}--\ref{h4}  hold    for all $n\geq 1$. 

First, we identify the phases where Anderson localization holds. 
We have to remove a zero-measure set of  phases   where double resonance  happens  infinitely often  times. Recall the  notation $l_n(E^*)$  (cf. Remark \ref{123123}),  so that $\delta_n(E^*)=e^{-l_n(E^*)^{2/3 }}$. 

Let $$J_0:=[\inf(\operatorname{Spec}(H)),\sup(\operatorname{Spec}(H))].$$ 
We define the $1$-step   $ E^*$-relevant interval by   $$J_1(E^*):=[E^*-h(\delta_{0}),E^*+h(\delta_{0})].$$  
 We cover the interval $J_0$ by   
a finite sequence of $1$-step intervals $\{J_1^{(k)}\}_{k=1}^{K}$ such that
\begin{equation*}
	J_1^{(k)}:=J_1(E^{(k)}_{1}),\  \ E^{(k)}_{1}\in J_0,  \ \   J\subset\bigcup_{k=1}^{K} J_1^{(k)}, \ \   \sum_{k=1}^{K} \operatorname{Leb}(J_1^{(k)})\leq 3\operatorname{Leb}(J_0). 
\end{equation*}
(Setting   $E^{(1)}_{1} =\inf(\operatorname{Spec}(H))$ and   $ E^{(k+1)}_{1}=E^{(k)}_{1}+h(\delta_{0})$ yields  the  desired covering.)   

Inductively,  for each $n\geq 1$, we cover the interval 
\begin{equation}\label{1239}
	J_n^{(\vec{k})}:=J_n(E_n^{(\vec{k})}):=[E_n^{(\vec{k})}-h(\delta_{n-1}(E_n^{(\vec{k})})),E_n^{(\vec{k})}+h(\delta_{n-1}(E_n^{(\vec{k})}))], 
\end{equation} 
where $\vec{k}=(k_1,\cdots,k_n)$ is  an $n$-tuple, 
by   a finite sequence of $(n+1)$-step intervals $\{J_{n+1}^{(\vec{k},k')}\}_{k'=1}^{K(\vec{k})}$ such that
\begin{align}
\nonumber	&J_{n+1}^{(\vec{k},k')}=J_{n+1}(E^{(\vec{k},k')}_{n+1}) \text{ (cf.  \eqref{1239} with   $n$ replaced by    $n+1$)},  \ \      E^{(\vec{k},k')}_{n+1}\in J_n^{(\vec{k})},   \\  &J_n^{(\vec{k})} \subset\bigcup_{k'=1}^{K(\vec{k})} 	J_{n+1}^{(\vec{k},k')}, \quad  \sum_{k'=1}^{K(\vec{k})} \operatorname{Leb}(J_{n+1}(E^{(\vec{k},k')}_{n+1}))\leq  3 \operatorname{Leb}(J_n^{(\vec{k})}). \label{covern}
\end{align}
Since $ E^{(\vec{k},k')}_{n+1}\in J_n^{(\vec{k})}\subset  D(   E_n^{(\vec{k})}, h(\delta_{n-1}(E_n^{(\vec{k})}))) $,  
 by Remark \ref{TY}, we can choose   
\begin{equation}\label{jiazhu}
	\text{ $l_r(E^{(\vec{k},k')}_{n+1})=l_r(E_n^{(\vec{k})})$ for any $1\leq r\leq n$} .
\end{equation}

\vspace{2mm}
\begin{center}
	In what follows, we regard $E^{(\vec{k})}_{n}$ as \textbf{playing the role of} the fixed energy $E^*$ in \S \ref{MSA}.
\end{center}
\vspace{2mm}
 For each $E^{(\vec{k})}_{n}$,  we define the $n$-step $E^{(\vec{k})}_{n}$-relevant  double resonant phases set by 
$${\rm DR}_n(E^{(\vec{k})}_{n}):=\bigcup_{E\in J_n^{(\vec{k})}}\left \{\theta\in \T :\ \exists   x\neq y, \  {\rm s.t.}, \ x,y\in \widetilde{S}_n(\theta,E,\delta_n(E_n^{(\vec{k})}))\bigcap Q_{l_n(E^{(\vec{k})}_{n})^{30} }\right \},$$
and define the $n$-step double resonant phases set by 
\begin{equation}\label{ndr}
	{\rm DR}_n:=\bigcup_{\vec{k}}{\rm DR}_n(E^{(\vec{k})}_{n}). 
\end{equation}
Finally, we define the infinitely often double resonant phases set by
$$\Theta:= \bigcap_{l\geq 1}\bigcup_{n\geq l }{\rm DR}_n.$$

\begin{prop}\label{0ce}We have $\operatorname{Leb}(\Theta)=0$. 
\end{prop}
\begin{proof}
We divide the proof into several lemmas.  
\begin{lem}\label{baohan}  We have 
	\begin{align}
	\nonumber	{\rm DR}_n(E^{(\vec{k})}_{n})\subset \bigcup_{\|x\|_1\leq l_n(E^{(\vec{k})}_{n})^{30}}   \left(\bigcup_{i=1}^{\kappa_{n-1}}\left\{\theta\in \widetilde{\mathcal{N}}_{n-1}^{(i)}\cap\T:\ \exists z\in Q_{2l_n(E^{(\vec{k})}_{n})^{30}}\setminus\{o\}, i'\in [1,\kappa_{n-1}],	   \right. \right. \\
	\left.	\left.  \text{ {\rm s.t.}, }  \theta+z\cdot\omega\in \widetilde{\mathcal{N}}_{n-1}^{(i')}\cap\T\text{ and }      |	R_n^{(i,i')}(\theta,z\cdot\omega)|\leq \delta_n(E_n^{(\vec{k})})^{\frac12 }                         \right\}-x\cdot\omega\right)   = : \widetilde{	{\rm DR}}_n(E^{(\vec{k})}_{n}). \label{jihe} 	\end{align}
\end{lem}
\begin{proof}[Proof of Lemma \ref{baohan}]
	 By the definition of  $\widetilde{S}_n(\theta,E,\delta_n(E_n^{(\vec{k})}))$ (cf. Definition \ref{swan}),  for  $$\theta\in   {\rm DR}_n(E^{(\vec{k})}_{n}),$$  there exist $E\in J_n^{(\vec{k})}$,  $x\neq y\in Q_{l_n(E^{(\vec{k})}_{n})^{30}}$ and $i,i'\in [1,\kappa_{n-1}]$  such that \begin{equation}\label{zaide}
	 	 \theta+x\cdot\omega \in \widetilde{\mathcal{N}}_{n-1}^{(i)} \cap\T , \ \  \theta+y\cdot\omega \in \widetilde{\mathcal{N}}_{n-1}^{(i')}  \cap\T ,   
	 \end{equation}
 and \begin{equation}\label{ronghe}
	 \min_{1\leq j\leq \iota_{n-1}^{(i)}}|E_n^{(i,j)}(\theta+x\cdot\omega )-E|<\delta_n(E_n^{(\vec{k})}) , \ \  \min_{1\leq j'\leq \iota_{n-1}^{(i')}}|E_n^{(i',j')}(\theta+y\cdot\omega )-E|<\delta_n(E_n^{(\vec{k})}). 
 \end{equation}
	 Let  \begin{equation}\label{2033}
	 \text{	$z:=y-x$ \ \  and \ \  $\theta':=\theta+x\cdot\omega$.}
	 \end{equation} It follows from \eqref{zaide} that $\theta'\in \widetilde{\mathcal{N}}_{n-1}^{(i)} \cap\T  $ and $z\in Q_{2l_n(E^{(\vec{k})}_{n})^{30}}\setminus\{o\}$ satisfying  
	$ \theta'+z\cdot\omega \in \widetilde{\mathcal{N}}_{n-1}^{(i')}   \cap\T$.  Moreover, it follows from \eqref{ronghe} that 
\begin{align*}
		|R_n^{(i,i')}(\theta',z\cdot\omega)| =\prod_{j,j'}|E_n^{(i,j)}(\theta')-E_n^{(i',j')}(\theta'+z\cdot\omega)| \leq \delta_n(E_n^{(\vec{k})})^{\frac12 }   .
\end{align*}
Thus \begin{align*}
	\theta'\in  \bigcup_{i=1}^{\kappa_{n-1}}\left\{\theta\in \widetilde{\mathcal{N}}_{n-1}^{(i)}\cap\T:\ \exists z\in Q_{2l_n(E^{(\vec{k})}_{n})^{30}}\setminus\{o\}, i'\in [1,\kappa_{n-1}], 	\text{ {\rm s.t.},}    \right. \\
	\left. \theta+z\cdot\omega\in \widetilde{\mathcal{N}}_{n-1}^{(i')}\cap\T\text{ and }      |	R_n^{(i,i')}(\theta,z\cdot\omega)|\leq \delta_n(E_n^{(\vec{k})})^{\frac12 }                         \right\}
\end{align*}
and then $\theta\in \widetilde{	{\rm DR}}_n(E^{(\vec{k})}_{n}) $ by \eqref{jihe} and \eqref{2033}.  
\end{proof}
	\begin{lem}\label{transha}
		$\operatorname{Leb}(\widetilde{	{\rm DR}}_n(E^{(\vec{k})}_{n})) \leq  h(\delta_{n-1}(E^{(\vec{k})}_{n}))^2$. 
	\end{lem}
	\begin{proof}
	Fix $i,i'\in [1,\kappa_{n-1}]$ and  $z\in Q_{2l_n(E^{(\vec{k})}_{n})^{30}}\setminus\{o\}$. We bound  the measure of 
	\begin{align}\label{yaoguji}
			\left\{\theta\in \widetilde{\mathcal{N}}_{n-1}^{(i)}\cap\T:\  \theta+z\cdot\omega\in \widetilde{\mathcal{N}}_{n-1}^{(i')}\cap\T \text{ and }    |	R_n^{(i,i')}(\theta,z\cdot\omega)|\leq \delta_n(E_n^{(\vec{k})})^{\frac12 }   \right\}  .
	\end{align}
 By $\omega\in 	{\rm DC}_{\tau, \gamma}$,  we have $\|z\cdot\omega\|_\T\gg \widehat{\delta}_{n-1}(E_n^{(\vec{k})})$. Thus if $\theta\in \widetilde{\mathcal{N}}_{n-1}^{(i)}\cap\T$  and   $\theta+z\cdot\omega\in \widetilde{\mathcal{N}}_{n-1}^{(i')}\cap\T$, it must be $i\neq i'$. It follows from  \eqref{hup}, \eqref{lowbu=} and \eqref{cishur} that 
		\begin{align*}
			 \max_{0\leq l\leq  	s M^{2t_{n}}+1}	|\partial_\theta^lR_n^{(i,i')}(\theta,z\cdot\omega)|&\leq  C (s M^{2t_{n}}+1)! (\widehat{\delta}_{n-1}(E_n^{(\vec{k})}))^{-(s M^{2t_{n}}+1)}, \\
			 			\max_{0 \leq l \leq sM^{2t_n}}| \partial^l_\theta 	R_n^{(i,i')}(\theta,z\cdot\omega)| 	&\geq  \delta_{n-1}(E_n^{(\vec{k})}),\\ 
			 			 sM^{2t_n}&\leq Cn^C. 
		\end{align*}
The measure of  \eqref{yaoguji} can be  bounded via the transversality of $R_n^{(i,i')}(\theta,z\cdot\omega)$ by the following  lemma from   \cite{Eli97}. 
		\begin{lem}[{\cite[Lemma 3]{Eli97}}]\label{tranle}
		Let $I\subset\T $ be an interval and $u(\theta)$ be a smooth function on $\theta\in I$ such that
		\begin{align*}
			\max_{0\leq l\leq N+1}\sup_{\theta\in I}|\partial_\theta^l u(\theta )|&\leq C_2, \\
			\max _{0 \leq l \leq  N}| \partial^l_\theta  u(\theta)| &\geq \beta, \quad  \forall \theta \in I.
		\end{align*}
		Then for any $\epsilon>0$, we have 
		$$\operatorname{Leb}(\{\theta\in I :\ |u(\theta)|\leq \epsilon  \})\leq2^{N+2}\left(\frac{2C_2|I|}{\beta}+1\right) \cdot \left(\frac{\epsilon}{\beta}\right)^{1/N}.   $$
	\end{lem}
	Setting  $u(\theta)= R_n^{(i,i')}(\theta,z\cdot\omega)$,  $N=sM^{2t_n}$, $C_2=C (s M^{2t_{n}}+1)! (\widehat{\delta}_{n-1}(E_n^{(\vec{k})}))^{-(s M^{2t_{n}}+1)}$, $\beta = \delta_{n-1}(E_n^{(\vec{k})})$ and $\epsilon=\delta_n(E_n^{(\vec{k})})^{\frac12 }  $ in  Lemma \ref{tranle} yields 
\begin{align*}\label{77}
		 \operatorname{Leb}(\eqref{yaoguji})&\leq 2^{Cn^C}(Cn^C)! (\widehat{\delta}_{n-1}(E_n^{(\vec{k})}))^{-Cn^C} (\delta_{n-1}(E_n^{(\vec{k})}) )^{-1} \left(\frac{\delta_n(E_n^{(\vec{k})})^{\frac12 }}{(\delta_{n-1}(E_n^{(\vec{k})}) )^{-1}}\right)^{\frac{1}{Cn^C}} \\ 
		 &\leq 
	  h(\delta_{n-1}(E^{(\vec{k})}_{n}))^3,  
\end{align*}
where we used  $|\ln \delta_{n}(E_n^{(\vec{k})})|=l_n(E_n^{(\vec{k})})^{\frac23}\geq l_{n-1}(E_n^{(\vec{k})})^{\frac23 \cdot4 }\geq |\ln\delta_0 |^{\frac23 \cdot4^n } \gg Cn^C$.   
	Since there are at most  $CM^2l_n (E^{(\vec{k})}_{n})^{30}$ many choices of $(i,i',z)$ in \eqref{yaoguji}, we have   $$\operatorname{Leb}(\widetilde{	{\rm DR}}_n(E^{(\vec{k})}_{n}))\leq CM^2l_n (E^{(\vec{k})}_{n})^{60}  h(\delta_{n-1}(E^{(\vec{k})}_{n}))^3   \leq  h(\delta_{n-1}(E^{(\vec{k})}_{n}))^2. $$
 \end{proof}
	 	By Lemmas  \ref{baohan},  \ref{transha} and \eqref{1239},  \eqref{covern}, \eqref{ndr},  we have 
		\begin{align*}
			\operatorname{Leb}({\rm DR}_n)&\leq \sum _{\vec{k}} \operatorname{Leb}(\widetilde{	{\rm DR}}_n  (E^{(\vec{k})}_{n}))\\ 
			&\leq \sum _{\vec{k}} \frac{\operatorname{Leb}(\widetilde{	{\rm DR}}_n  (E^{(\vec{k})}_{n}))}{\operatorname{Leb}(J_n^{(\vec{k})})} \cdot \operatorname{Leb}(J_n^{(\vec{k})})\\
				&\leq\sum _{\vec{k}}  h(\delta_{n-1}(E^{(\vec{k})}_{n}))\cdot  \operatorname{Leb}(J_n^{(\vec{k})})\\
				&\leq \max_{\vec{k}} h(\delta_{n-1}(E^{(\vec{k})}_{n})) \sum _{\vec{k}}\operatorname{Leb}(J_n^{(\vec{k})})   \\
				& \leq e^{ -|\ln\delta_0|^{ \frac23\cdot 4^{n-1} }}  3^n  \operatorname{Leb}(J_0) . 
		\end{align*}	 
		Thus 	$\sum_{n=1}^{+\infty}\operatorname{Leb}({\rm DR}_n)<+\infty$ and  $\operatorname{Leb}(\Theta)=0$ follows from Borel-Cantelli Lemma.  
\end{proof}
Thus $\mathbb{T}\setminus \Theta $ is a set of full Lebesgue measure. We will prove $H(\theta)$  satisfies Anderson localization for any $\theta \in \mathbb{T}\setminus \Theta $.  For the remainder of this section, we fix $$\theta \in \mathbb{T}\setminus \Theta .$$ Thus there exists some $N_1\in \N^*$ such that\begin{equation}\label{1615}
\theta\notin {\rm DR}_n , \ \ \forall  n\geq N_1. 
\end{equation}
To prove that $H (\theta)$ satisfies Anderson localization, by Schnol's Lemma, it suffices    to prove that every generalized eigenfunction $\psi$ of $H (\theta)$ that satisfies  $ |\psi (x)|\leq  ( 1+\|x\|_1)^d$ in fact decays exponentially. We fix a generalized eigenvalue $E$ and its  corresponding nonvanishing  generalized eigenfunction $\psi$.  

By \eqref{covern},  there exists an  infinite  path  $\{\vec{k}_n\}_{n=1}^{+\infty}$ ($\vec{k}_n$ is  an $n$-tuple) such that  $\vec{k}_n$ coincides with the first $n$ component of  $\vec{k}_{n+1}$  for all $n\geq 1$ and 
 $$ E\in \bigcap_{n\geq 1} J_n^{(\vec{k}_n)}.   $$
\begin{lem}\label{ALL}
	There exists some $N_2\in \N^*$  such that for all $n \geq N_2$,  there is some $x_n\in Q_{l_n(E^{(\vec{k}_n)}_{n})^{2}} \cap  \widetilde{S}_n(\theta,E,\delta_n(E_n^{(\vec{k}_n)}))$.
\end{lem}
\begin{proof}
 For brevity, we omit the dependence of $l_n$ and $\delta_n$ on $E_n^{(\vec{k}_n)}$.  If the conclusion  fails,  then   there exists an increasing integer sequence $\{n_i\}_{i=0}^{\infty}$ such that \begin{equation}\label{beishang}
	Q_{l_{n_i}^2}\cap  \widetilde{S}_{n_i}(\theta,E,\delta_{n_i})=\emptyset .
\end{equation}
	By Remark \ref{gouzao},  there exists a   sequence of sets  $\{U_i\}_{i=1}^\infty$ satisfying \begin{itemize}
	\item  $Q_{20l_{n_i}}\subset  U_i \subset \{x\in \Z^d:\ \operatorname{dist}_1 (x,Q_{20l_{n_i}})\leq 50M^2l_{n_i}\} \subset 	Q_{l_{n_i}^2}.$
	\item  $U_i$ is $(\theta, E_{n_i}^{(\vec{k}_{n_i})})$-$n_i$-regular.
\end{itemize}
Moreover, it follows from \eqref{beishang} that  $U_i$ is $E$-type-$(\theta, E,\delta_{n_i})$-$n_i$-nonresonant.  By  Hypothesis \ref{h2} \textbf{(2)} together with  $E\in J_{n_i}^{(\vec{k}_{n_i})}\subset D(E_{n_i}^{(\vec{k}_{n_i})},h(\delta_{n_i-1}))$, we have 
\begin{equation}\label{1751}
			|G_{U_i}^{\theta,E}(x,y)|\leq e^{-\gamma_{n_i}\|x-y\|_1}, \quad \forall \|x-y\|_1\geq l_{n_i}^\frac{2}{\alpha+1}
\end{equation}
 Fix an   arbitrary $x\in \Z^d$.  Let $i$ be sufficiently large so that $\|x\|_1\leq l_{n_i}$.
It follows from the  Poisson formula that 
\begin{equation}\label{p}
	\psi(x)=-\varepsilon \sum_{(w,w')\in \partial_{\Z^d}U_i}G_{U_i}^{\theta,E}(x,w)\psi(w').
\end{equation}
For $w\in\partial_{\Z^d}^-U_i $, we have \begin{equation}\label{1754}
	 \|x-w\|_1\geq \|w\|_1-\|x\|_1\geq 20l_{n_i}-l_{n_i}>l_{n_i}^\frac{2}{\alpha+1}.
\end{equation}
Combining \eqref{1751},  \eqref{p} and \eqref{1754}  yields 
\begin{align}
	\nonumber|\psi(x)|&\leq \sum_{(w,w')\in \partial_{\Z^d}U_i}|G_{U_i}^{\theta,E}(x,w)|\cdot |\psi(w')|\\
	\nonumber&\leq \sum_{(w,w')\in \partial_{\Z^d}U_i} e^{-\gamma_{n_i}\|x-w\|_1}(1+\|w'\|_1)^d\\
	&\lesssim e^{-19l_{n_i}} l_{n_i}^{2d}\label{00}.
\end{align}
Taking $i\to\infty$  in \eqref{00} yields $\psi(x)=0$. Since $x$ is arbitrary, it follows that $\psi\equiv0$, which contradicts the  nonvanishing  assumption of $\psi$.
\end{proof}
Now we prove Theorem \ref{maina}. 
\begin{proof}[Proof of Theorem \ref{maina}] 
		Let  $n\geq \max  (N_1,N_2)=:N$. 
	Since  $ \theta\notin {\rm DR}_n$ (cf. \eqref{1615}), by the definition of ${\rm DR}_n$ (cf. \eqref{ndr})  and Lemma \ref{ALL},   we have 
	\begin{equation}\label{huan}
		   \widetilde{S}_n(\theta,E,\delta_n(E_n^{(\vec{k}_n)}))\cap  (Q_{l_n(E^{(\vec{k}_n)}_{n})^{30} }\setminus    Q_{l_n(E^{(\vec{k}_n)}_{n})^{2} }) =\emptyset. 
	\end{equation}
		By Remark \ref{gouzao},  there exists a set  $V_n$ satisfying 
	\begin{itemize}
		\item  $  ( Q_{ l_n(E^{(\vec{k}_n)}_{n})^{30}/2 }\setminus  Q_{2 l_n(E^{(\vec{k}_n)}_{n})^{2} })\subset V_n \subset (Q_{l_n(E^{(\vec{k}_n)}_{n})^{30} }\setminus    Q_{l_n(E^{(\vec{k}_n)}_{n})^{2} }).$
		\item $V_n$ is $(\theta, E_{n}^{(\vec{k}_{n})})$-$n$-regular.
	\end{itemize}
Moreover, it follows from \eqref{huan} that  $V_n$ is $E$-type-$(\theta, E,\delta_{n}(E_n^{(\vec{k}_n)}))$-$n$-nonresonant.
 By  Hypothesis \ref{h2} \textbf{(2)} together with  $E\in J_{n}^{(\vec{k}_{n})}\subset D(E_{n}^{(\vec{k}_{n})},h(\delta_{n-1}))$, we have 
\begin{equation}\label{1820}
	|G_{V_n}^{\theta,E}(x,y)|\leq e^{-\gamma_{n}\|x-y\|_1}, \quad \forall \|x-y\|_1\geq l_{n}^\frac{2}{\alpha+1}.
\end{equation}
By \eqref{jiazhu}, we have $$ l_{n+1}(E^{(\vec{k}_{n+1})}_{n+1}) \in  [l_{n}(E^{(\vec{k}_{n+1})}_{n+1})^4,l_{n}(E^{(\vec{k}_{n+1})}_{n+1})^8] = [l_{n}(E^{(\vec{k}_{n})}_{n})^4,l_{n}(E^{(\vec{k}_{n})}_{n})^8].$$
	 Let $x\in \Z^d$ such that $\|x\|_1>l_N(E^{(\vec{k}_N)}_{N})^3 $.  Then  there exists  some $n\geq N$ such that    
	 $$x\in  (Q_{ l_{n+1}(E^{(\vec{k}_{n+1})}_{n+1})^3 
	 }\setminus Q_{l_{n}(E^{(\vec{k}_{n})}_{n})^3})\subset (Q_{ l_{n}(E^{(\vec{k}_{n})}_{n})^{24} 
 }\setminus Q_{l_{n}(E^{(\vec{k}_{n})}_{n})^3}) .$$
It follows from the  Poisson formula that 
	\begin{equation}\label{f}
		\psi(x)=-\varepsilon \sum_{(w,w')\in \partial_{\Z^d}V_n}G_{V_n}^{\theta,E}(x,w)\psi(w').
	\end{equation}
	For  $x\in (Q_{ l_{n}(E^{(\vec{k}_{n})}_{n})^{24} 
	}\setminus Q_{l_{n}(E^{(\vec{k}_{n})}_{n})^3})$ and  $w\in \partial_{\Z^d}^-V_n$, we have 
	\begin{equation}\label{1824} 
		\|x-w\|_1\geq\min( l_n(E^{(\vec{k}_n)}_{n})^{30}/2-\|x\|_1,\|x\|_1-2 l_n(E^{(\vec{k}_n)}_{n})^{2} -1)
		\geq   \|x\|_1/2>l_{n}^\frac{2}{\alpha+1}.
	\end{equation}
Combining \eqref{1820}, \eqref{f}  and \eqref{1824} yields 
	\begin{align*}
		\nonumber|\psi(x)|&\leq \sum_{(w,w')\in \partial_{\Z^d}V_n}|G_{V_n}^{\theta,E}(x,w)|\cdot |\psi(w')|\\
		\nonumber&\leq \sum_{(w,w')\in \partial_{\Z^d}V_n} e^{-\gamma_{n}\|x-w\|_1}(1+\|w'\|_1)^d\\
		&\leq C l_n(E^{(\vec{k}_n)}_{n})^{60d}e^{-2\|x\|_1}\\
		&\leq C \|x\|_1^{20d}e^{-2\|x\|_1}\\
		&\leq e^{-\|x\|_1}.
	\end{align*}
\end{proof}
\begin{rem}
Our analysis can be modified to prove a better decay rate of $\psi$ involving  $|\ln\varepsilon|$:   $$|\psi(x)|\leq e^{-(1-o_{\varepsilon\to 0}(1))|\ln\varepsilon|\cdot\|x\|_1 },$$ where $o_{\varepsilon\to 0}(1)$  is a quantity going to $0$ when $\varepsilon$ goes to $0$. Note that when the dimension   $d=1$,   $(1-o_{\varepsilon\to 0}(1))|\ln\varepsilon|$ plays the role of a lower bound for the Lyapunov exponent of the transfer matrix  (cf. \cite[Theorem 2]{BG00}).   
 \end{rem}

\section{H\"older continuity of the IDS}\label{IDS}
 In this section, we present the proof of Theorem \ref{mainb}. 
Unlike the proof of Anderson localization, the proof of H\"older continuity of the IDS does not need transversality condition (cf.  Hypothesis \ref{h4}) to eliminate double resonance; it only requires the Green's function estimate  for a fixed $E^*$.   Thus we only require $v$ to satisfy Condition \ref{con1} so that the dependence of $\varepsilon_0$ on $c$ and $s$ is eliminated. 
  Let $\varepsilon_0$ be sufficiently small so  that  inductive Hypotheses \ref{h1}--\ref{h3} hold for all $n\geq 1$. 

Fix $\theta \in \mathbb{T}$,  $E^* \in [\inf(\operatorname{Spec}(H)),\sup(\operatorname{Spec}(H))] $ and let $\beta > 0$.  We will estimate the number of eigenvalues of $H_{Q_N}(\theta)$ belonging to $[E^*-\beta, E^*+\beta]$ as $N \rightarrow +\infty$. To this end, we first introduce the following Lemma \ref{HOl}, which provides a relation between the Green's function and the number of eigenvalues of a matrix inside the energy interval $[E^*-\beta, E^*+\beta]$.
\begin{lem}[{\cite[Lemma 6.1]{CSZ23}}]\label{HOl} 
	 Let $H$ be a self-adjoint operator on $\mathbb{Z}^d$ and $\Lambda \subset \mathbb{Z}^d$ be a finite set. Fix $E^* \in \mathbb{R}$. Assume there exist some  $\beta'>0$ and  another  set $\Lambda^{\prime} \subset \mathbb{Z}^d$ such that $\|\left(H_{\Lambda^{\prime}}-E^*\right)^{-1}\|  \leq (2 \beta')^{-1}$ and  $\left|\Lambda \triangle \Lambda^{\prime}\right| \leq L$, where  $\Lambda \triangle   \Lambda^{\prime}=(\Lambda\setminus  \Lambda^{\prime}) \cup( \Lambda^{\prime}\setminus \Lambda) $. Then the number of eigenvalues of $H_{\Lambda}$ inside $[E^*-\beta', E^*+\beta']$ does not exceed $3 L$.
\end{lem}

Now we prove Theorem \ref{mainb}. 
\begin{proof}[Proof of Theorem \ref{mainb}] 
Since $E^*$ is fixed, all the   $l_n, B_n, A_n^{(i)}$ mentioned  in the proof  are associated with  this  $E^*$. Let\begin{equation}\label{1902}
	\delta=	\beta^{\frac{1}{m_0(E^*)}-\frac{\mu}{2}}.
\end{equation}
Let $N$ be  sufficiently large (depending on $\delta$).    We discuss two cases. 

\underline{\textbf{Case 1:} $\beta<h(\delta_0)$.} We choose an $n$ such that \begin{equation}\label{1916}
	\beta\in [h(\delta_n),h(\delta_{n-1})]. 
\end{equation}
 By Remark \ref{gouzao},  $Q_N$ has  an extension $\widetilde{Q_N}$ such that $\widetilde{Q_N}$ is $(\theta,E^*)$-$n$-regular and 
	\begin{equation}\label{568}
		 Q_N \subset\widetilde{Q_N} \subset Q_{N+50M^2l_n}.
	\end{equation}
Since $\widetilde{Q_N}$ is $(\theta,E^*)$-$n$-regular, using  the same argument as in the proof of Claim  \ref{1goc} yields  that    for any $x\in \widetilde{Q_N}\cap S_{n-1}(\theta,E^*,\delta_{n-1})$,  there exists a $y(x)$ such that $\|x-y(x)\|_1\leq l_n^{1/\alpha}$,  $\theta+y(x)\cdot\omega\in  \widetilde{\mathcal{N}}_{n-1}^{(i_{y(x)})}$ and  $(B_n+y(x))\subset \widetilde{Q_N}$. 
Define 
\begin{align}\label{588}
	\mathcal{S}:= \left\{y(x):\  x\in \widetilde{Q_N}\cap S_{n-1}(\theta,E^*,\delta_{n-1}) \right\}\bigcap  S_n(\theta,E^*,\delta)
\end{align}
and 
\begin{equation}\label{5678}
		Q_N':= \widetilde{Q_N}\setminus \bigcup_{y\in \mathcal{S}}\left(A_n^{(i_y)}+y\right). 
\end{equation}
Using the same argument as in the proofs of \eqref{r2} and \eqref{n1}, together with \eqref{123321} with $n+1$ replaced by $n$ and the bound   $|\mathcal{C}_{n-1}^{(i_y)}|\leq m_0(E^*)$, we obtain  
\begin{itemize}
	\item  $Q_N'$ is $(\theta,E^*)$-$n$-regular. 
	\item  For any $x\in Q_N'\cap S_{n-1}(\theta,E^*,\delta_{n-1})$,  there exists a block $B_{n}(y)(=B_{n}+y)$  such that  $x\in \operatorname{Core}(B_{n}(y))\subset B_{n}(y) \subset Q_N'$ and    \begin{equation*}
		\|G_{B_{n}(y)}^{\theta,E^*}\|\leq \delta_{n-1}^{-3}\delta^{-m_0(E^*)}. 
	\end{equation*}  
\end{itemize}
Since $\delta\geq h(\delta_n)$ by \eqref{1902} and \eqref{1916},    it follows from the resolvent identity  as in the proof of  Proposition \ref{1gon} that 
\begin{equation}\label{shengli}
		\|G_{Q_N'}^{\theta,E^*}\|\leq \delta_{n-1}^{-5}\delta^{-m_0(E^*)}\leq(2\beta)^{-1} 
\end{equation} provided $ \beta^{m_0(E^*)\mu/2} \leq \frac{1}{2}g(\beta)^5 $.
It follows from  the uniform distribution of $\{x\cdot\omega\}_{x\in \Z^d}$ by    $\omega\in 	{\rm DC}_{\tau, \gamma}$,    together with  \eqref{588} and the definition of  $S_n(\theta,E^*,\delta )$ (cf. Definition  \ref{swan}) that   $|\mathcal{S}|\lesssim  N^d\delta $  for all sufficiently large $N$. It follows from \eqref{1916},  \eqref{1902},  \eqref{568} and \eqref{5678} that 
\begin{equation}\label{shengli1}
|	Q_N'\triangle  Q_N|\leq C (M^n  N^d\delta+ l_nN^{d-1}) \leq C      M^n  N^d \beta^{\frac{1}{m_0(E^*)}-\frac{\mu}{2}}\leq \frac13 N^d \beta^{\frac{1}{m_0(E^*)}-\mu}
\end{equation}
provided  $|\ln\beta|\leq \beta^{-\mu/2}$ and  $N$ sufficiently large. 
 
It follows from \eqref{shengli}, \eqref{shengli1} and Lemma \ref{HOl} (setting $\beta'= \beta , L= \frac13 N^d \beta^{\frac{1}{m_0(E^*)}-\mu}$) that 
$$
\mathcal{N}_{Q_N}(E^*+\beta  , \theta)- \mathcal{N}_{Q_N}(E^*-\beta  , \theta)\leq  \beta^{\frac{1}{m_0(E^*)}-\mu }$$
provided $\beta<\beta_0(\mu)$ and $N$ sufficiently large.  Letting   $N\to +\infty$ yields  the conclusion. 

 \underline{\textbf{Case 2:} $h(\delta_0)\leq \beta\leq1$.} This case is easier. 
 Define 
 \begin{equation}\label{5880}
 	\mathcal{S}:= Q_N \bigcap  S_0(\theta,E^*,\delta)\ \   \text{ and }\ \ 	Q_N':= \widetilde{Q_N}\setminus \mathcal{S}.  
 \end{equation}
 By \eqref{1902}, we have $\delta\geq \beta\geq h(\delta_0)= g(\varepsilon_0)\gg \varepsilon^{\frac1M}.$
 It follows from 	Condition \ref{con1}  that  
\begin{equation}\label{idsxiajie}
	|v(\theta+x\cdot\omega)-E^*|\geq \widetilde{C}_v^{-1} \delta^{m_0(E^*)}>10d\varepsilon, \quad \forall x\in  Q_N' .
\end{equation}
It follows from \eqref{idsxiajie} and the Neumann series argument that 
\begin{equation}\label{shengli0}
	\|G_{Q_N'}^{\theta,E^*}\|\leq 2\widetilde{C}_v\delta^{-m_0(E^*)}\leq(2\beta)^{-1} 
\end{equation}
provided $ \beta^{m_0(E^*)\mu/2} \leq (4\widetilde{C}_v)^{-1} $.  It follows from  the uniform distribution of $\{x\cdot\omega\}_{x\in \Z^d}$ together with \eqref{1902} and  \eqref{5880} that 
 \begin{equation}\label{shengli10}
 	|	Q_N'\triangle  Q_N|=|\mathcal{S}| \leq 10m_0(E^*)  |Q_N|   \beta^{\frac{1}{m_0(E^*)}-\frac{\mu}{2}} \leq \frac13   |Q_N|  \beta^{\frac{1}{m_0(E^*)}-\mu}
 \end{equation}
 provided  $30m_0(E^*) \leq \beta^{-\mu/2}$ and  $N$ sufficiently large. 
  It follows from \eqref{shengli0}, \eqref{shengli10} and Lemma \ref{HOl}  that 
  $$
  \mathcal{N}_{Q_N}(E^*+\beta  , \theta)- \mathcal{N}_{Q_N}(E^*-\beta  , \theta)\leq  \beta^{\frac{1}{m_0(E^*)}-\mu }$$
  provided $\beta<\beta_0(\widetilde{C}_v , m_0(E^*), \mu)$  and $N$ sufficiently large.  Letting   $N\to +\infty$ yields  the conclusion. 
  
\end{proof}

\begin{rem}
The argument in the proof above originates from    \cite{Bou00}, where	Bourgain studied the regularity of the IDS for AMO by means of this method. 
\end{rem}

\begin{proof}[Proof of Corollary \ref{coor}] 
Using the same argument as in the proofs of Lemma  \ref{genz}, we can prove that there exists some $\varepsilon_0=\varepsilon_0(v_0,u)$ such that  for any $0\leq \varepsilon_2\leq \varepsilon_0$, the analytic function $v=v_0+\varepsilon_2u$ satisfies \eqref{tiao1} and Condition \ref{con1} with constants $C_v$, $\eta$, $\widetilde{C}_v$ and $M=2m$ that are uniform in $\varepsilon_2$. So we can apply Theorem \ref{mainb} to $v$.

\end{proof}

\appendix

\section{Proof of Condition \ref{con1}}
\begin{lem}\label{gen}
	Let $v$ be a non-constant analytic function in   $\T_{2\eta}$ with  $\sup_{\theta\in \T_{\eta}}|v(\theta)|\leq C_v$. Then there exists some integer  $M\in [2,+\infty)$  such that for any $E^*\in \{v(\theta):\  \theta\in \T_{\eta/4} \} $, $v(\theta)-E^*$ has no more than $M
	$ roots (counted with multiplicity) in $\T_{\eta/2}$. We denote  these roots by  $\theta_0^{(1)}(E^*),\cdots,\theta_0^{(m_0)}(E^*)$, where $m_0=m_0(E^*)\in [0,M]$ is the number of roots of $v(\theta)-E^*$ in $\theta\in \T_{\eta/2}$.  Moreover,  there exists some $\widetilde{C}_v\geq 1$ such that for any $E^*\in \{v(\theta):\  \theta\in \T_{\eta/4} \} $, 
	\begin{equation*}
		\widetilde{C}_v^{-1}\prod_{i=1}^{m_0}\|\theta- \theta_0^{(i)}(E^*)\|_\T \leq |v(\theta)-E^*|\leq \widetilde{C}_v\prod_{i=1}^{m_0}\|\theta- \theta_0^{(i)}(E^*)\|_\T, \quad  \forall    \theta\in \T_{\eta/4}.    
	\end{equation*}
\end{lem}
Since Lemma \ref{gen} is equivalent to the following Lemma \ref{genz} via  a  identification  $\theta\mapsto e^{2\pi i \theta}$, we will prove Lemma \ref{genz} instead of Lemma \ref{gen}. 
\begin{lem}\label{genz}
	For  $\eta\geq 0$, we denote 
	$$A_\eta:=\{z\in \C:\  e^{-2\pi\eta} \leq |z|\leq e^{2\pi\eta} \}. $$
	Let $u$ be a non-constant analytic function in   $A_{2\eta}$ with  $\sup_{z\in A_{\eta}}|u(z)|\leq C_u$. Then there exists some integer  $M\in [2,+\infty)$  such that for any $E^*\in \{u(z):\  z\in A_{\eta/4} \} $, $u(z)-E^*$ has no more than $M
	$ roots (counted with  multiplicity) in $A_{\eta/2}$. We denote  these roots by  $z^{(1)}(E^*),\cdots,z^{(m)}(E^*)$, where $m=m(E^*)\in [0,M]$ is the number of roots of $u(z)-E^*$ in $z\in A_{\eta/2}$.  Moreover,  there exists some $\widetilde{C}_u\geq 1$ such that for any $E^*\in \{u(z):\  z\in A_{\eta/4} \} $,  
	\begin{equation*}
		\widetilde{C}_u^{-1}\prod_{i=1}^{m}|z- z^{(i)}(E^*)| \leq |u(z)-E^*|\leq \widetilde{C}_u\prod_{i=1}^{m}|z- z^{(i)}(E^*)|, \quad  \forall    z \in  A_{\eta/4}.  	\end{equation*}
\end{lem}
\begin{proof}	We argue by the contradiction.
	\begin{itemize}
		\item First we prove $m(E^*)$ is bounded for all $E^*\in \{u(z):\  z\in A_{\eta/4} \}$.  
		Suppose the conclusion fails. Then for each $n\in \N$, there exists some $E_n\in \{u(z):\  z\in A_{\eta/4} \}$ such that  $m(E_n)\geq n$.
		Since 	$\{u(z):\  z\in A_{\eta/4} \}$ is compact, there exists a subsequence (still denoted by $E_n$) converging to some point $E_\infty\in \{u(z):\  z\in A_{\eta/4} \}$. 
		Since $u$ is  analytic and   non-constant,  $u(z)-E_\infty$ has finitely many roots in  $A_\eta$. Thus there exists some $\delta\in [\frac34 \eta , \eta ]$ such that $u(z)-E_\infty$ has no root on $\partial A_\delta$. It follows from the  argument principle that 
\begin{equation}\label{1956}
			 |\{\text{roots of $u(z)-E_\infty$ in $A_{\delta}$}\}|=\frac{1}{2 \pi i} \oint_{\partial A_\delta}\frac{ \partial_z  (u(z)-E_\infty)}{u(z)-E_\infty} d z<+\infty. 
\end{equation}
		On the other hand, since $E_n\to E_\infty$,  we have for all  sufficiently large $n$,
		$$	n\leq \frac{1}{2 \pi i} \oint_{\partial A_\delta}\frac{ \partial_z  (u(z)-E_n)}{u(z)-E_n} d z=	\frac{1}{2 \pi i} \oint_{\partial A_\delta}\frac{ \partial_z  (u(z)-E_\infty)}{u(z)-E_\infty} d z,$$
	which contradicts \eqref{1956}. This proves that   $m(E^*)$ is bounded by some $M$.  
		
		Now for any $E^*$,  the function 
		$$\varphi_{E^*}(z):=\frac{u(z)-E^*}{\prod_{i=1}^{m(E^*)}(z- z^{(i)}(E^*))} $$ becomes an analytic function in $A_{\eta}$, which has  no  root  in $A_{\eta/2}$. 
		\item 	Assume that for any $n\in \N$,  there exist some $E_n\in \{u(z):\  z\in A_{\eta/4} \}$ and $z_n\in A_{\eta/4}$ such that  $|\varphi_{E_n}(z_n)|> n$. Since $m(E_n)\leq M$, by the pigeonhole principle,  there exists  some $\delta\in [ \eta/4, \eta/2 ]$ (depending on $E_n$)  such that $$\min_{1\leq i\leq m(E_n) }\min_{z\in \partial A_\delta}|z-z^{(i)}(E_n)|\geq (e^{-\pi\eta /2 }-e^{-\pi \eta})/(10M).$$ By the maximum modulus principle, we have 
		$$n<|\varphi_{E_n}(z_n)|\leq \sup_{z\in \partial A_\delta} |\varphi_{E_n}(z)|\leq\frac{2C_v}{((e^{-\pi\eta /2 }-e^{-\pi \eta})/(10M))^M}.$$
		Letting $n\to +\infty$ yields  a contradiction. 
		\item 	Assume that for any $n\in \N$, there exist some $E_n\in \{u(z):\  z\in A_{\eta/4} \}$ and $z_n\in A_{\eta/4}$ such that  $|\varphi_{E_n}(z_n)|<1/n$. We assume that $E_n\to E_\infty$. Since $u(z)-E_\infty$ has finitely many roots in  $A_\eta $,  there exist some $\delta\in [ \eta/4 , \eta/2 ]$ and $\epsilon>0$ such that 
		$$\min_{z\in \partial A_\delta} |u(z)-E_\infty|>\epsilon.$$
		Thus for all  sufficiently large $n$, we have  
		$$\min_{z\in \partial A_\delta} |u(z)-E_n|>\epsilon/2$$
		hence  			 $$\min_{z\in \partial A_\delta} |\varphi_{E_n}(z)|>\frac{\epsilon/2}{(2e^{\pi  \eta})^M }.  $$ On the other hand, since $\varphi_{E_n}(z)$   has  no  root  in $A_{\eta/2}$,  by the maximum modulus principle together with $z_n\in A_{\eta/4}$ and $\delta\in [ \eta/4 , \eta/2 ]$, for all sufficiently large $n$  we have  
		$$1/n> \varphi_{E_n}(z_n)\geq \min_{z\in \partial A_\delta} |\varphi_{E_n}(z)|>\frac{\epsilon/2}{(2e^{\pi  \eta})^M },  $$ which yields  a contradiction.
	\end{itemize}	
\end{proof}

\section{Proof of Condition \ref{con2}}

	\begin{lem}\label{tran}
	Let $v$ be a non-constant analytic function on $\T$ with $1$ as its shortest period. Then there exist $s\in \N, c>0$  such that 
	$$\max _{0 \leq l \leq s}\left|\partial_\theta^l(v(\theta+\phi)-v(\theta))\right| \geq c \|\phi\|_\T,  \quad \forall \theta , \phi\in \T .$$
\end{lem}
\begin{proof}
	We argue by the contradiction. Suppose the conclusion fails. Then for each $s\in\mathbb{N}$, taking $c=1/s$, there exist $\theta_s,\phi_s\in\mathbb{T}$ such that
	\begin{equation}\label{trf}
		\max_{0\leq l\leq s}\left|\partial^l_{\theta} v(\theta_s+\phi_s)-\partial^l_{\theta} v(\theta_s)\right|<\frac{1}{s}\|\phi_s\|_{\mathbb{T}}.	
	\end{equation}
	Since $\mathbb{T}$ is compact, there exists a subsequence (still denoted by $\{(\theta_s,\phi_s)\}$) converging to some point $(\theta_\infty,\phi_\infty)\in\mathbb{T}\times\mathbb{T}$.  
There are two cases. 
	
	\underline{\textbf{Case 1:} $\|\phi_\infty\|_\T\neq 0$.} Then there exists some  $\delta>0$ such that $\|\phi_s\|_{\mathbb{T}}\geq\delta/2$ for all sufficiently large $s$. For any fixed $l\geq 0$, when $s\geq l$, from \eqref{trf} we obtain
	$$	\left|\partial^l_{\theta} v(\theta_s+\phi_s)-\partial^l_{\theta} v(\theta_s)\right|<\frac{1}{s}\|\phi_s\|_{\mathbb{T}}\to 0 \text{ as } s\to\infty .$$
	By the continuity of $\partial^l_{\theta} v(\theta)$,  letting $s\to\infty$ yields $\partial^l_{\theta} v(\theta_\infty+\phi_\infty)=\partial^l_{\theta} v(\theta_\infty)$ for every $l\geq 0$. Hence the analytic functions $v(\cdot)$ and $v(\cdot+\phi_\infty)$ have all derivatives equal at $\theta_\infty$. By analyticity, $v(\theta+\phi_\infty)=v(\theta)$ for all $\theta\in\mathbb{T}.$ But $\|\phi_\infty\|_\T\neq 0$, contradicting the assumption that $v$ has shortest period $1$.
	
	\underline{\textbf{Case 2:} $\|\phi_\infty\|_\T=0$.}  We may assume $\phi_s\to 0$ and  $\phi_s\neq 0$ (otherwise the left-hand side of \eqref{trf} is non-negative while the right-hand side is zero, which is impossible). For any fixed $l$, by the continuity of $\partial^{l+1}_\theta v(\theta)$,   we have 
	\begin{equation}\label{uni}
		\text{ $\frac{\partial^l_{\theta} v(\theta+\phi_s)-\partial^l_{\theta} v(\theta)}{\phi_s }\to \partial^{l+1}_\theta v(\theta)$ uniformly for $\theta\in \T$ as $\phi_s\to 0$.}
	\end{equation}
	It follows from  \eqref{trf} and \eqref{uni}  that for any $\epsilon>0$ and   all sufficiently large $s>s_0(\epsilon)$, we have  $|\partial^{l+1}_\theta v(\theta_s)|<\epsilon, $ 
	which implies $\partial^{l+1}_\theta v(\theta_\infty)=0. $
	This holds for every $l\geq 0$,  which implies $v(\theta)$ is a  constant by analyticity, contradicting the assumption. 
\end{proof}
\section{Some conclusions on Schur complement}
\begin{lem}\label{Su}
	Let $A$ be the matrix
	$$A=\begin{pmatrix}
		A_1 & A_2^{\text{T}} \\
		A_2 & A_3
	\end{pmatrix},$$
	where $A_3$ is an invertible $n \times n$ matrix, $A_2$ is an $n \times k$ matrix and $A_1$ is a $k \times k$ matrix. Let $S$ be  the Schur complement 
	$$
	S:=A_1-A_2^{\text{T}} A_3^{-1} A_2  .
	$$
	Then we have 
	$$\operatorname{det}A= \operatorname{det}S \cdot \operatorname{det}A_3 $$ 
	and 
	$$	\left\|S^{-1}\right\| \leq\left\|A^{-1}\right\| \leq (1+\|A_2\|)^2\left(1+\left\|A_3^{-1}\right\|\right)^2\left(1+\left\|S^{-1}\right\|\right). $$
	In particular,  $\operatorname{det}A=0$  if and only if $\operatorname{det}S=0$.
\end{lem}

\normalem


\begin{thebibliography}{CSZ24b}


\bibitem[AJ10]{AJ10}
A.~Avila and S.~Jitomirskaya.
\newblock Almost localization and almost reducibility.
\newblock {\em J. Eur. Math. Soc. (JEMS)}, 12(1):93--131, 2010.

\bibitem[Amo09]{Amo09}
S.~Amor.
\newblock H\"{o}lder continuity of the rotation number for quasi-periodic
  co-cycles in {${\rm SL}(2,\Bbb R)$}.
\newblock {\em Comm. Math. Phys.}, 287(2):565--588, 2009.

\bibitem[AYZ17]{AYZ17}
A.~Avila, J.~You, and Q.~Zhou.
\newblock Sharp phase transitions for the almost {M}athieu operator.
\newblock {\em Duke Math. J.}, 166(14):2697--2718, 2017.

\bibitem[BEK90]{rao}
R.~Bhatia, L.~Elsner, and G.~Krause.
\newblock Bounds for the variation of the roots of a polynomial and the
  eigenvalues of a matrix.
\newblock {\em Linear Algebra Appl.}, 142:195--209, 1990.

\bibitem[BG00]{BG00}
J.~Bourgain and M.~Goldstein.
\newblock On nonperturbative localization with quasi-periodic potential.
\newblock {\em Ann. of Math. (2)}, 152(3):835--879, 2000.

\bibitem[BGS02]{BGS02}
J.~Bourgain, M.~Goldstein, and W.~Schlag.
\newblock Anderson localization for {S}chr\"{o}dinger operators on {$\bold
  Z^2$} with quasi-periodic potential.
\newblock {\em Acta Math.}, 188(1):41--86, 2002.

\bibitem[BJ00]{BJ00}
J.~Bourgain and S.~Jitomirskaya.
\newblock Anderson localization for the band model.
\newblock {\em Geometric aspects of functional analysis}, 1745:67--79, 2000.


\bibitem[BJ02]{BJ02}
J.~Bourgain and S.~Jitomirskaya.
\newblock Absolutely continuous spectrum for 1{D} quasiperiodic operators.
\newblock {\em Invent. Math.}, 148(3):453--463, 2002.

\bibitem[BK13]{BK13}
J.~Bourgain and A.~Klein.
\newblock Bounds on the density of states for {S}chr\"{o}dinger operators.
\newblock {\em Invent. Math.}, 194(1):41--72, 2013.


\bibitem[BK19]{BK19}
J.~Bourgain and I.~Kachkovskiy.
\newblock Anderson localization for two interacting quasiperiodic
particles.
\newblock {\em Geom. Funct. Anal.}, 29(1):3--43, 2019.



\bibitem[Bou00]{Bou00}
J.~Bourgain.
\newblock H\"{o}lder regularity of integrated density of states for the almost
  {M}athieu operator in a perturbative regime.
\newblock {\em Lett. Math. Phys.}, 51(2):83--118, 2000.

\bibitem[Bou02]{Bou02}
J.~Bourgain.
\newblock On the spectrum of lattice {S}chr\"{o}dinger operators with
  deterministic potential. {II}.
\newblock {\em J. Anal. Math.}, 88:221--254, 2002.
\newblock Dedicated to the memory of Tom Wolff.




\bibitem[Bou07]{Bou07}
J.~Bourgain.
\newblock Anderson localization for quasi-periodic lattice {S}chr\"{o}dinger
  operators on {$\Bbb Z^d$}, {$d$} arbitrary.
\newblock {\em Geom. Funct. Anal.}, 17(3):682--706, 2007.

\bibitem[CD93]{CD93}
V.~A. Chulaevsky and E.~I. Dinaburg.
\newblock Methods of {KAM}-theory for long-range quasi-periodic operators on
  {${\bf Z}^\nu$}. {P}ure point spectrum.
\newblock {\em Comm. Math. Phys.}, 153(3):559--577, 1993.

\bibitem[CS83a]{CS832}
W.~Craig and B.~Simon.
\newblock Log {H}\"{o}lder continuity of the integrated density of states for
  stochastic {J}acobi matrices.
\newblock {\em Comm. Math. Phys.}, 90(2):207--218, 1983.

\bibitem[CS83b]{CS831}
W.~Craig and B.~Simon.
\newblock Subharmonicity of the {L}yaponov index.
\newblock {\em Duke Math. J.}, 50(2):551--560, 1983.

\bibitem[CSZ23]{CSZ23}
H.~Cao, Y.~Shi, and Z.~Zhang.
\newblock Localization and regularity of the integrated density of states for
  {S}chr\"{o}dinger operators on {$\Bbb Z^d$} with {$C^2$}-cosine like
  quasi-periodic potential.
\newblock {\em Comm. Math. Phys.}, 404(1):495--561, 2023.

\bibitem[CSZ24a]{CSZ24}
H.~Cao, Y.~Shi, and Z.~Zhang.
\newblock On the spectrum of quasi-periodic {S}chr\"{o}dinger operators on
  {$\Bbb Z^d$} with {$C^2$}-cosine type potentials.
\newblock {\em Comm. Math. Phys.}, 405(8):Paper No. 174, 84, 2024.

\bibitem[CSZ24b]{CSZ22}
H.~Cao, Y.~Shi, and Z.~Zhang.
\newblock Quantitative {G}reen's function estimates for lattice quasi-periodic
  {S}chr\"{o}dinger operators.
\newblock {\em Sci. China Math.}, 67(5):1011--1058, 2024.

\bibitem[CSZ25]{CSZ25}
H.~Cao, Y.~Shi, and Z.~Zhang.
\newblock Localization for {L}ipschitz monotone quasi-periodic
  {S}chr\"{o}dinger operators on {$\Bbb Z^d$} via {R}ellich functions analysis.
\newblock {\em Comm. Math. Phys.}, 406(5):Paper No. 102, 63, 2025.

\bibitem[Eli97]{Eli97}
L.~H. Eliasson.
\newblock Discrete one-dimensional quasi-periodic {S}chr\"{o}dinger operators
  with pure point spectrum.
\newblock {\em Acta Math.}, 179(2):153--196, 1997.

\bibitem[FS83]{FS83}
J.~Fr\"{o}hlich and T.~Spencer.
\newblock Absence of diffusion in the {A}nderson tight binding model for large
  disorder or low energy.
\newblock {\em Comm. Math. Phys.}, 88(2):151--184, 1983.

\bibitem[FSW90]{FSW90}
J.~Fr\"{o}hlich, T.~Spencer, and P.~Wittwer.
\newblock Localization for a class of one-dimensional quasi-periodic
  {S}chr\"{o}dinger operators.
\newblock {\em Comm. Math. Phys.}, 132(1):5--25, 1990.

\bibitem[FV25]{FV25}
Y.~Forman and T.~VandenBoom.
\newblock Localization and {C}antor spectrum for quasiperiodic discrete
  {S}chr\"{o}dinger operators with asymmetric, smooth, cosine-like sampling
  functions.
\newblock {\em Mem. Amer. Math. Soc.}, 312(1583):v+86, 2025.

\bibitem[GJ24]{GJ24}
L.~Ge and S.~Jitomirskaya.
\newblock Hidden subcriticality, symplectic structure, and universality of
  sharp arithmetic spectral results for type {I} operators.
\newblock {\em arXiv:2407.08866}, 2024.

\bibitem[GS01]{GS01}
M.~Goldstein and W.~Schlag.
\newblock H\"{o}lder continuity of the integrated density of states for
  quasi-periodic {S}chr\"{o}dinger equations and averages of shifts of
  subharmonic functions.
\newblock {\em Ann. of Math. (2)}, 154(1):155--203, 2001.

\bibitem[GS08]{GS08}
M.~Goldstein and W.~Schlag.
\newblock Fine properties of the integrated density of states and a
  quantitative separation property of the {D}irichlet eigenvalues.
\newblock {\em Geom. Funct. Anal.}, 18(3):755--869, 2008.

\bibitem[GY20]{GY20}
L.~Ge and J.~You.
\newblock Arithmetic version of {A}nderson localization via reducibility.
\newblock {\em Geom. Funct. Anal.}, 30(5):1370--1401, 2020.

\bibitem[GYZ22]{GYZ22}
L.~Ge, J.~You, and X.~Zhao.
\newblock H\"{o}lder regularity of the integrated density of states for
  quasi-periodic long-range operators on {$\ell^2(\Bbb Z^d)$}.
\newblock {\em Comm. Math. Phys.}, 392(2):347--376, 2022.

\bibitem[GYZ23]{GYZ23}
L.~Ge, J.~You, and Q.~Zhou.
\newblock Exponential dynamical localization: criterion and applications.
\newblock {\em Ann. Sci. \'{E}c. Norm. Sup\'{e}r. (4)}, 56(1):91--126, 2023.


\bibitem[Han24]{Han24}
R.~Han.
\newblock Sharp localization on the first supercritical stratum for {L}iouville frequencies.
\newblock {\em arXiv:2405.07810}, 2024.


\bibitem[HS26]{HS26}
R.~Han and W.~Schlag.
\newblock Avila's acceleration via zeros of determinants and applications to
  {S}chr\"{o}dinger cocycles.
\newblock {\em Comm. Pure Appl. Math.}, 79(3):729--761, 2026.

\bibitem[Jit94]{Jit94}
S.~Jitomirskaya.
\newblock Anderson localization for the almost {M}athieu equation: a
  nonperturbative proof.
\newblock {\em Comm. Math. Phys.}, 165(1):49--57, 1994.

\bibitem[Jit99]{Jit99}
S.~Jitomirskaya.
\newblock Metal-insulator transition for the almost {M}athieu operator.
\newblock {\em Ann. of Math. (2)}, 150(3):1159--1175, 1999.

\bibitem[Jit23]{Jit23}
S.~Jitomirskaya.
\newblock One-dimensional quasiperiodic operators: global theory, duality, and
  sharp analysis of small denominators.
\newblock In {\em I{CM}---{I}nternational {C}ongress of {M}athematicians.
  {V}ol. 2. {P}lenary lectures}, pages 1090--1120. EMS Press, Berlin, [2023]
  \copyright 2023.
  
\bibitem[JK16]{JK16}
S.~Jitomirskaya and I.~Kachkovskiy.
\newblock {$L^2$}-reducibility and localization for quasiperiodic operators.
\newblock {\em Math. Res. Lett.}, 23(2):431--444, 2016.

\bibitem[JL18]{JL18}
S.~Jitomirskaya and W.~Liu.
\newblock Universal hierarchical structure of quasiperiodic eigenfunctions.
\newblock {\em Ann. of Math. (2)}, 187(3):721--776, 2018.

\bibitem[JL24]{JL24}
S.~Jitomirskaya and W.~Liu.
\newblock Universal reflective-hierarchical structure of quasiperiodic
  eigenfunctions and sharp spectral transition in phase.
\newblock {\em J. Eur. Math. Soc. (JEMS)}, 26(8):2797--2836, 2024.

\bibitem[JLS20]{JLS20}
S.~Jitomirskaya, W.~Liu, and Y.~Shi.
\newblock Anderson localization for multi-frequency quasi-periodic operators on
  {${\Bbb Z}^D$}.
\newblock {\em Geom. Funct. Anal.}, 30(2):457--481, 2020.

\bibitem[Kle05]{Kle05}
S.~Klein.
\newblock Anderson localization for the discrete one-dimensional quasi-periodic
  {S}chr\"{o}dinger operator with potential defined by a {G}evrey-class
  function.
\newblock {\em J. Funct. Anal.}, 218(2):255--292, 2005.

\bibitem[Liu22]{Liu20}
W.~Liu.
\newblock Quantitative inductive estimates for {G}reen's functions of
  non-self-adjoint matrices.
\newblock {\em Anal. PDE}, 15(8):2061--2108, 2022.

\bibitem[Sch01]{Sch01}
W.~Schlag.
\newblock On the integrated density of states for {S}chr\"{o}dinger operators
  on {$\Bbb Z^2$} with quasi periodic potential.
\newblock {\em Comm. Math. Phys.}, 223(1):47--65, 2001.

\bibitem[Sin87]{Sin87}
Y.~G. Sinai.
\newblock Anderson localization for one-dimensional difference
  {S}chr\"{o}dinger operator with quasiperiodic potential.
\newblock {\em J. Statist. Phys.}, 46(5-6):861--909, 1987.

\bibitem[Sun96]{line}
J.~Sun.
\newblock On the variation of the spectrum of a normal matrix.
\newblock {\em Linear Algebra Appl.}, 246:215--223, 1996.


\bibitem[Sur96]{Sur96}
S.~Surace.
\newblock A finite range operator with a quasi-periodic potential.
\newblock {\em Differential Integral Equations}, 9(1):213--237, 1996.


\bibitem[WYZ26]{WYZ26}
Z.~Wang, J.~You, and Q.~Zhou.
\newblock Anderson localization of long-range quasi-periodic operators via {D}ynamical {R}igidity
\newblock {\em arXiv:2603.10649}, 2026.


\bibitem[You18]{You18}
J.~You.
\newblock Quantitative almost reducibility and its applications.
\newblock In {\em I{CM}---{I}nternational {C}ongress of {M}athematicians.
	{V}ol. 3. {I}nvited lectures}, pages 2113--2135. World Sci. Publ., Hackensack, NJ, [2018]
\copyright 2018.

\end{thebibliography}
\end{document}